%% file: paper.tex
\newcommand{\CLASSINPUTinnersidemargin}{0.85in}
\newcommand{\CLASSINPUTbottomtextmargin}{2.3cm}

\documentclass[journal,final,twocolumn,a4paper,comsoc]{IEEEtran}

\usepackage{epsfig,graphics,psfrag,amsmath}
\usepackage{latexsym,amssymb,amsmath,epsfig,subfig,algorithm,amsthm}
\usepackage{algorithmic}
\usepackage{ctable}
\usepackage{pbox}
\usepackage{color}
\usepackage{url}
\usepackage{scrtime}
\usepackage{bm}
\usepackage{graphicx}
\usepackage{epstopdf}
\usepackage{units}
\usepackage[font=footnotesize]{caption}
\usepackage{cite} 
\usepackage[bottom]{footmisc}
\usepackage{enumitem}

\theoremstyle{plain}
\newtheorem{corollary}{Corollary}
\newtheorem{theorem}{Theorem}
\newtheorem{lemma}{Lemma}

\theoremstyle{definition}

\theoremstyle{remark}
\newtheorem{remark}{Remark}

\newcommand{\V}[1]{\boldsymbol{#1}} 
\newcommand{\E}{\mathbb{E}}
\newcommand{\req}{{\rm req}}
\newcommand{\e}{{\rm e}}
\newcommand{\dd}{{\rm d}}
\newcommand{\quotes}[1]{``#1''}
\newcommand{\definedas}{\overset{\underset{\Delta}{}}{=}}
\makeatletter
\def\ps@IEEEtitlepagestyle{%
  \def\@oddhead{\mycopyrightnotice}%
  \def\@evenhead{}%
}
\def\mycopyrightnotice{%
  {\footnotesize 
  \begin{minipage}{\textwidth}
  \centering
  This article has been accepted for publication in a future issue of this journal, but has not been fully edited. Content may change prior to final publication. Citation information: DOI 10.1109/TCOMM.2019.2951109, IEEE Transactions on Communications\hfill
  \end{minipage}

}
  \gdef\mycopyrightnotice{ }
}
\makeatother
\begin{document}

\title{Conditional Capacity  and Transmit Signal Design for SWIPT Systems with Multiple Nonlinear Energy Harvesting Receivers}

\author{Rania Morsi, Vahid Jamali, Amelie Hagelauer, Derrick Wing Kwan Ng,\\ and Robert Schober  
\thanks{R. Morsi,  V. Jamali, and R. Schober are with the Institute for Digital Communications, and A. Hagelauer is with the Institute for Electronics Engineering in 
Friedrich-Alexander University (FAU), Erlangen, Germany.  D. W. K. Ng is with the University of New South Wales, Australia. (E-mails: rania.morsi@fau.de;
vahid.jamali@fau.de; amelie.hagelauer@fau.de; w.k.ng@unsw.edu.au;  robert.schober@fau.de)}
\thanks{This paper was presented in part in \cite{Morsi_ICC2018} at the IEEE International Conference on Communications (ICC), Kansas City, USA, 2018.}
\thanks{D. W. K. Ng is supported by funding from the UNSW Digital Grid Futures Institute, UNSW, Sydney, under a cross-disciplinary fund scheme and by the Australian Research Council's Discovery Early Career Researcher Award (DE170100137). Robert Schober's work is supported by DFG project SCHO 831/12-1.}
}
\maketitle
\begin{abstract}
\input{abstract}
\end{abstract}

\section{Introduction}
\label{s:introduction}
\input{introduction}
\section{System Model and Preliminaries}
\label{s:system_model}
\input{system_model}
\section{Problem Formulation and Solution}
\label{s:Problem_formulation}
\input{Problem_formulation}
\section{Special Cases and Generalizations}
\label{s:special_cases_and_generalizations}
\input{special_cases_and_generalizations}

\section{Numerical Results}
\label{s:numerical_results}
\input{numerical_results}

\section{Conclusion}
\label{s:conclusion}
\input{conclusion}
\appendices
\input{proof_unique_optimal_distribution}
\input{proof_C_necessary_sufficient_condition1}
\input{proof_C_necessary_sufficient_condition2}
\input{proof_discrete_optimal_distribution}
\input{proof_Only_one_EH_Rx_effective}
\input{proof_max_WPT}
\input{proof_indep_opt_amplitude_phase_UD}
\input{proof_C_necessary_sufficient_condition2_complex}
\bibliographystyle{IEEEtran}
\bibliography{literature}
\input{Biographies}
\end{document}

%% file: abstract.tex
In this paper, we study information-theoretic limits for simultaneous wireless information and power transfer (SWIPT) systems employing practical nonlinear radio frequency (RF) energy harvesting (EH) receivers (Rxs). In particular, we consider a SWIPT system with one transmitter that broadcasts a common signal to an information decoding (ID) Rx and multiple EH Rxs. Owing to the nonlinearity of the EH Rxs' circuitry, the efficiency of wireless power transfer depends on the waveform of the transmitted signal. We aim to answer the following fundamental question: \emph{What is the optimal input distribution of the transmit signal waveform that maximizes the information transfer rate at the ID Rx conditioned on individual minimum  required direct-current (DC) powers to be harvested at the EH Rxs?} Specifically, we study the conditional capacity problem of a SWIPT system impaired by  additive white Gaussian noise subject to average-power (AP) and peak-power (PP) constraints at the transmitter and nonlinear EH constraints at the EH Rxs. {\color{black} To this end, we develop a novel nonlinear EH model that captures  the saturation of the harvested DC power by taking into account not only the forward current of the rectifying diode but also the reverse breakdown current. Then, we derive a novel semi-closed-form expression for the harvested DC power, which simplifies to closed form for low input RF powers. The derived analytical expressions are shown to closely match circuit simulation results.} We solve the conditional capacity problem for real- and complex-valued signalling and prove that the optimal input distribution that maximizes the rate-energy (R-E) region is unique and discrete with a finite number of mass points. {\color{black} Furthermore, we show that, for the considered nonlinear EH model and a given AP constraint, the boundary of the R-E region saturates for high PP constraints due to the saturation of the harvested DC power for high input RF powers.}  In addition, we devise a suboptimal input distribution whose R-E tradeoff performance is close to optimal. All theoretical findings are verified by numerical evaluations. 

%% file: introduction.tex
In addition to their capability to convey information, radio frequency (RF) signals can transfer energy for wirelessly charging low-power devices. This property of RF signals has attracted  significant attention to  the study of simultaneous wireless information and power transfer (SWIPT) systems 
\cite{Varshney2008,Shannon_meets_tesla_Grover2010,Rate_energy_MIMO_RuiZhang2013,Schober_SWPT_review_2015,
Kwan_Robert_Book,SWIPT_relaying_survey_2018,
RE_tradeoff_nonlinear_EH_Kim2018,SWIPT_Clerckx_OFDM_Multisine2016,SWIPT_Clerckx_Single_Carrier2017}. {\color{black} In \cite{Varshney2008}, the author defined a rate-energy (R-E) function that characterizes the tradeoff between wireless information transfer (WIT) and wireless power transfer (WPT). Such a tradeoff exits as long as the optimal transmit waveform that maximizes the rate of information transfer is different from the one that maximizes the amount of harvested energy.} For example, in \cite{Shannon_meets_tesla_Grover2010}, the R-E tradeoff for frequency selective channels with additive white Gaussian noise (AWGN) is characterized, where water-filling power allocation is shown to be optimal for WIT, but allocating all the power to a single sinusoid is optimal for WPT. In  \cite{Rate_energy_MIMO_RuiZhang2013}, it is shown that for a multi-antenna broadcast channel, spatial multiplexing is optimal for  WIT, whereas energy beamforming is optimal for WPT. The aforementioned works lay the foundation for SWIPT research, but they are based on an overly simplistic linear energy harvesting (EH) model for WPT. This model assumes that the harvested direct-current (DC) power depends only on the \emph{average} power of the input RF signal and that this dependence is \emph{linear}  for all possible input RF powers.

In practice, however, the RF EH circuits of WPT systems  have a nonlinear  input-output characteristic {\color{black}\cite{JSAC_Schober_Clreckx_Poor2019,RF_EH_networks_survey_2015,Letter_non_linear,Theoretical_Analysis_Rectifiers_2014,Georgiadis_WPT_book_2016,
Optimum_behaviour_Georgiadis2013,Waveform_optimization_SPAWC_Rui_Zhang_2017,Waveform_design_WPT_Clerckx_2016,Polozec1994}}. In particular, EH circuits include a rectenna, i.e., an antenna followed by a rectifier. The rectifier typically contains diodes followed by a capacitor-based low-pass filter (LPF) to convert the received RF signal into a DC signal.  {\color{black}For high incident RF powers, rectifying diodes exhibit the reverse breakdown phenomenon, where a significant amount of reverse current flows through the diode causing the output DC power to saturate and leading to a reduced RF-to-DC conversion efficiency \cite{Georgiadis_WPT_book_2016,Theoretical_Analysis_Rectifiers_2014}.
 In \cite{Letter_non_linear}, the nonlinear RF-to-DC input-output characteristic of a rectenna is modelled by a three-parameter sigmoidal function, where curve fitting is performed to determine the parameters for a given rectenna circuit and a given excitation signal.}

{\color{black} In this paper, we adopt the same rectifier circuit as was considered in\cite{Waveform_design_WPT_Clerckx_2016,Waveform_optimization_SPAWC_Rui_Zhang_2017,Polozec1994}, namely a series single-diode rectifier. 
In \cite{Waveform_optimization_SPAWC_Rui_Zhang_2017}, a monotonically increasing function of the output DC power is derived in terms of an integral function of the input RF signal. In \cite{Waveform_design_WPT_Clerckx_2016}, a fourth-order Taylor series approximation of the expression in \cite{Waveform_optimization_SPAWC_Rui_Zhang_2017} is analyzed for a multisine excitation signal. 
Furthermore, the authors of \cite{Polozec1994} obtained a semi-closed-form expression for the output DC voltage assuming a sinusoidal input RF signal. 
However, the analysis in the aforementioned works  \cite{Polozec1994,Waveform_optimization_SPAWC_Rui_Zhang_2017,Waveform_design_WPT_Clerckx_2016} took into account only the forward-bias current-voltage (I-V) characteristic of the rectifying diode, as described by Shockley's diode equation \cite{Shockley1949}. This model ignores the reverse-breakdown behaviour of the diode and therefore does not capture the saturation  of the output DC power for high input RF powers. In contrast, in this paper, we take into account both the forward and the reverse breakdown I-V characteristic of the rectifying diode and obtain a novel semi-closed-form expression for the output DC power assuming a sinusoidal input excitation signal. Moreover, in the low-input power regime, we obtain the output DC power in  closed form.  A comparison with circuit simulations confirms the accuracy of the derived analytical expressions.} 

Owing to the rectifier's nonlinearity, the RF-to-DC conversion efficiency depends not only on the strength of the input RF signal, but also on its waveform  {\color{black}\cite{Optimum_behaviour_Georgiadis2013,Georgiadis_WPT_book_2016,Waveform_design_WPT_Clerckx_2016,Waveform_optimization_SPAWC_Rui_Zhang_2017}}. 
For example, experiments have shown that signals with high peak-to-average power ratio (PAPR), such as multisine signals, yield higher harvested DC powers for a given average incident RF power compared to constant-envelope signals  \cite{Optimum_behaviour_Georgiadis2013}. {\color{black}This  is because, for low average power levels, high PAPR signals are more likely to exceed the turn-on voltage of the diode  \cite{Georgiadis_WPT_book_2016}. Moreover, the peaks of a high PAPR signal can charge the capacitor to a high voltage level, and if the output LPF has a large time constant, the capacitor can maintain the charged voltage until the next signal peak, see e.g. \cite[Figure 9]{Optimum_behaviour_Georgiadis2013}. Thus, the nonlinearity of EH circuits motivates the optimization of the transmit signal for maximization of the amount of harvested energy.}

While the goal of waveform design for a pure WPT system is to maximize the harvested energy only, for a SWIPT system, the waveform design goal is to maximize both the information transfer rate and the harvested energy, i.e., to optimize the R-E tradeoff. {\color{black}In \cite{RE_tradeoff_nonlinear_EH_Kim2018}, the R-E tradeoff  of different receiver (Rx) architectures is studied for SWIPT systems with a nonlinear EH model.} In \cite{SWIPT_Clerckx_OFDM_Multisine2016}, the authors consider the superposition of deterministic and modulated multisine waveforms and optimize the amplitudes  and phases of all frequency tones to maximize the R-E region,  {\color{black}i.e., the region of all achievable R-E pairs. As an alternative to using multisine signals, the desired high PAPR characteristic of WPT signals can also be achieved by modulating the amplitude of a single-sine signal. Thereby, the amplitude modulation can be simultaneously used to transmit  information. From the WPT perspective, the optimal distribution of the transmit amplitude is expected to have a high PAPR. On the other hand, from a WIT perspective, the optimal distribution of the transmit amplitude  is known to be Gaussian for an average power limited AWGN channel \cite{Shannon}. Hence,  the aim of this paper is to answer the following fundamental question. \quotes{\emph{For a single-carrier SWIPT system with a nonlinear EH circuit, what is the optimal input distribution of the transmit waveform that maximizes the R-E region?}} }
  
 First steps towards answering this question are made in \cite{SWIPT_Clerckx_Single_Carrier2017}, where input distributions  that are fully characterized by their first- and second-order statistics are considered. It is shown that, in this case, the optimal input distribution is the zero-mean complex Gaussian distribution with asymmetric power allocation to the real and  imaginary parts. {\color{black}However, in general, higher-order statistics may be required to characterize the optimal input distribution that maximizes the R-E region, in which case the waveforms reported in  \cite{SWIPT_Clerckx_Single_Carrier2017} are no longer optimal.}

In this paper, we aim to answer the question above without {\color{black}limiting ourselves to input distributions that are fully characterized by their first- and second-order statistics.} In particular, we consider a SWIPT system, where a single-carrier signal is transmitted over AWGN channels to an information decoding (ID) Rx and simultaneously to multiple EH Rxs. Our objective is to find the optimal distribution of the transmit signal that maximizes the information rate at the ID Rx under individual minimum  harvested power constraints at the nonlinear EH Rxs.  {\color{black} To specify the EH constraints, we employ the harvested DC power function derived for our newly developed nonlinear EH saturation model.} In addition, we  impose average-power (AP) and peak-power (PP) constraints at the power transmitter\footnote{{\color{black} This problem is of practical interest in e.g. sensor networks where one sensor needs to update its software at the highest possible rate while the other sensors in its vicinity want to wirelessly charge their batteries.}}. We note that, for a linear EH model, the considered problem is trivial since, in this case, the harvested DC power depends only on the average input RF power, which renders the input distribution of the transmit symbols irrelevant for WPT. In particular, in addition to the commonly adopted maximum AP constraint, the linear EH constraint imposes a minimum AP  constraint, for which the problem is either infeasible or the EH constraint is inactive. In the latter case, the solutions for maximum WIT in \cite{Shannon,SMITH19712,QAGC_Shamai_Bar_David_1995} are optimal. On the other hand, with a nonlinear EH model, this problem is non-trivial and has been first studied in our preliminary work in 
 \cite{Morsi_ICC2018}, for one EH Rx\footnote{The work in  \cite{Morsi_ICC2018} was  first published in Nov. 2017, see https://arxiv.org/abs/1711.01082.}. {\color{black}Subsequently, a similar problem has been independently studied in \cite{Bruno_Capacity} using the nonlinear EH model from \cite{Waveform_design_WPT_Clerckx_2016}, which does not model the saturation of the output DC power but adopts a $4^{\text{th}}$-order truncated Taylor series approximation of the diode's forward current equation.}

The main contributions of this paper can be summarized as follows:

\begin{itemize}
\item We study the conditional capacity of a SWIPT system with one ID Rx and multiple nonlinear EH Rxs. In particular, we maximize the information rate at the ID Rx under AP and PP constraints at the transmitter and nonlinear EH constraints at the EH Rxs. Accordingly, we obtain the  R-E region, which specifies all combinations of achievable rates at the ID Rx and  jointly feasible harvested powers at the EH Rxs. The boundary of this R-E region is referred to as the R-E tradeoff curve.
\item We obtain necessary and sufficient conditions for the optimal input distribution and prove that it is unique and discrete with a finite number of mass points. The discreteness and finiteness of optimal input distributions for other channels have been  reported  in  \cite{SMITH19712,QAGC_Shamai_Bar_David_1995,capacity_Rayleigh_fading_Abou_Faycal2001,Hermite_bases_Abou_Faycal_2012,Nikola_FD}.
\item Different from our preliminary work in \cite{Morsi_ICC2018}, in this paper,  we additionally consider the following:
\begin{itemize}[leftmargin=*]
\item {\color{black}In order to accurately model the harvested DC power, we take into account not only the forward but also the reverse breakdown I-V characteristic of the rectifying diode. Consequently, the proposed nonlinear EH model captures the saturation behaviour of the  output DC power at high input RF powers. Moreover, our model includes the voltage multiplication effect of the matching network that maximizes the power transfer from the antenna to the rectifier. Accordingly, we obtain a semi-closed-form expression of the output DC power. In addition, in the low input RF power regime, the forward I-V characteristic of the diode is dominant and the harvested DC power is obtained in closed form. The derived analytical expressions are verified with circuit simulations and exploited to formulate the EH constraints for the R-E region maximization problem.} 
\item We consider multiple EH Rxs with individual minimum EH constraints and prove that, if none of the EH Rxs operates in saturation, at most one of these EH constraints is active. In particular, the active EH constraint is the one which, when all other EH constraints are removed, results in the smallest achievable rate at the ID Rx. Hence, the R-E tradeoff curve for this problem is completely characterized by the individual R-E curves obtained for each individual EH Rx separately. 
\item We  extend the problem to complex-valued transmission and show that the optimal input distribution is characterized by a discrete and finite amplitude set with an independent uniformly distributed phase. This result is in line with the results for  the same problem without EH constraints in \cite{QAGC_Shamai_Bar_David_1995}.
\item We solve the problem of maximum \emph{WPT} for one EH Rx under AP and PP constraints. We obtain the optimal input distribution and the maximum harvested DC power at the EH Rx in closed form. We show that on-off transmission is optimal for maximum WPT. 
Based on this insight, for the SWIPT system, we propose a suboptimal distribution  which superimposes the optimal distributions for maximum WPT and maximum WIT. We show that the R-E tradeoff obtained with the suboptimal distribution closely approaches that of the optimal one.
\item {\color{black} We show that, owing to the saturation behaviour of the harvested DC power, the optimal solution for maximum WPT and the boundary of the R-E region saturate for high PP constraints.}
\end{itemize}
\end{itemize}

The remainder of the paper is organized as follows. In Section \ref{s:system_model}, we present the  system model and the nonlinear EH circuit model. In Section \ref{s:Problem_formulation}, we  formulate the R-E problem and reveal several properties of the optimal input distribution. In Section \ref{s:special_cases_and_generalizations}, we study the extreme problems of maximum WIT and maximum WPT and propose a suboptimal input distribution for the SWIPT  problem. In Section \ref{s:numerical_results}, we provide numerical results for the considered problems. Finally,  Section \ref{s:conclusion} concludes the paper.

\textit{Notations:} We use boldface letters to denote random variables and the corresponding lightface letters to denote their realizations.  $j$ is the imaginary unit and $\Re\{\cdot\}$ denotes the real part of a complex number. $\E_{F}[V(\V{x})]\!=\!\int V(x)\dd F(x)$ is the statistical average of $V(\V{x})$ given that random variable $\V{x}$ has distribution function $F(x)$. Moreover, $\sim$ stands for  ``is distributed as" and $\definedas$ means  ``is defined as". $\mathcal{N}(0,\sigma^2)$ and $\mathcal{CN}(0,\sigma^2)$ represent real- and complex-valued Gaussian distributions with zero mean and variance $\sigma^2$. 

%% file: system_model.tex
\subsection{System Model}%
We consider a single-antenna SWIPT system, where a transmitter broadcasts a common single-carrier signal to an ID Rx and $L$ randomly deployed EH Rxs,  as shown in Fig. \ref{fig:system_model}. In particular, we consider a time-slotted system with time slot duration\footnote{In this paper, we assume a unit-length time slot, i.e., $T=1$. Hence,  we use the terms power and energy interchangeably.} $T$ . The transmitter emits a real-valued\footnote{As is customary for capacity analysis, see e.g. \cite{SMITH19712,Hermite_bases_Abou_Faycal_2012,Nikola_FD}, as a first step,  we  assume real-valued channel inputs and outputs.
The generalization to a complex-valued signal model is provided in Section \ref{s:complex_signaling}.} baseband information-bearing pulse-amplitude modulated signal $x(t)\!=\!\sum_{k=-\infty}^{\infty}x[k]g(t\!-\!kT)$, where $g(t)$ is the transmit pulse waveform and $x[k]$ is the information-bearing symbol in time slot $k$, which is a realization of an independent and identically distributed (i.i.d.) real-valued  random variable $\V{x}\!\in\! \mathbb{R}$ having cumulative distribution function $F$. The channel gains for  the ID and EH Rxs are denoted  by $h_{\rm I}\!\in\!\mathbb{R}$ and $h_{{\rm E}_l}\!\in\!\mathbb{R}$, respectively, and are assumed to be fixed over all time slots, where $l\!\in\!\mathcal{L}\definedas\{1,\ldots,L\}$. To obtain an upper bound on the performance of the SWIPT system, all channel gains are assumed to be perfectly known at the transmitter\footnote{\color{black}In practice, assuming a time-division duplex (TDD) system, for channel acquisition, the ID and EH nodes may transmit pilot signals to the transmitting node, which then  estimates the uplink channel gains and exploits the uplink-downlink channel reciprocity to obtain estimates for the downlink channel gains.} and the information channel gain is known at the ID Rx. The baseband model of the received signal at the ID Rx is given by $y_{\rm I}(t)\!=\!x(t)h_{\rm I}\!+\!n(t)$, where $n(t)$ is real-valued AWGN with average power $\sigma_n^2$. At the EH Rxs, the additive noise is ignored since its contribution to the harvested  power is negligible.  Hence, at the $l^{\text{th}}$ EH Rx, the received signals in the baseband and the RF domains are $y_{{\rm E}_l}(t)\!\!=\!\!x(t)h_{{\rm E}_l}$ and $y_{{\rm E}_l}^{\rm RF}(t)\!\!=\!\!\sqrt{2}\Re\{y_{{\rm E}_l}(t)\e^{j 2\pi f_{\rm c} t}\}$, respectively, where $f_{\rm c}$ is the carrier frequency and {\color{black}$w_{\rm c}\definedas 2\pi f_{\rm c}$ is the corresponding angular frequency.} Assuming a rectangular pulse $g(t)$ with unit amplitude and duration $T$, in time slot $k$, i.e., $kT\!\!-\!\!T/2\!\!<\!\!t\!\!\leq\!\!kT\!\!+\!\!T/2$, the baseband transmit signal is constant and given by $x(t)\!=\!\sum_{k=-\infty}^{\infty}x[k]g(t\!-\!kT)\!=\!x[k]$. Assuming all EH Rxs employ  identical EH circuits, we focus on modelling one EH Rx and drop index $l$, for convenience. Thus, the received RF signal at the EH Rx is
\begin{equation}
y_{\rm E}^{\rm RF}(t)\!=\!\sqrt{2} x[k] h_{\rm E}\cos(2\pi f_{\rm c} t), \quad\!\! kT\!-\!T/2\!<\!t\!\leq\!kT\!+\!T/2.
\label{eq:recieved_signal_EH}
\end{equation}

\begin{figure}[!t]
\centering\hspace{-1cm}
\includegraphics[width=0.35\textwidth]{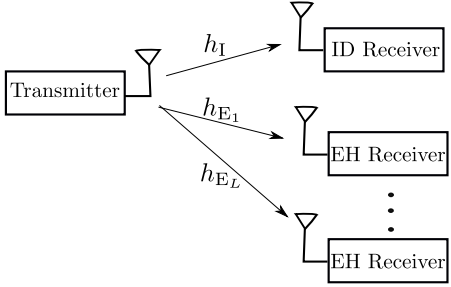}
\caption{SWIPT system with one ID Rx and $L$ separate EH Rxs.}
\label{fig:system_model}
\end{figure}

\begin{figure*}[!t]
\centering
\includegraphics[width=1\textwidth]{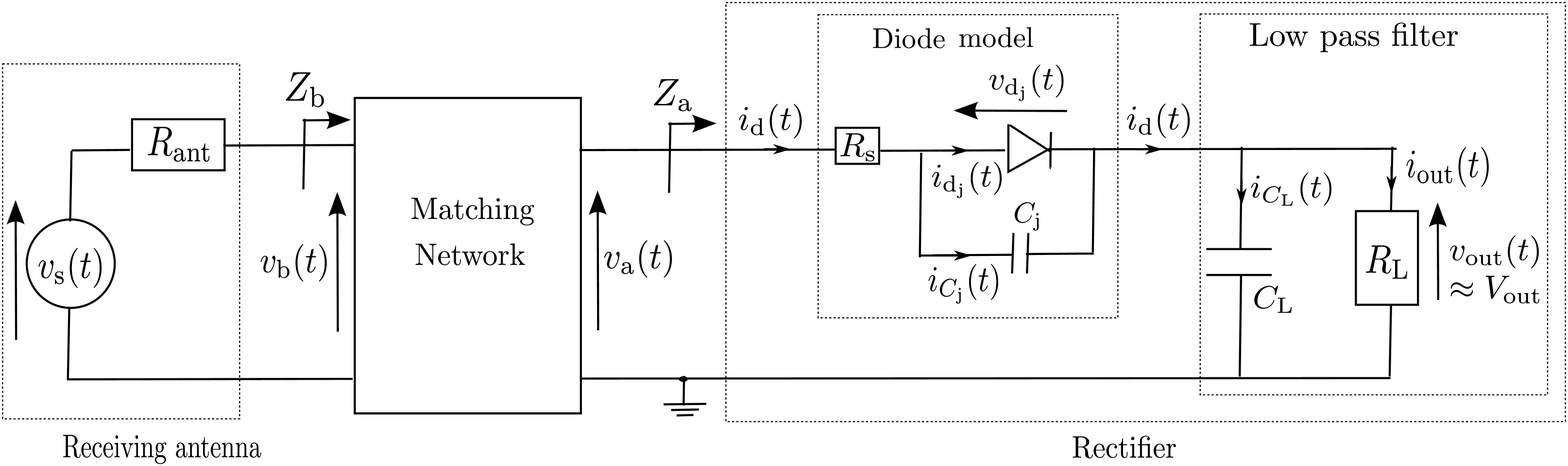}
\caption{{\color{black}Nonlinear rectenna circuit model}.}
\label{fig:rectenna_model}
\end{figure*}

\subsection{Rectenna Nonlinear Circuit Model}
\label{ss:rectenna_model}  
In this section, we derive a novel expression for the harvested DC power at the EH Rx averaged over the symbol duration $T$ in which symbol $x$ is transmitted. As shown in Fig.~\ref{fig:rectenna_model}, the EH Rx includes a rectenna, which consists of an antenna and a rectifier. The antenna is commonly modelled by a {\color{black}Thevenin equivalent voltage source $v_{\rm s}(t)$ in series with an impedance $R_{\rm ant}$\cite{Georgiadis_WPT_book_2016,Waveform_design_WPT_Clerckx_2016,Waveform_optimization_SPAWC_Rui_Zhang_2017,Polozec1994}.} The rectifier converts the received RF signal to a DC signal across a load resistance $R_{\rm L}$. {\color{black}In order to ensure maximum power transfer, a matching network is needed to match the rectifier's input impedance  $Z_{\rm a}$  to the antenna impedance $R_{\rm ant}$, cf. Fig. \ref{fig:rectenna_model}. Since the received RF signal $y_{\rm E}^{\rm RF}(t)$ given in (\ref{eq:recieved_signal_EH}) is sinusoidal, it follows that signals $v_{\rm s}(t)$, $v_{\rm b}(t)$, and $v_{\rm a}(t)$ in Fig. \ref{fig:rectenna_model} are also sinusoidal. In the following, we use the notation $\hat{v}$ to denote the peak amplitude of sinusoidal signal $v(t)$, i.e., $v(t)\!=\!\hat{v}\cos(2\pi f_{\rm c} t+\phi_v)$, where $v\!\in\!\{y_{\rm E}^{\rm RF},v_{\rm s},\,v_{\rm b},\,v_{\rm a}\}$ and $\phi_v\!\in\![-\pi,\pi]$.   We adopt the  rectifier circuit used in  \cite{Waveform_design_WPT_Clerckx_2016,Waveform_optimization_SPAWC_Rui_Zhang_2017,Polozec1994}, which consists of a single series diode 
followed by a capacitor-based LPF with capacitance $C_{\rm L}$. Unlike in \cite{Waveform_optimization_SPAWC_Rui_Zhang_2017} and  \cite{Waveform_design_WPT_Clerckx_2016}, our model includes the diode's series resistance $R_{\rm s}$ and junction capacitance $C_{\rm j}$, see Fig.~\ref{fig:rectenna_model}, \cite[Fig. 1]{Polozec1994}, \cite[Fig. 5.2]{Georgiadis_WPT_book_2016}.
 {\color{black} The main differences between the rectifier model developed in this paper and the models in  \cite{Waveform_design_WPT_Clerckx_2016,Waveform_optimization_SPAWC_Rui_Zhang_2017,Polozec1994} are:
\subsubsection{Saturation Circuit Model}
\label{sss:Saturation_Circuit_Model}
First, we derive the harvested DC power at the rectifier's output in terms of the peak voltage $\hat{v}_a$ at the rectifier's input. To this end, we consider not only the forward-bias mode of the rectifying diode but also the reverse-bias breakdown mode. A typical I-V characteristic of a rectifying diode is shown in \cite[Fig. 6.5]{Georgiadis_WPT_book_2016}. In particular, when the amplitude of the voltage signal across the diode junction $v_{\rm d_j}(t)$ reaches the diode breakdown voltage $B_{\rm v}$, a significant amount of reverse current will pass through the diode for negative  values  of the input signal. This leads to a reduction of the average current in the diode, the saturation of the output DC power, and a degradation of the RF-to-DC power conversion efficiency as the input signal power increases \cite{Georgiadis_WPT_book_2016}. In order to account for this nonlinear behaviour of the diode current, the total current in the diode junction $i_{\rm d_j}(t)$ is modelled as  the sum of the diode forward current $i_{\rm F}(t)$ and the reverse current $i_{\rm R}(t)$, i.e., \cite[Eqs. (6.1)-(6.3)]{Georgiadis_WPT_book_2016}\footnote{\color{black}In \cite[Eq. (6.2)]{Georgiadis_WPT_book_2016}, the last term $I_{B_{\rm v}}\e^{-\frac{{B_{\rm v}}}{\eta V_{\rm T}}}$ in (\ref{eq:total_junction_current}) is set to zero, which is valid for typical rectifying diodes, see e.g.\cite[Table 5.4]{Georgiadis_WPT_book_2016}.} 
\begin{equation}
\begin{aligned}
i_{\rm d_j}(t)&\!=\!i_{\rm F}(t)+i_{\rm R}(t)\\
&\!=\!I_{\rm s}\big(\e^{\frac{v_{\rm d_j}(t)}{\eta V_{\rm T}}}-1\big)-I_{B_{\rm v}}\e^{-\frac{{B_{\rm v}}}{\eta V_{\rm T}}}\big(\e^{-\frac{v_{\rm d_j}(t)}{\eta V_{\rm T}}}-1\big),
\end{aligned}
\label{eq:total_junction_current}
\end{equation}
where $I_{\rm s}$ is the diode's reverse bias saturation current, $\eta$ is the diode ideality factor, which typically lies between 1 and 2, $V_{\rm T}=KT_K/q$ is the thermal voltage, where $K$ is Boltzmann's constant, $q$ is the electron charge, and $T_K$ is the junction temperature in Kelvin. The reverse breakdown current is characterized by $I_{B_{\rm v}}$ and ${B_{\rm v}}$, which represent the breakdown saturation current and  the reverse breakdown voltage, respectively  \cite{Georgiadis_WPT_book_2016}. Applying Kirchoff's current law to the rectifier in Fig. \ref{fig:rectenna_model}, we obtain%
\begin{gather}
\begin{aligned}
i_{\rm d}(t)\!&=\!i_{\rm d_j}(t)+i_{C_{\rm j}}(t)=i_{\rm C_{\rm L}}(t)+i_{\rm out}(t)\\ 
&=\!I_{\rm s}\big(\e^{\frac{v_{\rm d_j}(t)}{\eta V_{\rm T}}}\!\!-\!1\big)\!-\!I_{B_{\rm v}}\e^{\!-\!\frac{{B_{\rm v}}}{\eta V_{\rm T}}}\big(\e^{\!-\!\frac{v_{\rm d_j}(t)}{\eta V_{\rm T}}}\!\!-\!1\big)\!+\!C_{\rm j}\frac{\dd v_{\rm d_j}\!(t)}{\dd t}\\
&=\!C_{\rm L}\frac{\dd v_{\rm out}(t)}{\dd t}+\frac{v_{\rm out}(t)}{R_{\rm L}}.
\end{aligned}
\raisetag{12pt}
\label{eq:KCL_rectifier}
\end{gather}
Due to the nonlinearity of the diode detector circuit, the voltage signals $v_{\rm out}(t)$ and $v_{\rm d_j}(t)$ in (\ref{eq:KCL_rectifier}) contain in general DC and harmonic components \cite{Georgiadis_WPT_book_2016,Polozec2015_Input_Impedance}. Hence, in the most general sense, we can write $v_{\rm z}(t)=\sum_{n=0}^\infty \hat{v}^{(n)}_{\rm z}\cos(2\pi nf_c t+\phi^{(n)}_{z})$, where $z\in\{{\rm d_j},{\rm out}\}$ and the superscript $(n)$ denotes the $n^{\rm th}$ harmonic component. Thus, $\frac{1}{T}\int_T v_{\rm z}(t) \dd t=\hat{v}_z^{(0)}$ and $\frac{1}{T}\int_T \frac{\dd v_{\rm z}(t)}{\dd t}\dd t=0$. Therefore, integrating both sides of (\ref{eq:KCL_rectifier}) over one symbol duration $T$ results in 
\begin{IEEEeqnarray}{ll}
\!\!\!\!\!\!\frac{1}{T}\!\!\int\limits_T \!\!\left[I_{\rm s}\big(\!\e^{\frac{v_{\rm d_j}(t)}{\eta V_{\rm T}}}\!\!-\!1\big)\!-\!I_{B_{\rm v}}\e^{\!-\!\frac{{B_{\rm v}}}{\eta V_{\rm T}}}\!\big(\!\e^{\!-\!\frac{v_{\rm d_j}(t)}{\eta V_{\rm T}}}\!\!-\!1\big)\!\right]\!\dd  t\!=\!\frac{V_{\rm out}}{R_{\rm L}}\!,
\label{eq:1st_Ritz_condition}
\end{IEEEeqnarray}
where we define $ V_{\rm out}\definedas\hat{v}^{(0)}_{\rm out}$. Assuming the rectifier's time constant $R_{\rm L}C_{\rm L}$ is much larger than the period $1/f_{\rm c}$ of the sinusoidal RF signal, the ripples in the output voltage will be negligible \cite{Georgiadis_WPT_book_2016}. In this case, at steady state, the output voltage can be assumed constant (DC), i.e., $v_{\rm out}(t)= V_{\rm out}$ and $i_{\rm d}(t)=C_{\rm L}\frac{\dd v_{\rm out}(t)}{\dd t}+\frac{v_{\rm out}(t)}{R_{\rm L}}=\frac{V_{\rm out}}{R_L}$. Hence, the junction voltage in (\ref{eq:1st_Ritz_condition})  can be written as  $v_{\rm d_j}(t)=v_{\rm a}(t)\!-\!i_{\rm d}(t)R_{\rm s}\!-\!v_{\rm out}(t)=v_{\rm a}(t)\!-\!V_{\rm out}\left(1+\frac{R_{\rm s}}{R_{\rm L}}\right)$ and (\ref{eq:1st_Ritz_condition}) reduces to
\begin{IEEEeqnarray}{ll}
I_{\rm s}\left[\e^{\frac{-V_{\rm out}(\beta)}{\eta V_{\rm T}}\left(1+\frac{R_{\rm s}}{R_{\rm L}}\right)}I_0\left(\beta\right)-1\right]\notag\\
-I_{B_{\rm v}}\e^{-\!\frac{{B_{\rm v}}}{\eta V_{\rm T}}}\left[\e^{\frac{V_{\rm out}(\beta)}{\eta V_{\rm T}}\left(1+\frac{R_{\rm s}}{R_{\rm L}}\right)}I_0\!\left(\beta\right)\!-\!1\right]\!=\!\frac{V_{\rm out}(\beta)}{R_{\rm L}}\!,
\label{eq:Vout_intermsof_va}
\end{IEEEeqnarray}
where $\beta\!\definedas\!\frac{\hat{v}_{\rm a}}{\eta V_{\rm T}}$ and the notation $V_{\rm out}(\beta)$ is used to explicitly indicate the dependence of $V_{\rm out}$ on $\beta$. To arrive at (\ref{eq:Vout_intermsof_va}), we assumed $f_{\rm c}\!=\!m/T$, with integer $m$, in order to use $I_0(\beta)=\frac{1}{T}\int_{T} \e^{\beta\cos(2\pi \frac{m}{T} t+\phi_{v_{\rm a}})}\dd t$, where $I_0(\cdot)$ is the modified Bessel function of the first kind and order zero.
Given the amplitude of the voltage signal at the rectifier's input, $\hat{v}_{\rm a}$, and therefore $\beta$,  (\ref{eq:Vout_intermsof_va})  is a semi-closed-form expression for the DC voltage $V_{\rm out}(\beta)$, which can be solved using e.g. Newton's method. Then, the harvested DC power is
\begin{equation}
P_{\rm out}(\beta)\Big|_{\rm exact}\!\!=\!\!\frac{V_{\rm out}^2(\beta)}{R_{\rm L}}, \text{  } V_{\rm out}(\beta)
\text{ is the solution of (\ref{eq:Vout_intermsof_va}).}
\label{eq:Pout_Sat_exact}
\end{equation}

\begin{remark}
We note that, for low input RF powers, the amplitude of the diode junction voltage $v_{\rm d_j}(t)$ is small and does not reach the breakdown voltage of the diode. Hence, the reverse current in  (\ref{eq:total_junction_current}) becomes negligible, i.e., $i_{\rm R}(t)\approx 0$. In this case,  the diode current in (\ref{eq:total_junction_current}) reduces to the well-known Shockley diode equation given by $i_{\rm d_j}(t)=i_{\rm F}(t)=I_{\rm s}\big(\e^{\frac{v_{\rm d_j}(t)}{\eta V_{\rm T}}}-1\big)$ \cite{Shockley1949}. Consequently, the second bracketed term on the left hand side (LHS) of (\ref{eq:Vout_intermsof_va}) tends to zero. Thus, in the low input power regime, (\ref{eq:Vout_intermsof_va}) reduces  to 
\begin{IEEEeqnarray}{ll}
I_0\left(\beta\right)=\left(1+\frac{V_{\rm out}}{I_{\rm s}R_{\rm L}}\right)\e^{\frac{V_{\rm out}}{\eta V_{\rm T}}\left(1+\frac{R_{\rm s}}{R_{\rm L}}\right)}, 
\label{eq:Vout_intermsof_va_LP}
\end{IEEEeqnarray}
which was given in \cite[Eq. (15)]{Polozec1994}. Multiplying both sides of  (\ref{eq:Vout_intermsof_va_LP}) by $\frac{I_{\rm s}(R_{\rm L}+R_{\rm s})}{\eta V_{\rm T}}\e^{\frac{I_{\rm s}(R_{\rm L}+R_{\rm s})}{\eta V_{\rm T}}}$, we get
\begin{equation}
\begin{aligned}
&\frac{I_{\rm s}(R_{\rm L}+R_{\rm s})}{\eta V_{\rm T}}\e^{\frac{I_{\rm s}(R_{\rm L}+R_{\rm s})}{\eta V_{\rm T}}}I_0\!\left(\beta\right)\\
&\!=\!\frac{I_{\rm s}(R_{\rm L}\!+\!R_{\rm s})}{\eta V_{\rm T}}\!\left(\!1\!+\!\frac{V_{\rm out}}{I_{\rm s}R_{\rm L}}\!\right)\e^{\frac{I_{\rm s}(R_{\rm L}+R_{\rm s})}{\eta V_{\rm T}}\left(1+\frac{V_{\rm out}}{I_{\rm s}R_{\rm L}}\right)}.
\label{eq:wew_known}
\end{aligned}
\end{equation}
The right hand side (RHS) of (\ref{eq:wew_known}) has the form $w\e^w$, where $w\!=\!\frac{I_{\rm s}(R_{\rm L}+R_{\rm s})}{\eta V_{\rm T}}\left(1\!+\!\frac{V_{\rm out}}{I_{\rm s}R_{\rm L}}\right)$. Since function $w\e^w$ is invertible for $w\e^w\in[0,\infty)$ and the LHS of (\ref{eq:wew_known}) is $\in [0,\infty)$, the unknown $w$ has a unique solution given by $w\!=\!W_0\left(a\e^aI_0\left(\beta\right)\right)$, where $a\!=\!\frac{I_{\rm s}(R_{\rm L}+R_{\rm s})}{\eta V_{\rm T}}$ and $W_0(\cdot)$ is the principal branch of the LambertW function \cite{Corless1996}.
Hence, for low input RF powers, $V_{\rm out}$ can be obtained in closed form as $V_{\rm out}(\beta)=\left[\frac{1}{a}W_0\left(a\e^aI_0\left(\beta\right)\right)-1\right] I_{\rm s}R_{\rm L}$ and the harvested DC power $P_{\rm out}(\beta)=V_{\rm out}^2(\beta)/R_{\rm L}$ reduces to\footnote{\color{black}In \cite[Eq. (22)]{WPT_LambertW_2019}, the output DC power is expressed in terms of the LambertW function of an integral involving the received  signal. Unlike the analysis in this paper, \cite[Eq. (22)]{WPT_LambertW_2019} considers only the diode's forward current, uses an approximation of Shockley's diode equation, assumes perfect matching between the antenna and the rectifier, and assumes a zero diode series resistance, i.e., $R_{\rm s}= 0$.}
\begin{equation}
P_{\rm out}(\beta)\Big|_{\rm low\,\, power}\!=\left[\frac{1}{a}W_0\left(a\e^aI_0\left(\beta\right)\right)\!-\!1\right]^2 I_{\rm s}^2R_{\rm L}.
\label{eq:Pout_interms_Va_LP}
\end{equation}
\end{remark}

 \subsubsection{Matching Network Model}
\label{sss:Matching_Network_Model}
In Section \ref{sss:Saturation_Circuit_Model}, we obtained the output DC power in terms of the amplitude of the signal at the rectifier's input. In this section, we obtain the output DC power in terms of the power of the RF input signal received by the antenna. To this end, it is essential to model the power transfer from the antenna to the rectifier. This power transfer is maximized by a complex-conjugate matching network that matches the rectifier input impedance\footnote{\color{black}Note that the rectifier's input impedance $Z_{\rm a}$ is not only frequency dependent but also input power dependent. Hence, in general, the matching network should be tuned to the RF input power.} $Z_{\rm a}$ to the antenna impedance $R_{\rm ant}$ \cite[Section 5.3.3]{Georgiadis_WPT_book_2016} \footnote{{\color{black}As seen from the antenna side, the matching network typically includes a series capacitor that acts as a DC block followed by a  shunt inductor that acts as a DC feed providing a DC return path for the rectified current \cite[Fig. 2]{Optimum_behaviour_Georgiadis2013}. Otherwise, the rectified current would pass through the RF signal generator (the antenna) and a part of the rectified power would be consequently lost in the antenna \cite{Polozec1994,Polozec2015_Input_Impedance,Georgiadis_WPT_book_2016}.}}. In particular,
let $\{Z_{\rm b}$, $v_{\rm b}(t)\}$ and $\{Z_{\rm a}$, $v_{\rm a}(t)\}$ be the input impedance and the voltage signal \emph{before} and \emph{after} the matching network, respectively, cf. Fig.~\ref{fig:rectenna_model}.  Then, average power conservation during one symbol duration {\color{black}assuming a lossless matching network} implies 
\begin{equation}
\begin{aligned}
&\Re\left\{\frac{1}{T}\int_T\frac{|v_{\rm a}(t)|^2}{Z_{\rm a}^*}  \dd t\right\}=\Re\left\{\frac{1}{T}\int_T\frac{|v_{\rm b}(t)|^2 }{Z_{\rm b}^*} \dd t\right\} \\
& \Rightarrow \quad \frac{1}{2} \hat{v}^2_{\rm a}\Re\left\{\frac{1}{Z_{\rm a}^*}\right\}=\frac{1}{2} \hat{v}^2_{\rm b}\Re\left\{\frac{1}{Z_{\rm b}^*}\right\}.
\end{aligned}
\label{eq:power_conservation}
\end{equation}
For perfect matching, $Z_{\rm b}\!=\!R_{\rm ant}$ and $\hat{v}_{\rm a}\!=\!\hat{v}_{\rm b}/\sqrt{\Re\{R_{\rm ant}/Z_{\rm a}^*\}}$. Since  $Z_{\rm a}$ is typically larger than  $R_{\rm ant}$,  it follows that the amplitude of the voltage signal at the output of the matching network  is higher than that at its input, i.e., $\hat{v}_{\rm a}>\hat{v}_{\rm b}$  \cite{EH_solid_state_2008}. Hence, the matching network effectively acts as a voltage multiplier. }

Next, we obtain relationships between the input RF power and the signal peak amplitudes $\hat{y}_{\rm E}^{ \rm RF}$, $\hat{v}_{\rm s} $, $\hat{v}_{\rm b}$, and $\hat{v}_{\rm a}$. In particular, for perfect matching, the average received RF power captured by the antenna during one time slot, denoted by $P_{\rm in}$, is completely transferred to the rectifier, i.e., $P_{\rm in}\!=\!\frac{1}{T}\int_T |y_{\rm E}^{\rm RF}(t)|^2 \dd t\!=\!\frac{1}{T}\int_T |v_{\rm b}(t)|^2/R_{\rm ant}  \dd t$, or equivalently $\hat{v}_{\rm b}\!=\!\hat{y}_{\rm E}^{\rm RF}\sqrt{R_{\rm ant}}$. 
Hence, $\hat{v}_{\rm a}\!=\!\hat{y}_{\rm E}^{\rm RF}/\sqrt{\Re\{1/Z_{\rm a}^*\}}$. Assuming symbol $x$ is transmitted in  the time slot under consideration, then from (\ref{eq:recieved_signal_EH}), $\hat{y}_{\rm E}^{ \rm RF}\!=\!\sqrt{2} xh_{\rm E}$ and the average input power of the received RF signal, $P_{\rm in}$, can be written as  
\begin{equation}
P_{\rm in}\!=\!\frac{1}{2}\!\left(\hat{y}_{\rm E}^{ \rm RF}\right)^2\!\!=\! (x h_{\rm E})^2\!=\!\frac{1}{2} \!\frac{\hat{v}_{\rm b}^2}{R_{\rm ant}}\!=\!\frac{1}{2} \hat{v}_{\rm a}^2\Re\left\{\!\!\frac{1}{Z_{\rm a}^*}\!\!\right\}\!\!=\!\frac{\hat{v}_{\rm s}^2}{8 R_{\rm ant}}, 
\label{eq:Pin} 
\end{equation}
where we used $\hat{v}_{\rm b}=\hat{v}_{\rm s}/2$ for perfect matching.
Using (\ref{eq:Pin}) and defining $B\definedas 1/(\eta V_{\rm T}\sqrt{\Re\{1/Z_{\rm a}^*\}})$,  then  the argument of the modified Bessel function, $\beta=\frac{\hat{v}_{\rm a}}{\eta V_{\rm T}}$, defined  in  (\ref{eq:Vout_intermsof_va}) can be written as 
\begin{equation}
\beta\!=\!\frac{\hat{v}_{\rm a}}{\eta V_{\rm T}}=B \hat{y}_{\rm E}^{\rm RF}=B\frac{ \hat{v}_{\rm b}}{\sqrt{R_{\rm ant}}}=B \sqrt{2 P_{\rm in}}=\sqrt{2}B h_{\rm E} x.
\label{eq:Bessel_function_argument}
\end{equation}
Using (\ref{eq:Bessel_function_argument}), the harvested DC power  in (\ref{eq:Pout_Sat_exact}) and (\ref{eq:Pout_interms_Va_LP}) can be expressed in terms of the input power $P_{\rm in}$, the transmit symbol $x$, and the peak amplitudes $\hat{v}_{\rm a}$, $\hat{y}_{\rm E}^{\rm RF}$, $\hat{v}_{\rm b}$. However, from  (\ref{eq:Bessel_function_argument}), this requires the knowledge of $B=1/(\eta V_{\rm T}\sqrt{\Re\{1/Z_{\rm a}^*\}})$ and therefore the input impedance of the rectifier $Z_{\rm a}$. According to the diode model in \cite[Fig. 12]{Polozec2015_Input_Impedance}, the diode junction can be modelled as a variable resistor $R_{\rm d}$ whose value depends on the input RF power. Hence, the rectifier input impedance $Z_{\rm a}$ in Fig. \ref{fig:rectenna_model} can be written as
\begin{equation}
Z_{\rm a}(R_{\rm d})=R_{\rm s}+\left[1/R_{\rm d}+ j\omega_{\rm c} C_{\rm j}\right]^{-1}+\left[1/R_{\rm L}+ j\omega_{\rm c} C_{\rm L}\right]^{-1}.
\label{eq:Za_general}
\end{equation}
It was shown in \cite{Polozec2015_Input_Impedance} that for small input RF powers, $R_{\rm d}\!\to\! R_{\rm j_0}\!\definedas\!\eta V_{\rm T}/I_{\rm s}$, whereas for high input RF powers, $R_{\rm d}\!\to\! R_{\rm L}/2$, see \cite[caption of Fig. 12, Eq. (49), and Fig. 10(a)]{Polozec2015_Input_Impedance}. Note that, $C_{\rm j}$ depends on the output DC voltage \cite[Eq. (2)]{Polozec2015_Input_Impedance}. However, we use the approximation $C_{\rm j}\approx C_{\rm j_0}$, where $C_{\rm j_0}$ is the diode's junction capacitance at zero output DC voltage provided in the diode's datasheet \cite{Polozec2015_Input_Impedance}.  As will be shown in Section \ref{s:numerical_results}, these approximations provide DC powers close to those obtained by circuit simulations. 

\begin{remark}
 Note that if perfect matching between the antenna and the rectifier is assumed without including a matching network, as was done in \cite{Waveform_optimization_SPAWC_Rui_Zhang_2017} and  \cite{Waveform_design_WPT_Clerckx_2016}, the  voltage multiplication in (\ref{eq:power_conservation}) is not included in the model,, i.e., $\hat{v}_{\rm a}=\hat{v}_{\rm b}$ is assumed, which leads to an underestimation of the actual harvested DC power.
\end{remark}
 \subsubsection{Approximate Saturation Model}
\label{sss:Approximate_Saturation_Model}
In \cite[Section 6.5.1]{Georgiadis_WPT_book_2016}, it is shown that for very high input RF powers, the DC output voltage of the diode detector in Fig. \ref{fig:rectenna_model} saturates at $V_{\rm out}|_{\max}={B_{\rm v}}/2$. Hence, the saturated harvested DC power is given by  $P_{\rm out}|_{\max}={B_{\rm v}}^2 /(4R_{\rm L})$ \cite[Eq. (6.5)]{Georgiadis_WPT_book_2016}. Since the solution of the saturation model in (\ref{eq:Vout_intermsof_va}) cannot be obtained in closed form, we combine  the low-power approximate solution in (\ref{eq:Pout_interms_Va_LP}) with  $P_{\rm out}|_{\max}$ to obtain an approximate solution for the saturation model, namely 
\begin{equation}
P_{\rm out}(\beta)\Big|_{\rm approx.}\!\!\!\!\!\!=\!\min\!\!\left(\!\left[\!\frac{1}{a}W_0\left(a\e^aI_0\left(\beta\right)\right)\!-\!1\!\right]^2\!\!\!I_{\rm s}^2R_{\rm L},\frac{B_{\rm v}^2}{4R_{\rm L}}\!\right)\!.
\label{eq:Pout_Sat_approx}
\end{equation}
Using (\ref{eq:Bessel_function_argument}), we write the harvested DC power in terms of  transmit symbol $x$ for the $l^{\text{th}}$ EH Rx as
\begin{equation}  
\begin{aligned}
&P_l(x)\definedas P_{\rm out}\left(\sqrt{2}B h_{{\rm E}_l} x\right)\Big|_{\rm approx.}\\&=\min\!\left(\!\left[\frac{1}{a}W_0\left(a\e^aI_0\left(\!\sqrt{2}B h_{{\rm E}_l} x\right)\!\right)\!-\!1\right]^2\!\!I_{\rm s}^2R_{\rm L},\frac{B_{\rm v}^2}{4R_{\rm L}}\right)\!,
\end{aligned}
\label{eq:P_x_closedform}
\end{equation}
where $B\!=\!1/(\eta V_{\rm T}\sqrt{\Re\{1/Z_{\rm a}(R_{\rm j_0})^*\}})$, i.e.,  we use the low-power approximation of the rectifier's input impedance\footnote{{\color{black}We note that the harvested DC power function in (\ref{eq:P_x_closedform}) depends on the circuit parameters $R_{\rm s}$, $C_{\rm j}$, $C_{\rm L}$, and $R_{\rm L}$ since $P_l(x)$ is a function of $B\!=\!1/(\eta V_{\rm T}\sqrt{\Re\{1/Z_{\rm a}(R_{\rm j_0})^*\}})$ and the rectifier's input impedance $Z_{\rm a}(R_{\rm j_0})$ in (\ref{eq:Za_general}) is a function of these circuit parameters.}}.
Our numerical results in Section \ref{s:numerical_results} confirm that both the exact and the approximate output DC power functions in (\ref{eq:Pout_Sat_exact}) and (\ref{eq:Pout_Sat_approx}), respectively, are in good agreement with circuit simulations, cf. Fig. \ref{fig:Harvested_DC_power_ct_theory}. Hence, both  expressions may be used for the EH constraints of the SWIPT problem. For notational simplicity, in the following, we will use (\ref{eq:P_x_closedform}) for the EH constraints of the conditional capacity SWIPT problem.
Next, we derive the input RF power $P_{\rm in,sat}$ at which the output DC power starts to saturate. From (\ref{eq:Pout_Sat_approx}), saturation of the output DC power occurs when $\left[\frac{1}{a}W_0\left(a\e^aI_0\left(\beta\right)\right)-1\right]^2 I_{\rm s}^2R_{\rm L}=\frac{{B_{\rm v}}^2}{4R_{\rm L}}$. Using (\ref{eq:Pin}), (\ref{eq:Bessel_function_argument}), we obtain   
\begin{equation}
P_{\rm in,sat}=\left(\frac{\beta_{\rm sat}}{\sqrt{2}B}\right)^2=\frac{1}{2}\left(\eta V_{\rm T}\beta_{\rm sat}\right)^2\Re\left\{\frac{1}{Z_{\rm a}^*(R_{\rm j_0})}\right\},
\label{eq:Pin_sat}
\end{equation}
where $\beta_{\rm sat}$ is the solution of  $I_0\left(\beta_{\rm sat}\right)=\e^{\frac{a{B_{\rm v}}}{2I_{\rm s}R_{\rm L}}}(1+{B_{\rm v}}/(2I_{\rm s}R_{\rm L}))$. In Table \ref{tab:Sat_model_summary}, we summarize the main results of the saturation model derived in this section.
}
{\color{black} \subsection{Amplitude Constraints for the Transmitter and the EH Rxs}
\label{sss:Amplitude_Constraints_Tx_Rx}
 At the transmitter, the PP is usually limited to avoid the negative impact of power amplifier nonlinearities, i.e., we set $|\V{x}|\leq A_{\rm T}$. Moreover, it may be desired to limit the peak amplitude of the received RF signal at  EH Rx $l$ to some value $A_{{\rm R}_l}$, i.e., $\hat{y}_{{\rm E}_l}^{\rm RF}\!=\!|\sqrt{2}\V{x}h_{{\rm E}_l}|\!\leq\! A_{{\rm R}_l}$. 
Considering the PP constraints at the transmitter and all $L$ EH Rxs, the effective amplitude (or PP) constraint on the transmit signal reduces to
\begin{IEEEeqnarray}{llll}
|\V{x}|\leq\min\left(A_{\rm T},\min\limits_{l\in\mathcal{L}}A_{{\rm R}_l}/|\sqrt{2}h_{{\rm E}_l}|\right)\definedas A.
\label{eq:PP_constraint}
\end{IEEEeqnarray}
For example, setting the maximum received amplitude to $A_{{\rm R,sat}_l}\definedas\sqrt{2P_{\rm in,sat}}$,  $\forall \,l\in\mathcal{L}$, ensures that none of the EH Rxs operates in saturation\footnote{\color{black}We assume that the transmitter has perfect knowledge of the circuit parameters of the EH Rxs. The EH Rxs may retransmit these parameters to the  transmitter in intervals dictated by variations due to temperature changes and aging.}. In this case, the maximum transmit amplitude $A$ is given by   $A_{\rm sat}\definedas\min\Big(A_{\rm T},\min_{l\in\mathcal{L}}A_{{\rm T,sat}_l}\Big)$, where $A_{{\rm T,sat}_l}\definedas A_{{\rm R,sat}_l}/|\sqrt{2}h_{{\rm E}_l}|=\sqrt{P_{\rm in,sat}}/|h_{{\rm E}_l}|$.}

Using the closed-form expression for the harvested DC power in  (\ref{eq:P_x_closedform}) and the PP constraint in (\ref{eq:PP_constraint}), 
 we formulate next the conditional capacity problem of the SWIPT system in Fig. \ref{fig:system_model}.
\begin{table*}[!tp] 
\caption{ {\color{black}Summary of the main results of the saturation circuit model developed in Section \ref{ss:rectenna_model}.}}
{\color{black}
\begin{tabular}{@{}ll@{}}  \toprule   
Parameter & Expression \\ \midrule 
Given & $I_{\rm s}$, $\eta$, $V_{\rm T}$, $I_{B_{\rm v}}$, ${B_{\rm v}}$, $R_{\rm s}$, $C_{\rm j_0}$, $R_{\rm L}$, $C_{\rm L}$, $R_{\rm ant}$, $R_{\rm j_0}\!=\!\eta V_{\rm T}/I_{\rm s}$, $C_{\rm j}\approx C_{\rm j_0}$, $f_{\rm c}$, \\
& and $P_{\rm in}$ or $\hat{v}_b$  or $\hat{v}_a$ or $x$ and $h_{\rm E}$\\ 
Diode I-V characteristic & 
$i_{\rm d_j}(t)=i_{\rm F}(t)+i_{\rm R}(t)=I_{\rm s}\big(\e^{\frac{v_{\rm d_j}(t)}{\eta V_{\rm T}}}-1\big)-I_{B_{\rm v}}\e^{-\frac{{B_{\rm v}}}{\eta V_{\rm T}}}\big(\e^{-\frac{v_{\rm d_j}(t)}{\eta V_{\rm T}}}-1\big)$\\ \addlinespace[0.5em]
Rectifier input impedance $Z_{\rm a}$ & $Z_{\rm a}(R_{\rm d})=R_{\rm s}+\left[1/R_{\rm d}+ j\omega_{\rm c} C_{\rm j}\right]^{-1}+\left[1/R_{\rm L}+ j\omega_{\rm c} C_{\rm L}\right]^{-1}$\\ \addlinespace[0.5em]
Approximations of $Z_{\rm a}$ & $Z_{\rm a}\big|_{\rm low\,\,power}\to Z_{\rm a}(R_{\rm j_0})$ and $Z_{\rm a}\big|_{\rm high\,\,power}\to Z_{\rm a}(R_{\rm L}/2)$ \\ \addlinespace[0.5em]
Bessel function argument $\beta$ & $\beta=\frac{\hat{v}_{\rm a}}{\eta V_{\rm T}}=B \hat{y}_{\rm E}^{\rm RF}=B\frac{\hat{v}_{\rm b}}{\sqrt{R_{\rm ant}}}=B \sqrt{2 P_{\rm in}}=\sqrt{2}B h_{\rm E} x$, where $B= \left[\eta V_{\rm T}\sqrt{\Re\{1/Z_{\rm a}^*\}}\right]^{-1}$\\ \addlinespace[0.5em]
Exact harvested DC power & $P_{\rm out}(\beta)\Big|_{\rm exact}\!=\!\frac{V_{\rm out}^2(\beta)}{R_{\rm L}}$, where $V_{\rm out}(\beta)$ is the solution of\\\addlinespace[0.5em] &\hspace{0cm} $I_{\rm s}\left[\e^{\frac{-V_{\rm out}(\beta)}{\eta V_{\rm T}}\left(1+\frac{R_{\rm s}}{R_{\rm L}}\right)}I_0\left(\beta\right)-1\right]-I_{B_{\rm v}}\e^{-\frac{{B_{\rm v}}}{\eta V_{\rm T}}}\left[\e^{\frac{V_{\rm out}(\beta)}{\eta V_{\rm T}}\left(1+\frac{R_{\rm s}}{R_{\rm L}}\right)}I_0\left(\beta\right)-1\right]=\frac{V_{\rm out}(\beta)}{R_{\rm L}}$\\\addlinespace[0.5em]
Approximate harvested DC power & 
$P_{\rm out}(\beta)\Big|_{\rm approx.}=\min\left(\left[\frac{1}{a}W_0\left(a\e^aI_0\left(\beta\right)\right)-1\right]^2 I_{\rm s}^2R_{\rm L},\frac{{B_{\rm v}}^2}{4R_{\rm L}}\right),$ where $a\!=\!\frac{I_{\rm s}(R_{\rm L}+R_{\rm s})}{\eta V_{\rm T}}$.\\\addlinespace[0.5em]
RF input power at $P_{\rm out}\!=\!\frac{{B_{\rm v}}^2}{4R_{\rm L}}$  & $P_{\rm in,sat}=\frac{1}{2}\left(\eta V_{\rm T}\beta_{\rm sat}\right)^2\Re\left\{\frac{1}{Z_{\rm a}^*(R_{\rm j_0})}\right\}$, where $\beta_{\rm sat}$ is the solution of  $I_0\left(\beta_{\rm sat}\right)=\e^{\frac{a{B_{\rm v}}}{2I_{\rm s}R_{\rm L}}}\left(1\!+\!\frac{{B_{\rm v}}}{2I_{\rm s}R_{\rm L}}\right)$. \\
\bottomrule
\end{tabular}}
\label{tab:Sat_model_summary}
\end{table*}

%% file: Problem_formulation.tex
In this section, we study the conditional capacity of the considered AWGN channel under AP and PP  constraints on the transmit signal and EH constraints at the EH Rxs. We first prove that the optimal input distribution for the transmit  symbols is unique and discrete with a finite number of mass points. In addition, we provide necessary and sufficient conditions for the optimal input distribution. Moreover, we show that if none of the EH Rxs operates in saturation, the R-E tradeoff curve for the problem with multiple EH Rxs can be obtained from the individual R-E curves obtained for each EH Rx separately.
\subsection{Problem Formulation} 
The discrete-time baseband model for the information channel after down-conversion, matched filtering, and sampling of the continuous-time signal received at the ID Rx is given by $\V{y}\!=\!h_{\rm I} \V{x}\! +\!\V{n}$,
where $\V{n}\sim\mathcal{N}(0,\sigma_n^2)$ is the Gaussian distributed noise and $\V{y}$ is the information channel output with probability density function (pdf) $p(y)$. 
We aim at maximizing the average mutual information between $\V{x}$ and $\V{y}$ subject to maximum AP and PP constraints at the transmitter and minimum harvested power constraints at the EH Rxs. In particular, we formulate the problem as
\begin{IEEEeqnarray}{llll}
\!\!\!\!\!C\!=&\sup\limits_{F\in\mathcal{F}_A} &&I(F) \nonumber\\
&\,\,{\rm s.t.} &&{\rm C}_0:\,\E_F[\V{x}^2] \leq \sigma^2; \notag\\
& &&{\rm C}_l\,:\, \E_F\left[P_l(\V{x})\right] \geq P_{l,\req},\,\, \forall\,l\in\mathcal{L},
\label{eq:capacity_problem}
\end{IEEEeqnarray}
where $\mathcal{F}_A$ is the set of all input distributions of $\V{x}$ that satisfy the PP constraint $|\V{x}|\!\leq\!A$ in (\ref{eq:PP_constraint}), i.e., $\forall F\!\in\!\mathcal{F}_A$, $\int_{-A}^A\dd F(x)\!=\!1$.
$I(F)$ is the mutual information between $\V{x}$ and $\V{y}$ achieved by input distribution $F$ and given by $I(F)\!=\!\int_{-A}^A i(x;F)\dd F(x)$, where $i(x;F)$ is the marginal information density defined as $i(x;F)\!\definedas\!\int_yp(y|x) \log_2\frac{p(y|x)}{p(y;F)}\dd y$, $p(y;F)$ is the output pdf assuming input distribution $F$, and $p(y|x)$ is the  output pdf conditioned on the transmission of symbol $x$ \cite{SMITH19712}.  $\sigma^2$ is the AP budget, $P_l(\V{x})$ is the harvested power function at the $l^{\text{th}}$ EH Rx given in (\ref{eq:P_x_closedform}), and $P_{l,\req}$ is the minimum required   harvested power at EH Rx $l$. For the purpose of exposition, we define $g_0(F)\!\definedas\!\int_{-A}^A x^2\dd F(x)\! - \!\sigma^2$ and $g_l(F)\!\definedas\! P_{l,\req}\!-\!\int_{-A}^A P_l(x)\dd F(x)$, $\forall\,l\in\mathcal{L}$. Hence, constraints ${\rm C}_{0}$ and ${\rm C}_{l}, \,\forall\,  l\in\mathcal{L}$, can be written as $g_l(F)\!\leq\! 0$, $\forall\,l\in\{0\}\cup\mathcal{L}$. 
\subsection{Properties of the Optimal Input Distribution}
In the following, we investigate some important properties of the optimal input distribution.
\subsubsection{Uniqueness of the Optimal Input Distribution}
We establish the uniqueness of the optimal input distribution for problem (\ref{eq:capacity_problem}) in the following theorem. 
\begin{theorem}\normalfont
The conditional capacity $C$ in (\ref{eq:capacity_problem}) is achieved by a \emph{unique} optimal input distribution function $F_0$, i.e., $C=\sup\limits_{F\in\Omega}I(F)=I(F_0)$, where $\Omega\subset \mathcal{F}_A$ is the set of input distributions that satisfy the PP constraint and constraints ${\rm C}_{l}$, $\forall\,l\in\{0\}\cup\mathcal{L}$, in (\ref{eq:capacity_problem}). Furthermore, there exist $\lambda_l\geq 0$, $\forall\,l\in\{0\}\cup\mathcal{L},$ such that the conditional capacity $C$ is equivalently given by $C\!=\!\sup\limits_{F\in\mathcal{F}_A} I(F)-\sum_{l\in\{0\}\cup\mathcal{L}} \lambda_l g_l(F)$, which is  also achieved by $F_0$ and
$\lambda_lg_l(F_0)\!=\!0$, $\forall\,\,l\in\{0\}\cup\mathcal{L}$.
\label{theo:unique_distribution}
\end{theorem}
\begin{proof}
The proof is provided in Appendix \ref{App:proof_unique_optimal_distribution}.
\end{proof} 
\subsubsection{Necessary and Sufficient Conditions for the Optimal Input Distribution}
The following theorem provides a necessary and sufficient condition for the optimal input distribution $F_0$.
\begin{theorem}\normalfont\label{theo:C_necessary_sufficient_condition1}
A necessary and sufficient condition for the input distribution $F_0$ to achieve the conditional capacity $C$ in (\ref{eq:capacity_problem}) is that $\forall F\in\mathcal{F}_A$, there exist $\lambda_l\geq 0$, $\forall\,l\in\{0\}\cup\mathcal{L},$ such that
\begin{IEEEeqnarray}{ll} 
\int\limits_{-A}^A \left[i(x;F_0)-\lambda_0 x^2+\sum\limits_{l\in\mathcal{L}} \lambda_lP_l(x)\right]\dd F(x) \notag\\
 \leq C-\lambda_0\sigma^2+\sum\limits_{l\in\mathcal{L}} \lambda_l P_{l,\req}.\label{eq:C_necessary_sufficient_condition1}
\end{IEEEeqnarray}
\end{theorem}
\begin{proof}
The proof is provided in Appendix \ref{App:proof_C_necessary_sufficient_condition1}.
\end{proof}
Define the points of increase of a distribution function $F$ as those points which have non-zero probability \cite{SMITH19712}.
Next, we provide a more useful condition for characterizing the optimal input distribution.
\begin{corollary}\normalfont
Let $E_0$ be the set of points of increase of a distribution function $F_0$ on $[-A,A]$, then $F_0$ is the optimal input distribution of problem  (\ref{eq:capacity_problem}) if and only if there exist $\lambda_l\geq 0$, $\forall\,l\in\{0\}\cup\mathcal{L}$, such that 
\begin{equation}
\begin{aligned}
&s(x)\!\definedas\!\lambda_0\!\left(x^2\!-\!\sigma^2\right)\!-\!\sum\limits_{l\in\mathcal{L}}\!\lambda_l\left(P_l(x)\!-\!P_{l,\req}\right)\!+\!C\\&\!+\!\frac{1}{2}\!\log_2(2\pi\e\sigma_n^2)\!+\!\! \int\!\!\frac{\e^{-\!\frac{(y-xh_{\rm I})^2}{2\sigma_n^2}}}{\sqrt{2\pi\sigma_n^2}}\log_2(p(y;\!F_0))\dd y \!\geq\! 0,
\end{aligned}
\label{eq:C_necessary_sufficient_condition2}
\end{equation}
$\forall x\in[-A,A]$, where equality holds if $x$ is a point of increase of $F_0$, i.e., if $x\in E_0$.
\label{corol:C_necessary_sufficient_condition2}
\end{corollary}
\begin{proof}
The proof is provided in Appendix \ref{App:proof_C_necessary_sufficient_condition2}.
\end{proof}

\subsubsection{Discreteness of the Optimal Input Distribution}
\label{sss:discrete_dist}
The discreteness of the optimal input distribution $F_0$ for problem (\ref{eq:capacity_problem})   is formally stated in the following theorem. 
\begin{theorem}\normalfont
The optimal input distribution that achieves the conditional capacity in (\ref{eq:capacity_problem}) is discrete with a finite number of mass points.
\label{theo:discrete_optimal_distribution}
\end{theorem}
\begin{proof}
The proof is provided in Appendix \ref{App:proof_discrete_optimal_distribution}.
\end{proof}
\subsection{The Activeness of Only One EH Constraint in (\ref{eq:capacity_problem}) for $A< A_{\rm sat}$
\label{ss:one_EH_constraint_active}}
\input{Only_one_EH_Rx_effective}

%% file: Only_one_EH_Rx_effective.tex
In this section, we consider the case when the transmit amplitude is set to  $A< A_{\rm sat}$ to avoid the saturation of the harvested DC power at the EH Rxs. We prove that, in this case, at the optimal solution,  at most one of the EH constraints of problem (\ref{eq:capacity_problem}) is active. 
 We note that, owing to the random deployment of the EH Rxs, their channel gains are different {\color{black}with probability one}. Moreover, each EH Rx sets its minimum  DC power requirement independently. As a result, the solution of problem (\ref{eq:capacity_problem}) with only EH constraint ${\rm C}_l$ is different from that with only EH constraint ${\rm C}_{l'}$,  $\forall\,l\neq l'\in\mathcal{L}$.

\begin{lemma}
\label{lem:monotonicity_Power_function_before_sat}
Considering the harvested power model in (\ref{eq:P_x_closedform}), in the unsaturated case, i.e., for $A< A_{\rm sat}$, 
 if an input distribution $F_{\rm L}(x) {\color{black}\in\mathcal{F}_A}$  provides a larger average harvested power than another distribution $F_{\rm s}(x){\color{black}\in \mathcal{F}_A}$, for one EH Rx, then $F_{\rm L}(x)$ also provides larger average harvested powers than  $F_{\rm s}(x)$, for all other EH Rxs. That is, if for some $l$, $\E_{F_{\rm L}}[P_{l}(\V{x})]\!>\!\E_{F_{\rm s}}[P_{l}(\V{x})]$, then $\E_{F_{\rm L}}[P_{\tilde{l}}(\V{x})]\!>\!\E_{F_{\rm s}}[P_{\tilde{l}}(\V{x})]$, $\forall$ $\tilde{l}\in\mathcal{L}$. 
\end{lemma} 
\begin{proof}
Lemma 
\ref{lem:monotonicity_Power_function_before_sat} follows since the harvested power function in (\ref{eq:P_x_closedform})  for $A\!<\! A_{\rm sat}$ is  monotonically increasing for {\color{black}$0\!<\!x\!<\!A$}, $\forall$ EH Rxs, since the Bessel function $I_0(\cdot)$, the LambertW function $W_0(\cdot)$, and the quadratic function are all monotonically increasing for $x\!>\!0$.
 Similarly, the harvested power function in (\ref{eq:P_x_closedform}) is monotonically increasing in the channel gain $h_{{\rm E}_l}$. That is, if $h_{{\rm E}_1}\!>\!h_{{\rm E}_2}$, then $P_1(x)\!>\!P_2(x)$, $\forall$ $0<x<A$.  
Hence, if the integration of one $P_l(x)$ with respect to some distribution $F_{\rm L}(x)\!\in \!\mathcal{F}_A$ is larger than with respect to distribution $F_{\rm s}(x)\!\in \!\mathcal{F}_A$, i.e., $\int_0^A P_{l}(x)\dd F_{\rm L}(x)\!>\!\int_0^A P_{l}(x)\dd F_{\rm s}(x)$, then this relation must also hold for any other EH Rx $l'$, i.e., $\int_0^A P_{l'}(x)\dd F_{\rm L}(x)\!>\!\int_0^A P_{l'}(x)\dd F_{\rm s}(x)$, $\forall$ $l'\!\neq \!l\!\in\!\mathcal{L}$.
\end{proof}
\begin{theorem}\normalfont
In problem (\ref{eq:capacity_problem}), for $A< A_{\rm sat}$, at most one EH constraint is active. In particular, the active EH constraint is the one, which when all other EH constraints are removed, results in the smallest achievable rate at the ID Rx, denoted by $I(F_0)$.
\label{theo:Only_one_EH_Rx_effective}
\end{theorem}
\begin{proof}
The proof is provided in Appendix \ref{App:proof_Only_one_EH_Rx_effective}.
\end{proof}
The R-E tradeoff curve associated with problem (\ref{eq:capacity_problem}) is an $(L\!+\!1)$-dimensional curve formed by the points $(I(F_0),\E_{F_0}[P_1(\V{x})],\ldots,\E_{F_0}[P_L(\V{x})])$ obtained by solving (\ref{eq:capacity_problem}) for all combinations of feasible minimum required DC powers $P_{l,\req}$, $l\in\mathcal{L}$,  at the EH Rxs. Owing to Theorem \ref{theo:Only_one_EH_Rx_effective}, for $A< A_{\rm sat}$, this $(L\!+\!1)$-dimensional R-E curve can be obtained from the $L$ two-dimensional R-E curves of the individual EH Rxs, where the individual  R-E curve of EH Rx $l$ is obtained by solving problem (\ref{eq:capacity_problem}) with the AP and PP constraints and only EH constraint ${\rm C}_l$ for different required DC powers $P_{l,\req}$. In particular, assuming $K$ required DC powers for each EH Rx, problem (\ref{eq:capacity_problem}) has to be solved only $KL$ times to determine the corresponding $(L\!+\!1)$-dimensional R-E curve instead of $K^L$ times.
{\color{black}
\begin{remark}
The results in this section hold only for $A\!<\! A_{\rm sat}$. If $A\!\geq\!A_{\rm sat}$, then some EH Rxs may operate in saturation. In particular, let $A\!\geq\!A_{{\rm T,sat}_l}$, $\forall$ $l\!\in\!\mathcal{L}_{\rm sat}$ and $A\!< \!A_{{\rm T,sat}_l}$, $\forall$ $l\!\in\!{\cal{L}_{\rm non-sat}}$, then according to Theorem \ref{theo:Only_one_EH_Rx_effective}, at most one EH constraint of the EH Rxs in set $\mathcal{L}_{\rm non-sat}$ may be active. However, in addition, also more than one EH constraint for the EH Rxs in set $\mathcal{L}_{\rm sat}$ may be active. This is because, when $A\!\geq\!A_{\rm sat}$ holds, Lemma \ref{lem:monotonicity_Power_function_before_sat} does not hold since $P_l(x)$ is not monotonically increasing in $0\!<\!x\!<\!A$ for $l\!\in\!\mathcal{L}_{\rm sat}$. For example, it will be shown in Section \ref{s:Maximum_Energy_Transfer} that the optimal input distribution that maximizes the average harvested power is different for the EH Rxs in set $\mathcal{L}_{\rm sat}$.   Hence, an input  distribution that provides more energy for one EH Rx in set $\mathcal{L}_{\rm sat}$ may provide less energy for another EH Rx in set $\mathcal{L}_{\rm sat}$.
\end{remark}}

Having established the properties of the optimal solution for problem (\ref{eq:capacity_problem}), we aim next at getting more insights into the optimal distribution by studying special cases and generalizations of problem (\ref{eq:capacity_problem}).

%% file: special_cases_and_generalizations.tex
In this section, we study the special cases of maximum WIT  and maximum WPT systems to obtain further insight. Based on these extreme cases, we propose a suboptimal  but insightful distribution which bridges the gap between the two systems. Then, we generalize problem  (\ref{eq:capacity_problem}) to the complex domain.
 \subsection{Maximum Information Transfer} 
 \label{ss:Max_Info_transfer}
\input{max_WIT}
 \subsection{Maximum Energy Transfer} 
 \label{s:Maximum_Energy_Transfer}
 \input{max_EH_problem_one_user}
 \subsection{Proposed Suboptimal Distribution} 
 \label{ss:suboptimal_distribution}
 \input{suboptimal_distribution}
 \subsection{Complex Signaling}
 \label{s:complex_signaling}
\input{complex_signaling}



%% file: max_WIT.tex
{\color{black}The maximum information transfer rate can be obtained by dropping all $L$ EH constraints in  (\ref{eq:capacity_problem}). In this case, the problem reduces to the capacity of an AP and PP constrained AWGN channel, which was solved by Smith in \cite{SMITH19712}, who proved that the optimal input distribution, denoted by $F_0^{{\rm WIT}}$, is discrete with a finite number of mass points. In addition, he showed that  $F_0^{{\rm WIT}}$ cannot be expressed in closed form but can be obtained numerically resulting in a maximum achievable information rate of  $C_{\max}\!\definedas\!I(F_0^{{\rm WIT}})$. With this solution, the average  harvested power at the $l^{\text{th}}$ EH Rx is $P_{l,\min}\!\definedas\!\E_{F_0^{{\rm WIT}}}[P_l(\V{x})]$. Hence, in problem (\ref{eq:capacity_problem}), if $P_{l,\req}\!\leq\! P_{l,\min}$, $\forall\, l\!\in\!\mathcal{L}$, holds, the harvested power $P_{l,\min}$  is attained at the EH Rxs without compromising the maximum  information rate at the ID Rx\footnote{\color{black}If however, $\exists\, l$ such that $P_{l,\req}> P_{l,\min}$, then the achievable information rate has to be compromised, i.e., $I(F_0)<C_{\max}$, in order for the $l^{\rm th}$ EH Rx to be able to harvest enough energy.}. Furthermore, if additionally the PP constraint is relaxed, i.e., $A\!\to\!\infty$, problem (\ref{eq:capacity_problem}) reduces to the maximization of the mutual information of the AP constrained AWGN channel. For this special case, the optimal input distribution is known to be the \emph{continuous} zero-mean Gaussian distribution given by $\frac{1}{\sqrt{2\pi\sigma^2}}\e^{-x^2/(2\sigma^2)}$ \cite{Shannon}, cf. Case 1 in Appendix \ref{App:proof_discrete_optimal_distribution} for $A\!\to\!\infty$, and the maximum achievable rate is the well-known Shannon capacity given by $\frac{1}{2}\log_2\left(1\!+\!\frac{\sigma^2h_I^2}{\sigma_n^2}\right)$.}

%% file: max_EH_problem_one_user.tex
In this section, we formulate the maximum WPT problem for EH Rx $l$ in (\ref{eq:energy_problem}) and obtain the optimal input distribution  and the maximum  average harvested power in closed form in Theorem \ref{theo:EH_only}. In particular,
 \begin{IEEEeqnarray}{llll}
P_{l,\max}=&\sup\limits_{F\in\mathcal{F}_A}& &\E_F\left[P_l(\V{x})\right]\nonumber\\
&\,\,{\rm s.t.} &&{\rm C_0:}\quad \E_F[\V{x}^2] \leq \sigma^2.
\label{eq:energy_problem}
\end{IEEEeqnarray}
\begin{theorem}\normalfont
Define $A_l'\definedas \min(A,A_{{\rm T,sat}_l})$ with $A_{{\rm T,sat}_l}$ as defined in Section \ref{sss:Amplitude_Constraints_Tx_Rx}. Then, the optimal distribution obtained from problem  (\ref{eq:energy_problem}) has a probability mass function given by
\begin{equation} 
{\color{black}\dd F_0^{{\rm WPT}}(x,A_l')}=\begin{cases}p& x=-A_l'\\
\left[1-(p+q)\right]^+ &x=0\\
q &x=A_l'
\end{cases}.
\label{eq:PMF_max_EH}
\end{equation}
where $p,q\!\geq \!0$ and $p\!+\!q=\min\{\sigma^2/A_l'^2,1\}$. The maximum average harvested power at EH Rx $l$ is 
\begin{equation}
{\color{black}P_{l,\max}(A_l')}=\E_{F_0^{{\rm WPT}}}\left[P_l(\V{x})\right]=\begin{cases}\frac{\sigma^2}{A_l'^2}P_l(A_l') & \frac{\sigma^2}{A_l'^2}<1\\
P_l(A_l') & \frac{\sigma^2}{A_l'^2}>1
\end{cases},
\label{eq:max_EH_analytic}
\end{equation}
where $l\in\mathcal{L}$ and the average mutual information at the ID Rx is $I(F_0^{{\rm WPT}})=\int i(x;F_0^{{\rm WPT}})\dd F_0^{{\rm WPT}}(x)$.
\label{theo:EH_only}
\end{theorem}
\begin{proof}
The proof is provided in Appendix \ref{App:proof_max_WPT}.
\end{proof}
\begin{remark}
{\color{black}Theorem \ref{theo:EH_only} indicates that for $A\!<\! A_{{\rm T,sat}_l}$, the larger the peak amplitude $A$ in problem (\ref{eq:energy_problem}) is, the higher the maximum average power harvested by EH Rx $l$ given in (\ref{eq:max_EH_analytic}),  since $P_l(x)$ and $P_l(x)/x^2$ increase monotonically for $0\!<\!x\!\leq A$ as shown in the proof of Lemma \ref{lem:monotonicity_Power_function_before_sat} and in Appendix \ref{App:proof_max_WPT}, respectively. However, increasing the peak amplitude $A$  beyond $A_{{\rm T,sat}_l}$ has no effect on the maximum average harvested power, which saturates to an asymptotic value of $P_{l,\max}(A_{{\rm T,sat}_l})$, $\forall A\!\geq\!A_{{\rm T,sat}_l}$, cf. (\ref{eq:max_EH_analytic}), and the optimal input distribution for maximum WPT  in (\ref{eq:PMF_max_EH}) saturates to the asymptotic on-off distribution $\dd F_0^{{\rm WPT}}(x,A_{{\rm T,sat}_l})$. This asymptotic behaviour is confirmed by the numerical results provided in Section \ref{s:numerical_results}, cf. Fig. \ref{fig:publication_figure_compare_with_Suboptimal_pout}.}
{\color{black} We also note that from the WPT perspective, the specific values of $p$ and $q$  in (\ref{eq:PMF_max_EH}) are irrelevant as long as they satisfy $p\!+\!q=\min\{\sigma^2/A_l'^2,1\}$. For WIT, the rate is maximized when $p\!=\!q\!=\!\min\{\sigma^2/(2A_l'^2),1/2\}$.}
\end{remark}

\begin{remark}
For the linear EH model, 
$P_l(x)/x^2$ is constant $\forall\, x$. From (\ref{eq:compare_dist_max_EH}) in Appendix \ref{App:proof_max_WPT},  any distribution satisfying the AP and PP constraints maximizes the harvested energy. Consequently, the optimal distribution for maximum WIT is also optimal for maximum WPT. Hence, for an AWGN channel with the linear EH model, a tradeoff between WIT and WPT does not exist. This result was stated in \cite[p. 5]{Varshney2008}.
\end{remark}

%% file: suboptimal_distribution.tex
Motivated by studying the extreme cases of WIT and WPT, we propose a suboptimal distribution for problem (\ref{eq:capacity_problem}) with one EH Rx $l$. This distribution  superimposes a truncated Gaussian distribution with mass points at $-A_l'$ and $A_l'$, i.e.,
 \begin{equation}
f_{\rm s}(x)\!\!=\!\!\begin{cases}
\!b\e^{-dx^2}\!\!+\!\!c\big[\delta(x\!+\!A_l')\!+\!\delta(x\!-\!A_l')\big],
\quad\! &\!\!\!\!-\!A_l'\!\!\leq\! x\!\leq \!\!A_l'\\
\!0, &\mathrm{otherwise},
\end{cases}
\label{eq:sub_opt_dist}
\raisetag{2.5\normalbaselineskip}
\end{equation}
where $b$ and $c$ are obtained to satisfy the AP constraint and the unit area condition of the pdf $f_{\rm s}(x)$. In particular,  $b\!=\! \frac{1-2c}{\sqrt{\frac{\pi}{d}}\mathrm{erf}\left(\sqrt{d}A_l'\right)}$ and $c\!\leq\!\left(\sigma_x^2\!-\!\frac{1}{2d}\!+\!\frac{A_l'\exp(-dA_l'^2)}{\sqrt{\pi d}\mathrm{erf}(\sqrt{d}A_l')}\right)\!\!\Big/\!\!\left(2A_l'^2\!-\!\frac{1}{d}\!+\!\frac{2A_l'\exp(-dA_l'^2)}{\sqrt{\pi d}\mathrm{erf}(\sqrt{d}A_l')}\right)$, where ${\rm erf}(x)\!=\!\frac{1}{\sqrt{\pi}}\int_{-x}^x\exp(-t^2)\dd t$ is the error function and $d$ is a design parameter with which the harvested power increases. Since (\ref{eq:sub_opt_dist}) superimposes the optimal distributions for WIT and WPT, it is expected to provide a close-to-optimal R-E tradeoff performance. This is confirmed by numerical evaluations in Section \ref{ss:numerical_evaluations_main}.

%% file: complex_signaling.tex
In this section, we extend problem (\ref{eq:capacity_problem}) to the complex domain. In particular, the transmit signal $x(t)\!=\!\sum_{k=-\infty}^{\infty}x[k]g(t\!-\!kT)$ is composed of complex-valued symbols $x[k]\! \definedas \!r[k]\e^{j\theta[k]}$, where $r[k]$ and $\theta[k]$ are the amplitude and phase of the transmit symbol $x[k]$, respectively. The channel fading gains for  the ID and EH Rxs are also complex-valued given by $h_{\rm I}\!=\!|h_{\rm I}|\e^{j\phi_{\rm I}}$ and $h_{{\rm E}_l}\!=\!|h_{{\rm E}_l}|\e^{j\phi_{{\rm E}_l}}$, respectively. 
Hence, assuming a rectangular pulse, the received signal at the EH Rx is given by $y_{{\rm E}_l}(t)\!=\!r[k]|h_{{\rm E}_l}|\e^{j(\theta[k]+\phi_{{\rm E}_l})}$ and  $y_{{\rm E}_l}^{\rm RF}(t)\!=\!\sqrt{2} r[k] |h_{{\rm E}_l}|\cos(2\pi f_{\rm c} t+\theta[k]+\phi_{{\rm E}_l})$, $kT\!-\!T/2\!<\!t\!\leq\!kT\!+\!T/2$, in the equivalent complex baseband and RF domains, respectively. Let $\V{r}$ and $\V{\theta}$ be the random variables, whose  realizations in time slot $k$ are $r[k]$ and $\theta[k]$, respectively, i.e., $\V{x}\!=\!\V{r}\e^{j\V{\theta}}$. {\color{black}Hence, from Section \ref{ss:rectenna_model}, the integral involved in the forward current in the first term on the LHS of (\ref{eq:1st_Ritz_condition}) for EH Rx $l$ reduces to  $I_0(\beta)=\frac{1}{T}\int_{T} \e^{\frac{v_{\rm a}(t)}{\eta V_{\rm T}}}\dd t=\frac{1}{T}\!\int_{T}\e^{B y_{{\rm E}_l}^{\rm RF}(t)}\dd t\!=\!\frac{1}{T}\!\int_{T}\e^{\sqrt{2}B |h_{{\rm E}_l}| \V{r} \cos(2\pi f_{\rm c} t+\V{\theta}+\phi_{{\rm E}_l})} \dd t \!=\!I_0\left(\sqrt{2}B |h_{{\rm E}_l}| \V{r} \right)$.}
This indicates that the power harvested at the EH Rx  does not depend on the phase of the received signal. Hence, similar to (\ref{eq:P_x_closedform}), with complex signaling, the harvested power at the $l^{\rm th}$ EH Rx can be approximated by {\color{black}$P_l(r)\definedas P_{\rm out}\left(\sqrt{2}B |h_{{\rm E}_l}| r\right)\Big|_{\rm approx.}=\min\left(\left[\frac{1}{a}W_0\left(a\e^aI_0\left(\sqrt{2}B |h_{{\rm E}_l}| r\right)\right)-1\right]^2 I_{\rm s}^2R_{\rm L},\frac{{B_{\rm v}}^2}{4R_{\rm L}}\right),$} and the AP constraint in (\ref{eq:capacity_problem}) can be written as $ \E[\V{r}^2]\!\leq\! \sigma^2$. 
At the ID Rx, the baseband transmission model $\V{y}\!=\!\V{x}h_{\rm I}\!+\!\V{n}$ can be written in polar coordinates as $\V{R}\e^{j\V{\psi}}\!=\!|h_I|\V{r}\e^{j(\V{\theta}+\phi_{\rm I})}\!+\!\V{n}$, where $\V{R}$ and $\V{\psi}$ are random variables representing   the amplitude and phase of the received signal $\V{y}$ and $\V{n}\sim\mathcal{CN}(0,2\sigma_n^2)$.
\begin{lemma}\normalfont
The optimal distribution of transmit signal $\V{x}=\V{r}\e^{j\V{\theta}}$ for problem (\ref{eq:capacity_problem}) in the complex domain is characterized by mutually  \emph{independent} amplitude $\V{r}$ and phase $\V{\theta}$, and a  uniformly distributed phase $\V{\theta}$.
\label{lemma:indep_opt_amplitude_phase_UD}
\end{lemma}
\begin{proof}
The proof is provided in Appendix \ref{App:proof_indep_opt_amplitude_phase_UD}.
\end{proof}
Hence, with complex signalling, the conditional capacity problem in  (\ref{eq:capacity_problem}) reduces to finding the optimal distribution $F_{\V{r}}$ of the amplitude of the transmit signal based on the following optimization problem
\begin{IEEEeqnarray}{llll}
\!\!\!\!\!C\!=&\sup\limits_{F_{\V{r}}\in\mathcal{F}_A}&&I(F_{\V{r}}) \nonumber\\
&\,\,{\rm s.t.} &&\!\!{\rm C}_0:\E_{F_{\V{r}}}[\V{r}^2] \leq \sigma^2; \notag\\ 
& &&\!\!{\rm C}_l:\, \E_{F_{\V{r}}}\!\left[P_l(\V{r})\right] \geq P_{l,\req}, \,\, \forall\, l\in\mathcal{L}.
\label{eq:capacity_problem_Complex}
\end{IEEEeqnarray}
Next, we investigate the properties of the optimal input amplitude distribution  in the following theorem.
\begin{theorem}\normalfont
The optimal input amplitude distribution of problem (\ref{eq:capacity_problem_Complex}) is unique and discrete with finite number of mass points. Furthermore, if $E_0$ is the set of points of increase of a distribution function $F_{\V{r_0}}$ on $[0,A]$, then $F_{\V{r_0}}$ is the optimal input distribution if and only if there exist $\lambda_l\!\geq\!0$, $\forall\,l\!\in\!\{0\}\!\cup\!\mathcal{L}$, such that
\begin{equation}
\begin{aligned}
&s(r)\!\definedas\!\lambda_0\!\left(r^2\!-\!\sigma^2\right)\!-\!\sum\limits_{l\in\mathcal{L}}\! \!\lambda_l\!\left(P_l(r)\!\!-\!\!P_{l,\req}\right)\!+\!C+\!\log_2(\e\sigma_n^2)
\\
&+\! \!\int\!\!\! \frac{R}{\sigma_n^2}\e^{-\frac{R^2\!+\!r^2|h_{\rm I}|^2}{2\sigma_n^2}}\!I_0\!\left(\!\frac{Rr|h_{\rm I}|}{\sigma_n^2}\!\right)\log_2\!\!\left(\!\frac{f_{\V{R}}(\!R;F_{\V{r_0}}\!)}{R}\!\right)\! \dd R\!\geq\! 0,
\end{aligned}
\label{eq:C_necessary_sufficient_condition2_complex}
\end{equation} 
$\forall r\in[0,A]$, where equality holds if $r$ is a point of increase of $F_{\V{r_0}}$, i.e., if $r\in E_0$.
\label{theo:C_necessary_sufficient_condition2_complex}
\end{theorem}
\begin{proof}
The proof is provided in Appendix \ref{App:proof_C_necessary_sufficient_condition2_complex}.
\end{proof}

%% file: numerical_results.tex
 \begin{table*}[!tp]
\caption{Numerical parameters.}
\begin{tabular}{@{}ll@{}}  \toprule   
Parameter & Value \\ \midrule 
Carrier frequency  & $f_{\rm c}=2.45\,$GHz\\ 
Path loss exponent & $\alpha=2.5$ \\ 
Noise power at the EH Rxs &  $\sigma_n^2\!=\!-80\,$dBm in Figs. \ref{fig:RE_region_diff_EH_diff_PP_diff_AP}-\ref{fig:publication_figure_discrete_dist_real_signaling} and  $\sigma_n^2\!=\!-50\,$dBm per signal dimension 
in Figs. \ref{fig:publication_figure_complex_problem_multiple_users_pout}, \ref{fig:publication_figure_complex_Three_users_one_effective}.\\ 
Distance between transmitter and ID Rx & $d_{\rm I}\!=\!25\,$m\\ \addlinespace[0.5em]
Distance between transmitter and EH Rxs & \pbox{20cm}{In Figs. \ref{fig:RE_region_diff_EH_diff_PP_diff_AP}-\ref{fig:publication_figure_discrete_dist_real_signaling}, one EH Rx at $d_{{\rm E}_1}\!=\!5\,$m. \\
In Figs. \ref{fig:publication_figure_complex_problem_multiple_users_pout}, \ref{fig:publication_figure_complex_Three_users_one_effective}, three EH Rxs at distances $d_{{\rm E}_1}\!=\!3\,$m, $d_{{\rm E}_2}\!=\!3.5\,$m, and $d_{{\rm E}_3}\!=\!4\,$m}\\ \addlinespace[0.5em]
Circuit parameters, cf. Fig. \ref{fig:rectenna_model} &  $R_{\rm ant}\!=\!50\,\Omega$, $R_{\rm L}\!=\!10\,$k$\Omega$, $C_{\rm L}\!=\!1\,$nF\\
SMS7630 Schottky diode parameters \cite{Skyworks_SMS7630}& $I_{\rm s}\!=\!5\,\mu$A, $R_{\rm s}\!=\!20\Omega$, $\eta\!=\!1.05$, $C_{\rm j_0}\!=\!0.14\,$pF, {\color{black}$I_{B_{\rm v}}\!=\!100\,\mu$A, and ${B_{\rm v}}\!=\!2\,$V.}\\
\bottomrule
\end{tabular}
\label{tab:numerical_parameters}
\end{table*}

In this section, we first validate the accuracy of the derived harvested DC power functions in (\ref{eq:Pout_Sat_exact}) and (\ref{eq:Pout_Sat_approx}) via circuit simulations. 
Afterwards, we evaluate the solutions for problems (\ref{eq:capacity_problem}) and (\ref{eq:capacity_problem_Complex}) for  real and complex AWGN channels, respectively, under AP, PP, and EH constraints. 
The channel gains are given by $|h_k|^2\!=\!\left(v/(4\pi d_k f_{\rm c})\right)^\alpha$ for  $k\!\in\!\{{\rm I},{\rm E}_l\}$, where $v$ is the speed of light, $\alpha$ is the path loss exponent, $d_{\rm I}$ and $d_{{\rm E}_l}$ are the distances between the transmitter and the ID  and the $l^{\text{th}}$ EH Rx, respectively. Table \ref{tab:numerical_parameters} summarizes the parameters adopted in the numerical results. 
\begin{figure*}[!t]
\centering
\includegraphics[width=0.8\textwidth]{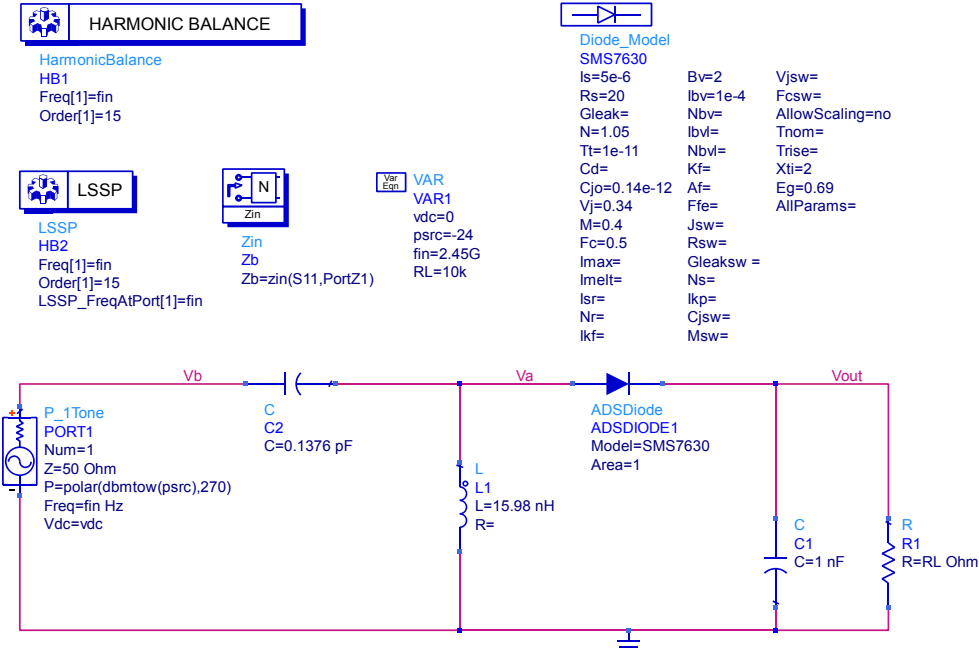}\vspace{-0.3cm}
\caption{\color{black}ADS schematic of the rectenna circuit model in Fig. \ref{fig:rectenna_model}.}
\label{fig:rectenna_ct_ADS}
\end{figure*}

\subsection{ADS Circuit Simulation and Validation of the Harvested DC Power Function $P(x)$ in  Table \ref{tab:Sat_model_summary}}
\label{ss:ADS_circuit_simulations}

In this section, we validate the harvested DC power function given in  Table \ref{tab:Sat_model_summary} through circuit simulations on ADS \cite{ADS}, as shown in  Fig. \ref{fig:rectenna_ct_ADS}. In particular, we use the SMS7630 Schottky diode, since it operates at very low input RF powers\footnote{\color{black}In the  SMS7630 Schottky diode's data sheet \cite{Skyworks_SMS7630}, the diode detector circuit shown in  \cite[Fig. 2]{Skyworks_SMS7630}, which is similar to the one considered in this paper, is functional for RF input powers as low as $-40\,$dBm as is evident from the measured output voltage shown in  \cite[Fig. 7]{Skyworks_SMS7630}.}  and does not need external bias \cite{Skyworks_SMS7630}. An LC matching network is fine-tuned for every input power to provide perfect matching (reflection coefficient {\color{black} $\!<\!-50\,$dB)}. For example, at $P_{\rm in}\!=\!-24\,$dBm, the matching network elements are {\color{black}$L\!=\!15.98\,$nH} and {\color{black}$C\!=\!0.1376\,$pF}, cf. Fig. \ref{fig:rectenna_ct_ADS}. The remaining circuit parameters are as in Table \ref{tab:numerical_parameters}.   
Fig. \ref{fig:Harvested_DC_power_ct_theory} shows a very good match between the harvested DC power obtained from the circuit simulations and the analytical expressions from  Table \ref{tab:Sat_model_summary}.

\begin{figure}[!t]
\centering
\includegraphics[width=0.48\textwidth]{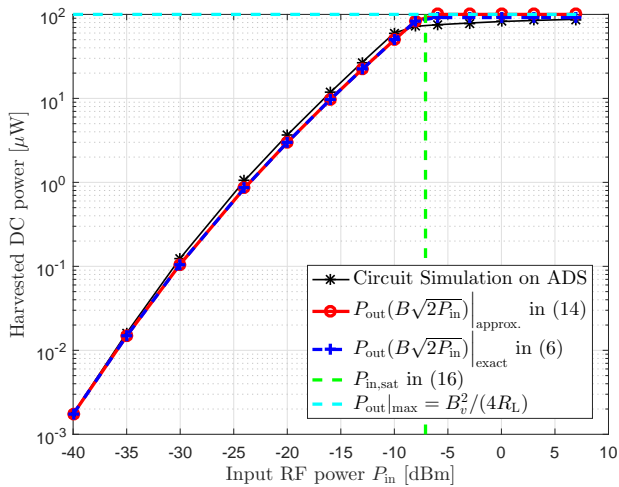}
\caption{\color{black}Harvested DC power vs. input RF power: ADS circuit simulation and analytical results from Table \ref{tab:Sat_model_summary}.\\[-1ex]}
\label{fig:Harvested_DC_power_ct_theory}
\end{figure}

\subsection{Numerical Evaluation of the Conditional Capacity Problems (\ref{eq:capacity_problem}) and (\ref{eq:capacity_problem_Complex})} 
\label{ss:numerical_evaluations_main}
Although we showed in Sections \ref{sss:discrete_dist} and \ref{s:complex_signaling} that the optimal input distribution is discrete with a finite number of  mass points, the number and positions of the mass points are not known. However, since problems (\ref{eq:capacity_problem}) and (\ref{eq:capacity_problem_Complex}) are convex $\forall \, F\!\in\! \mathcal{F}_A$, they can be solved numerically using CVX \cite{CVX} by discretizing the interval $x\!=\![-A,A]$ with a sufficiently small step size, i.e.,  $\Delta x\!\to\! 0$ and the interval $r\!=\![0,A]$ with $\Delta r\!\to\! 0$, to obtain the symbol set. {\color{black} Then, for this symbol set, the harvested power functions $P_l(x)$ in (\ref{eq:P_x_closedform}) and  $P_l(r)$ in Section \ref{s:complex_signaling} are calculated and used in the EH constraints in CVX.} The optimality of the numerically obtained input distribution can be checked by verifying the necessary and sufficient conditions in Corollary \ref{corol:C_necessary_sufficient_condition2} for real signaling and  in Theorem \ref{theo:C_necessary_sufficient_condition2_complex} for complex signaling.

\begin{figure}[!t]
\centering
\includegraphics[width=0.48\textwidth]{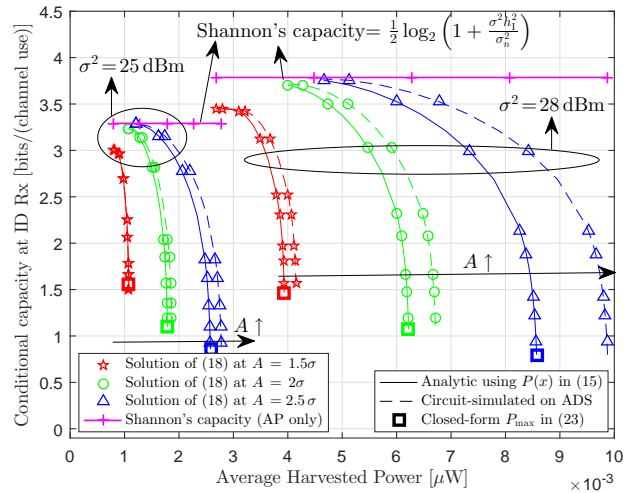}
\caption{R-E regions for different AP and PP  constraints.}
\label{fig:RE_region_diff_EH_diff_PP_diff_AP}
\end{figure}
{\color{black}In Figs. \ref{fig:RE_region_diff_EH_diff_PP_diff_AP}-\ref{fig:publication_figure_discrete_dist_real_signaling}, we consider a real-valued AWGN channel and a SWIPT system with one ID Rx and only one EH Rx located at $d_{{\rm E}_1}\!=\!5\,$m,  i.e., with $A_{{\rm T,sat}_1}\!=\!\sqrt{P_{\rm in,sat}}/|h_{{\rm E}_1}|\!=\!33.28\,$V$=23.56\,\sigma\,$.} {\color{black}Except for Fig. \ref{fig:publication_figure_compare_with_Suboptimal_pout}, all peak amplitudes $A$ in Figs. \ref{fig:RE_region_diff_EH_diff_PP_diff_AP}-\ref{fig:publication_figure_discrete_dist_real_signaling} are less than $A_{{\rm T,sat}_1}$ to avoid driving the rectifier into saturation.}

\begin{figure}[!t]
\includegraphics[width=0.48\textwidth]{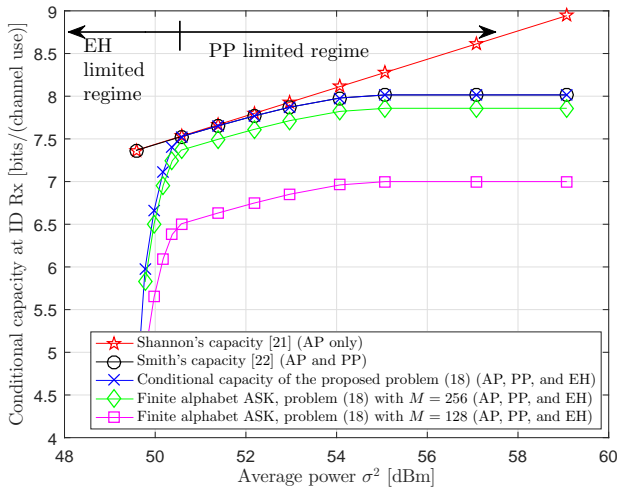}
\caption{Conditional capacity of problem (\ref{eq:capacity_problem}) and achievable rate  for different finite alphabets with $A\!=\!30.75\,$V and $P_{\req}=7.5\,\mu$W.}
\label{fig:Different_AP_optimal_finite_Shannon_Smith}
\end{figure}
 In Fig. \ref{fig:RE_region_diff_EH_diff_PP_diff_AP}, we plot the R-E curves of the considered system for different AP and PP constraints. In particular, we obtain each R-E curve by solving problem (\ref{eq:capacity_problem}) for different required harvested powers $P_{\rm req}$. The optimal input distribution $F_0(x)$ is then used to obtain the conditional capacity at the ID Rx, $I(F_0)$, and the average harvested power  at the EH Rx, $\E_{F_0}[P_1(\V{x})]\!=\!\int_{-A}^{A}P_1(x) \dd F_0(x)$, with $P_1(x)$ in (\ref{eq:P_x_closedform}). In addition, we plot the circuit-simulated average harvested power $\E[P_{\rm ct}(\V{x})]\!=\!\int_{-A}^{A}P_{\rm ct}(x) \dd F_0(x)$, where $P_{\rm ct}(x)\!=\!P_{\rm ct}\left(P_{\rm in}\!=\!(xh_{\rm E})^2\right)$ is the power function shown in Fig. \ref{fig:Harvested_DC_power_ct_theory} obtained by interpolating the circuit-simulated data from ADS. In Fig. \ref{fig:RE_region_diff_EH_diff_PP_diff_AP}, it is observed that the circuit-based average harvested DC power is slightly higher than the analytical one. This is because, as shown in Fig. \ref{fig:Harvested_DC_power_ct_theory}, the circuit-simulated harvested power function lies slightly above the analytical one. Moreover,  Fig. \ref{fig:RE_region_diff_EH_diff_PP_diff_AP} reveals that unlike for the linear EH model \cite{Varshney2008}, for the considered nonlinear EH model, there exists a non-trivial tradeoff between the information rate transmitted to the ID Rx and the power delivered to the EH Rx. In fact, for a larger required average DC power, the optimal input distribution forces the transmitter to transmit more often with the peak amplitudes $x=\pm A$ and less often in the range  $x\in(-A,A)$. This leads to higher average harvested power for the EH Rx at the expense of a lower information rate at the ID Rx. Moreover, the maximum feasible average harvested DC power, obtained by solving problem (\ref{eq:energy_problem}), matches the closed-form expression in (\ref{eq:max_EH_analytic}). Furthermore, it can be observed that the higher the peak-amplitude $A$, the larger the achieved R-E region. In deed, for a larger peak amplitude, the transmitter has to transmit less often with the peak amplitudes to achieve the same average harvested power and can more often choose $x\in(-A,A)$, allowing for a higher information rate at the ID Rx. Moreover, as $A$ increases, the maximum possible average harvested power increases, as given in (\ref{eq:max_EH_analytic}) {\color{black}for $A<A_{{\rm T,sat}_1}$.}

\begin{figure}[!t]
\centering
\includegraphics[width=0.48\textwidth]{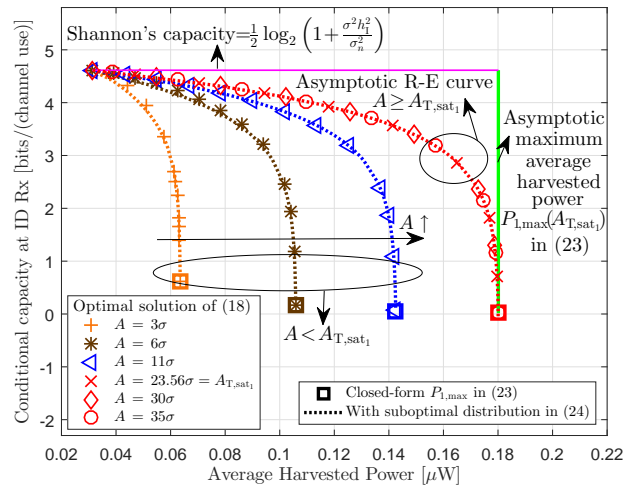}
\caption{{\color{black}R-E curves for optimal distribution from (\ref{eq:capacity_problem}) and  suboptimal distribution from (\ref{eq:sub_opt_dist}) for $\sigma^2\!=\!33\,$dBm.}}
\label{fig:publication_figure_compare_with_Suboptimal_pout}
\end{figure}

{\color{black}In Fig. \ref{fig:Different_AP_optimal_finite_Shannon_Smith}, we 
plot the conditional capacity of problem (\ref{eq:capacity_problem})  as a function of the AP constraint $\sigma^2$ for a required DC power of $7.5\,\mu$W. The peak amplitude constraint is set to $A\!=\!30.75\,$V, which results in a received RF peak power of $(Ah_{\rm E_1})^2\!\equiv\!-8\,$dBm, i.e., the rectifier is not driven into saturation, cf. Fig. \ref{fig:Harvested_DC_power_ct_theory}.} Fig.  \ref{fig:Different_AP_optimal_finite_Shannon_Smith} reveals that, for low APs, the system is EH-limited. In particular, compared to Smith's problem in \cite{SMITH19712} with AP and PP constraints, the EH constraint imposed in problem (\ref{eq:capacity_problem}) incurs a capacity loss which decreases with the AP. On the other hand, for large APs, the system is PP limited. That is, the EH constraint is inactive and the conditional capacity of our problem coincides with Smith's capacity in \cite{SMITH19712}. In addition, we plot the maximum information rate for amplitude shift keying (ASK), which is obtained by solving problem (\ref{eq:capacity_problem}) for $M$ symbols at $x\!=\!\frac{2Ak}{M-1}\!-\!A,\,\,k\!=\!0,1,\dots,M\!-\!1$. The larger the alphabet size, the closer the rate achieved by the finite alphabet is to that achieved by the optimal input distribution. Moreover, in the PP-limited regime, the capacities of all PP-constrained schemes saturate for high APs.

{\color{black}In Fig. \ref{fig:publication_figure_compare_with_Suboptimal_pout}, we show the R-E curves obtained for the optimal solution of problem (\ref{eq:capacity_problem}) and the suboptimal distribution given in (\ref{eq:sub_opt_dist}), respectively.  The AP constraint is set to $\sigma^2\!=\!33\,$dBm and different values of the PP constraint are considered. It is observed that the R-E curves for the suboptimal distribution are very close to the optimal ones. This behaviour is expected since the suboptimal distribution is a weighted sum of the optimal distribution for maximum WIT with $A\!\to\!\infty$ and the optimal distribution for maximum WPT, cf. Sections \ref{ss:Max_Info_transfer} to \ref{ss:suboptimal_distribution}. Moreover, interestingly, it is observed that all R-E curves for peak amplitudes $A\!\geq\!A_{{\rm T,sat}_1}$ are identical. This is because $\forall\,A\!\geq\!A_{{\rm T,sat}_1}$,  the optimal input distribution for maximum WPT saturates at the asymptotic on-off distribution $\dd F_0^{{\rm WPT}}(x,A_{{\rm T,sat}_1})$ in (\ref{eq:PMF_max_EH}) and the corresponding maximum average harvested power saturates at $P_{1,\max}(A_{{\rm T,sat}_1})$ in (\ref{eq:max_EH_analytic}). At the same time, since $A_{{\rm T,sat}_1}\!=\!23.56\sigma$ is large compared to $\sigma$, $\forall\,A\!\geq\!A_{{\rm T,sat}_1}$ the optimal input distribution for maximum WIT converges to that for $A\!\to\!\infty$, namely to the asymptotic zero-mean Gaussian distribution. Hence, the optimal distribution that maximizes the R-E region also converges to an asymptotic distribution (very close to the suboptimal distribution in (\ref{eq:sub_opt_dist}) with $A_1'\!=\!A_{{\rm T,sat}_1}$), which yields the asymptotic R-E curve shown in red in Fig. \ref{fig:publication_figure_compare_with_Suboptimal_pout}.} 
\begin{figure}[!t]
\includegraphics[width=0.48\textwidth]{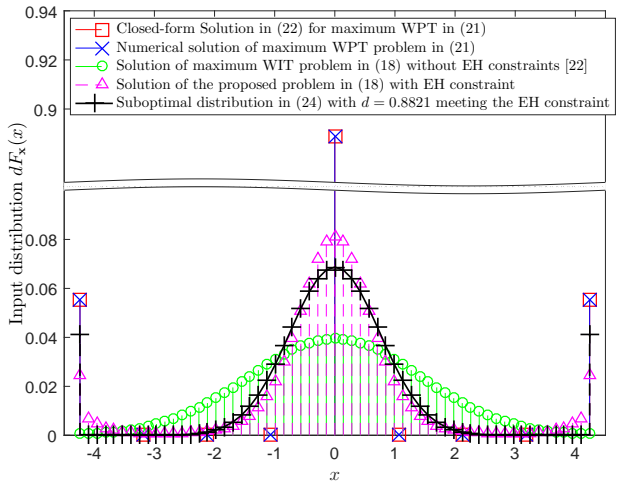}
\caption{ Numerically-obtained and closed-form input distributions for the SWIPT problem in (\ref{eq:capacity_problem}) and the maximum WPT solution for $\sigma^2\!=\!33\,$dBm, $A\!=\!3\sigma$, and $P_{\rm req}\!=\!0.047\mu$W.}
\label{fig:publication_figure_discrete_dist_real_signaling}
\end{figure}

In Fig. \ref{fig:publication_figure_discrete_dist_real_signaling}, we assume an AP constraint of  $\sigma^2\!=\!33\,$dBm and a peak amplitude constraint of $A\!=\!3\sigma$. We plot the numerically-obtained optimal distributions for (a) the maximum WIT problem  \cite{SMITH19712}, which has the shape of a truncated Gaussian distribution (b) the maximum WPT problem in (\ref{eq:energy_problem}), which perfectly matches the closed-form optimal distribution in (\ref{eq:PMF_max_EH}) with mass points at $0$ and $\pm A$, and (c) the SWIPT problem in (\ref{eq:capacity_problem}), whose envelope is close to a truncated Gaussian with additional mass points at $\pm A$. This explains why in Fig. \ref{fig:publication_figure_compare_with_Suboptimal_pout}, the suboptimal distribution in (\ref{eq:sub_opt_dist}) leads to a close-to-optimal R-E performance.

{\color{black}In Figs. \ref{fig:publication_figure_complex_problem_multiple_users_pout} and \ref{fig:publication_figure_complex_Three_users_one_effective}, we consider complex-valued transmission and solve problem (\ref{eq:capacity_problem_Complex}) for an AP of  $\sigma^2\!=\!43\,$dBm and a  peak amplitude of $A\!=\!3\sigma=\!13.4\,$V. 
We consider a system with one ID Rx at $d_{\rm I}\!=\!25\,$m and three EH Rxs at  $d_{{\rm E}_1}\!=\!3\,$m, $d_{{\rm E}_2}\!=\!3.5\,$m, and $d_{{\rm E}_3}\!=\!4\,$m,  i.e., with  $A_{{\rm T,sat}_1}\!=\!18\,$V,  $A_{{\rm T,sat}_2}\!=\!21.8\,$V, and  $A_{{\rm T,sat}_3}\!=\!25.8\,$V, respectively. Hence, $A\!<\!A_{\rm sat}$ and the results in Theorem \ref{theo:Only_one_EH_Rx_effective} hold.} In Fig. \ref{fig:publication_figure_complex_problem_multiple_users_pout}, we first consider the case when only one EH Rx requires a certain amount of DC power, while the other two EH Rxs have no power demands and they passively  harvest from the received power. In this case, the EH constraint of only the power-demanding EH Rx is present in (\ref{eq:capacity_problem_Complex}) and the corresponding individual R-E curves are plotted in Fig. \ref{fig:publication_figure_complex_problem_multiple_users_pout}. It can be observed that the closer the EH Rx is to the transmitter, the larger the R-E region gets.  Furthermore, at low required DC powers, the EH constraint of the power-demanding EH Rx is inactive and the individual R-E curves converge to the capacity of the complex AWGN channel with AP and PP constraints only, as obtained in \cite{QAGC_Shamai_Bar_David_1995}. For the considered low AP constraint,  this limiting capacity practically coincides with Shannon's capacity given by $\log_2\left(\!1\!+\!\frac{\sigma^2|h_{\rm I}|^2}{2\sigma_n^2}\!\right)\!=\!2\,$bits/(channel use).

\begin{figure}[t]
\centering
\includegraphics[width=0.48\textwidth]{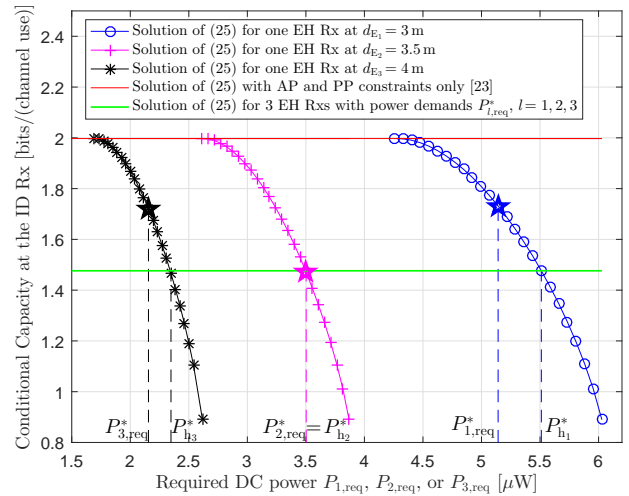}
\caption{Individual R-E curves for three EH Rxs for $\sigma^2\!=\!43\,$dBm, $A\!=\!3\sigma$. The star markers represent the intersection between the individual required DC powers with the individual R-E curves.}
\label{fig:publication_figure_complex_problem_multiple_users_pout}
\end{figure}
\begin{figure}[!t]
\includegraphics[width=0.48\textwidth]{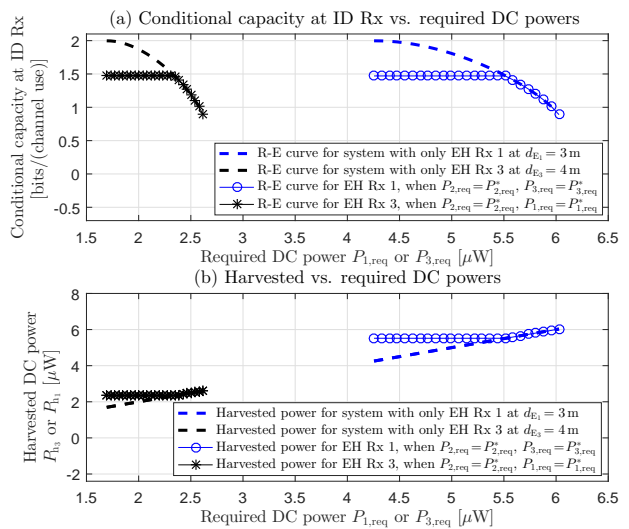}
\caption{ R-E curves and average harvested DC powers  for $\sigma^2\!=\!43\,$dBm, $A\!=\!3\sigma$ using the required DC powers with star markers in Fig. \ref{fig:publication_figure_complex_problem_multiple_users_pout} for two EH Rxs and varying the required DC power of only one EH Rx ($1$ or $3$).}
\label{fig:publication_figure_complex_Three_users_one_effective}
\end{figure}

Next, we consider the case when all three EH Rxs require a certain DC power, given by $P_{l,\req}^*$, $l\!=\!1,2,3$,  as shown in Fig.  \ref{fig:publication_figure_complex_problem_multiple_users_pout} by the projection of the star markers \quotes{\large $\star$\normalsize} onto the x-axis.  According to Theorem \ref{theo:Only_one_EH_Rx_effective}, the only active EH constraint, is that of the EH Rx for which the ID Rx rate of its individual R-E curve is the smallest. In the considered example, the EH constraint of EH Rx $2$ is the only active one and the conditional capacity at the ID Rx is $1.47\,$bits/(channel use). Furthermore, the actual harvested DC powers at the EH Rxs, denoted by $P_{{\rm h}_l}^*$, $l\!=\!1,2,3$, are the DC power values of the points obtained by the intersection of the individual R-E curves with the horizontal line of the conditional capacity of the ID Rx.  Only the active EH Rx harvests as much as it requires, while the inactive ones  harvest more power than required, i.e., $P_{{\rm h}_1}^*\!>\!P_{1,\req}^*$, $P_{{\rm h}_2}^*\!=\!P_{2,\req}^*$, and $P_{{\rm h}_3}^*\!>\!P_{3,\req}^*$, as shown in Fig. \ref{fig:publication_figure_complex_problem_multiple_users_pout}. 

To shed further light on the behaviour of the SWIPT system with multiple EH Rxs, in Fig. \ref{fig:publication_figure_complex_Three_users_one_effective}, we sweep the required DC power of either EH Rx $1$ or $3$, respectively, and fix the required DC powers of the other two EH Rxs to the values given by the star markers  \quotes{\large $\star$\normalsize} in Fig. \ref{fig:publication_figure_complex_problem_multiple_users_pout}, namely,  $P_{2,\req}^*$ and $P_{l,\req}^*$, $l=3,1$. Then, we plot the conditional capacity achieved at the ID Rx as well as the harvested powers at EH Rx $l\!\in\!\{1,3\}$. As can be observed, for all required DC powers $P_{l,\req}\!<\!P_{{\rm h}_l}^*$, the EH constraint of EH Rx $2$ is the only active EH constraint and the conditional  capacity at the ID Rx is fixed, whereas the harvested DC power at user $l$ is $P_{{\rm h}_l}^*$. On the other hand, when $P_{l,\req}\!>\!P_{{\rm h}_l}^*$, EH Rx $l$ provides the only active EH constraint and the conditional capacity of the ID Rx is determined by the individual R-E curve of EH Rx $l$, as shown in Fig. \ref{fig:publication_figure_complex_Three_users_one_effective}(a). Moreover, EH Rx $l$ harvests as much as it requires, i.e., $P_{{\rm h}_l}\!=\!P_{l,\req}$.

%% file: conclusion.tex
We studied the conditional capacities of real- and complex-valued SWIPT systems with separated ID and EH Rxs under AP, PP, and EH constraints. {\color{black}We developed a novel circuit-based nonlinear EH model that accounts for the saturation of the harvested DC powers at high input RF powers. The accuracy of this model was verified with circuit simulations.}  {\color{black}Our results reveal that, for a given AP constraint, the R-E tradeoff curve saturates for  high PP constraints due to the saturation of the harvested DC power.} We proved that the optimal input distribution  that maximizes the R-E region is discrete with a finite number of mass points. Moreover, the optimal input distribution for maximum WPT was found to be an on-off distribution. We proposed a suboptimal distribution, which superimposes the optimal distributions for WPT and WIT and showed that its R-E performance closely approaches the optimal one. Future work  may include extensions to multisine signals, fading channels, and co-located EH and ID Rxs. 

%% file: proof_unique_optimal_distribution.tex
\section{Proof of Theorem \ref{theo:unique_distribution}}
\label{App:proof_unique_optimal_distribution}
We first prove the existence of a unique distribution $F_0 \in\Omega$ that maximizes the mutual information $I(F)$. It suffices to show that the optimization problem in (\ref{eq:capacity_problem}) is convex, i.e., that the set $\Omega$ is convex and compact in some topology and that $I(F)$ is continuous and strictly concave in $F$.  The convexity of the set $\Omega$ follows from the convexity of the set of distribution functions $\mathcal{F}_A$ (defined by $\int_{-A}^A \dd F(x)=1$) and the linearity of the AP and  EH constraints in $F$.  Hence, constraints $g_l(F)\leq 0$, $\forall\,l\in\{0\}\cup\mathcal{L}$, are convex. 
The proof of the compactness of $\Omega$ is similar to that in \cite[Appendix I.A]{capacity_Rayleigh_fading_Abou_Faycal2001}. Next, we show that the mutual information is continuous and strictly concave in $F$. The mutual information resulting from an input distribution $F$ is $I(F)\!=\!h_{\V{y}}(F)\!-\!h_N$, where $h_{\V{y}}(F)$ is the entropy of  output $\V{y}$ assuming an input distribution $F$,  and $h_N$ is the noise entropy given by $h_N\!=\!\frac{1}{2}\log_2(2\pi\e\sigma_n^2)$. Since $h_N$ is constant, it suffices to show that $h_{\V{y}}(F)$ is continuous and strictly concave. The proof of the continuity of $h_{\V{y}}(F)$ is given in \cite[Appendix I.B]{capacity_Rayleigh_fading_Abou_Faycal2001}. Next, we show that the entropy function $h_{\V{y}}(F)$ is strictly concave in $F$. Since $h_{\V{y}}(F)\!=\!-\int_{-\infty}^{\infty}p(y; F)\log_2(p(y; F)) \dd y$ is a strictly concave function of the output pdf $p(y;F)$ and $p(y;F)\!=\!\int_{-\infty}^{\infty}p(y|x)\dd F(x)$ is a linear function in $F$, it follows that $h_{\V{y}}(F)$ is a strictly concave function in $F$. Hence, we conclude that problem (\ref{eq:capacity_problem}) is convex and has a unique solution. 

Next, the proof that $C$ in (\ref{eq:capacity_problem}) is also given by $C\!=\!\sup\limits_{F\in\mathcal{F}_A} I(F)\!-\!\sum_{l\in\{0\}\cup\mathcal{L}} \lambda_l g_l(F)$ follows from the Lagrangian theorem for constrained optimization problems. In particular, this equivalence holds for the convex problem in (\ref{eq:capacity_problem}) if $C$ is finite and Slater's condition holds, i.e., if there exists an interior point $\tilde{F}\!\in\!\mathcal{F}_A$ such that all constraints hold with strict inequality, i.e., $g_l(\tilde{F})\!<\!0$, $\forall\,l\in\{0\}\cup\mathcal{L}$. The finiteness of the conditional capacity $C$ is guaranteed by the AP  constraint. Next, we prove that for the considered problem, Slater's condition holds. Let $\tilde{x}$ satisfy $|\tilde{x}|\!<\!\sigma\!<\!A$ and $P_l(\tilde{x})\!>\!P_{l,\req}$, $\forall\,\,l\in\mathcal{L}$, and let $\tilde{F}$ be the unit-step function at $\tilde{x}$, then $g_0(\tilde{F})=\tilde{x}^2\!-\!\sigma^2\!<\!0$ and $g_l(\tilde{F})\!=\!-P_l(\tilde{x})\!+\!P_{l,\req}\!<\!0$, $\forall\,\,l\in\mathcal{L}$, and hence, Slater's condition holds. Thus, from the Lagrangian theorem, strong duality holds, i.e., $\exists\,\lambda_l\!\geq\! 0$, $\forall\,l\in\{0\}\cup\mathcal{L}$,  such that $C\!=\!\sup\limits_{F\in\mathcal{F}_A} I(F)-\sum_{l\in\{0\}\cup\mathcal{L}} \lambda_l g_l(F)$ and is achieved also by $F_0$. Moreover, the complementary slackness conditions $\lambda_lg_l(F_0)\!=\!0$ must hold  $\forall\,\,l\in\{0\}\cup\mathcal{L}$. This completes the proof. 

%% file: proof_C_necessary_sufficient_condition1.tex
\section{Proof of Theorem \ref{theo:C_necessary_sufficient_condition1}}
\label{App:proof_C_necessary_sufficient_condition1}
Define $J(F)\!\definedas\!I(F)\!-\!\sum_{l\in\{0\}\cup\mathcal{L}} \lambda_l g_l(F)$, then from Theorem \ref{theo:unique_distribution}, $C$ can be written as $C\!=\!\sup_{F\in\mathcal{F}_A} \, J(F)$.  From  \cite[Theorem 3]{capacity_Rayleigh_fading_Abou_Faycal2001}, if $\mathcal{F}_A$ is convex, and $J(F)$ is concave and weakly differentiable, then $J'_{F_0}(F)\!\leq \!0$ is a necessary and sufficient condition for $J(F)$ to achieve its maximum at $F_0$, where $J'_{F_0}(F)\!\definedas\!\lim\limits_{\theta\to 0}\left(J((1\!-\!\theta)F_0\!+\!\theta F)\!-\!J(F_0)\right)\!/\theta$ is the weak derivative of $J(F)$ at $F_0$. In Appendix \ref{App:proof_unique_optimal_distribution}, we established that $\mathcal{F}_A$ is convex and that $J(F)$ is strictly concave in $F$, since $I(F)$ is strictly concave in $F$ and $g_l(F)$ is affine in $F$  $\forall\,l\in\{0\}\cup\mathcal{L}$. It remains to be proved that $J(F)$ is weakly differentiable and to determine the derivative $J'_{F_0}(F)\!=\! I'_{F_0}(F)\!-\!\sum_{l\in\{0\}\cup\mathcal{L}} \lambda_l g'_{l,F_0}(F)$. In \cite[Proof of Theorem 3]{capacity_Rayleigh_fading_Abou_Faycal2001}, it is shown that $I'_{F_0}(F)$ exists and is given by $I'_{F_0}(F)\!=\!\int i(x;F_0)\dd F(x)\!-\!I(F_0)$. It is also shown that for any linear constraint function $g_l(F)$, the derivative is $g'_{l,F_0}(F)\!=\!g_l(F)\!-\!g_l(F_0)$. From complementary slackness, $\lambda_l g_l(F_0)\!=\!0$ $\forall\, l$. Hence, the condition $J'_{F_0}(F)\leq 0$ for the optimality of $F_0$ is $\int i(x;F_0)\dd F(x)\!-C-\!\sum_{l\in\{0\}\cup\mathcal{L}} \lambda_l g_l(F) \!\leq\! 0$, which reduces to (\ref{eq:C_necessary_sufficient_condition1}). This completes the proof.

%% file: proof_C_necessary_sufficient_condition2.tex
\section{Proof of Corollary \ref{corol:C_necessary_sufficient_condition2}}
\label{App:proof_C_necessary_sufficient_condition2}
We start with condition (\ref{eq:C_necessary_sufficient_condition1}) which guarantees the optimality of $F_0$. From Appendix \ref{App:proof_C_necessary_sufficient_condition1},  (\ref{eq:C_necessary_sufficient_condition1}) can be written as $\int \!i(x;F_0)\dd F(x)\! -\! C\!-\!\sum_{l\in\{0\}\cup\mathcal{L}}\lambda_lg_l(F)\!\leq \!0$. Define $g_l(F)\!=\!\int A_l(x)\dd F(x)\! -\!a_l$, $l\in\{0\}\cup\mathcal{L}$. Hence, $A_0(x)\!=\!x^2$, $a_0\!=\!\sigma^2$, $A_l(x)\!=\!-P_l(x)$, and $a_l\!=\!-P_{l,\req}$, for $l\in\mathcal{L}$. Thus,  (\ref{eq:C_necessary_sufficient_condition1}) can be written as 
\begin{equation}
\int\! \Big(i(x;F_0)-\!\!\!\sum_{l\in\{0\}\cup\mathcal{L}}\!\!\!\lambda_lA_l(x)\Big)\dd F(x) \leq C-\!\!\!\sum_{l\in\{0\}\cup\mathcal{L}}\!\!\!\lambda_la_l.
\label{eq:necessary_condition1_general}
\end{equation} 
Next, we prove that (\ref{eq:necessary_condition1_general}) holds if and only if
\begin{equation}
\textstyle{\quad \,\,i(x;F_0)\!\leq\! C\!+\!\!\!\!\sum\limits_{l\in\{0\}\cup\mathcal{L}}\!\!\!\lambda_l\left(A_l(x)\!-\!a_l\right), \quad \forall x\in[-A,A]},\\[-1ex]
\label{eq:necessary_condition2_general1}
\end{equation}
and
\begin{equation}
\textstyle{i(x;F_0)= C+\sum\limits_{l\in\{0\}\cup\mathcal{L}}\lambda_l\left(A_l(x)-a_l\right), \quad \forall x\in E_0}.
\label{eq:necessary_condition2_general2}
\end{equation}
Clearly, if both conditions (\ref{eq:necessary_condition2_general1}) and (\ref{eq:necessary_condition2_general2}) hold, $F_0$ must be optimal because the necessary and sufficient condition in (\ref{eq:necessary_condition1_general}) is satisfied. The converse remains to be proved, i.e., if (\ref{eq:necessary_condition1_general}) holds, (\ref{eq:necessary_condition2_general1}) and (\ref{eq:necessary_condition2_general2}) must also hold. We prove this by contradiction. Assume that  (\ref{eq:necessary_condition1_general}) holds but  (\ref{eq:necessary_condition2_general1}) does not. It means that $\exists\, \tilde{x}\in[-A,A]$ such that $i(\tilde{x};F_0)> C+\sum_{l\in\{0\}\cup\mathcal{L}}\lambda_l\left(A_l(\tilde{x})-a_l\right)$. Now, let $F$ be the unit-step function at $\tilde{x}$, then the LHS of (\ref{eq:necessary_condition1_general}) becomes $i(\tilde{x};F_0)-\sum_{l\in\{0\}\cup\mathcal{L}}\lambda_lA_l(\tilde{x})>C-\sum_{l\in\mathcal{L}}\lambda_la_l$, which violates (\ref{eq:necessary_condition1_general}). Hence, if (\ref{eq:necessary_condition1_general}) holds, (\ref{eq:necessary_condition2_general1}) must also hold. Now, assume that (\ref{eq:necessary_condition1_general}) holds but (\ref{eq:necessary_condition2_general2}) does not. That is, for a subset of $E_0$ defined as $E'\!\subset \!E_0$, with positive measure, i.e., $\int_{E'}\dd F_0(x)\!=\!\delta\!>\!0$, (\ref{eq:necessary_condition2_general2}) does not hold. Then, from (\ref{eq:necessary_condition2_general1}), $i(x;F_0)\!< \!C\!+\!\sum_{l\in\{0\}\cup\mathcal{L}}\!\lambda_l\!\left(A_l(x)\!-\!a_l\right)\!,\, \forall\,x\!\in\! E'$. Since $F_0$ has points of increase on $E_0$ only, we have $\int_{E_0}\dd F_0(x)=\int_{E'}\dd F_0(x)\!+\!\int_{E_0-E'}\!\dd F_0(x)\!=\!\delta\!+\!(1\!-\!\delta)\!=\!1$. Now, we can write  $C\!-\!\sum_{l\in\{0\}\cup\mathcal{L}} \lambda_l a_l\!\leq \!I(F_0)\!-\!\sum_{l\in\{0\}\cup\mathcal{L}}\lambda_l\int\! A_l(x)\dd F_0(x)\!=\!\int\Big(i(x;F_0)\!-\!\sum_{l\in\{0\}\cup\mathcal{L}}\lambda_lA_l(x)\Big)\dd F_0(x),$ as
\begin{equation*}
\begin{aligned}
&C-\!\!\!\!\sum\limits_{l\in\{0\}\cup\mathcal{L}}\!\!\! \lambda_l a_l\!\leq\!\!\!\!\underbrace{\underset{x\in E'}{\int}\!\!\!\Big(i(x;F_0)\!-\!\!\!\!\sum\limits_{l\in\{0\}\cup\mathcal{L}}\!\!\!\!\lambda_l A_l(x)\Big)\dd F_0(x)}_{<\delta(C-\!\sum_{l\in\{0\}\cup\mathcal{L}} \lambda_l a_l)}\\
&+\!\!\!\underbrace{\underset{x\in E_0\!-\!E'}{\int}\!\!\!\Big(i(x;F_0)\!-\!\!\!\!\!\!\!\sum\limits_{l\in\{0\}\cup\mathcal{L}}\!\!\!\!\!\lambda_l A_l(x)\Big)\dd F_0(x)}_{=(1-\delta)(C-\!\sum_{l\in\{0\}\cup\mathcal{L}}\!\lambda_l a_l)}
<C-\!\!\!\!\!\!\sum\limits_{l\in\{0\}\cup\mathcal{L}}\!\!\!\!\!\!\lambda_l a_l,
\end{aligned}
\label{eq:expression_to_violate2}
\end{equation*}
which is a contradiction. Hence, if (\ref{eq:necessary_condition1_general}) holds, (\ref{eq:necessary_condition2_general2}) must also hold. Therefore, (\ref{eq:necessary_condition2_general1}) and (\ref{eq:necessary_condition2_general2}) are  necessary and sufficient conditions for the optimality of the input distribution $F_0$. Next, we obtain  condition (\ref{eq:C_necessary_sufficient_condition2}) from (\ref{eq:necessary_condition2_general1}) and (\ref{eq:necessary_condition2_general2}). By definition, the marginal information density $i(x,F_0)$ is given by \cite{SMITH19712}
\begin{equation}
\begin{aligned}
&i(x,F_0)=\int_y p(y|x)\log_2\left(\frac{p(y|x)}{p(y;F_0)}\right)\dd y\\
&=\int \!p(y|x)\log_2(p(y|x))\dd y\!-\!\int\! p(y|x)\log_2(p(y;F_0))\dd y\\
&=-\frac{1}{2}\log_2(2\pi\e\sigma_n^2)-\!\!\int\frac{\e^{-\frac{(y-xh_{\rm I})^2}{2\sigma_n^2}}}{\sqrt{2\pi\sigma_n^2}}\log_2(p(y;F_0))\dd y,
\end{aligned}
\label{eq:i_x_F0}
\end{equation}
where the first term on the RHS is the negative of the entropy of the noise. Finally, using the definitions of $A_l(x)$ and $a_l$ for $l\in\{0\}\cup\mathcal{L}$, (\ref{eq:necessary_condition2_general1}) and (\ref{eq:necessary_condition2_general2}) reduce to (\ref{eq:C_necessary_sufficient_condition2}). This completes the proof.

%% file: proof_discrete_optimal_distribution.tex
\section{Proof of Theorem \ref{theo:discrete_optimal_distribution}}
\label{App:proof_discrete_optimal_distribution}
Our proof of the discreteness and finiteness of the optimal input distribution parallels that in \cite{capacity_Rayleigh_fading_Abou_Faycal2001} and \cite[Section IV]{Hermite_bases_Abou_Faycal_2012}. Specifically, we prove by contradiction that the set of mass points $E_0$ of the optimal input distribution must be discrete with finite number of mass points.  In particular, assuming $E_0$ is continuous or discrete with infinite number of mass points, then according to the Bolzano-Weierstrass theorem, since $E_0\subset [-A,A]$, i.e., $E_0$ is bounded, $E_0$ must have an accumulation point  \cite{SMITH19712}.
On the other hand, from Corollary \ref{corol:C_necessary_sufficient_condition2}, a necessary condition for the optimal input distribution to be $F_0$ is that $s(x)$  in (\ref{eq:C_necessary_sufficient_condition2}) must be zero, $\forall \,x\in E_0$. Hence, $s(x)$ must be zero on an infinite set of points having an accumulation point. 
Next, we extend $s(x)$ in (\ref{eq:C_necessary_sufficient_condition2}) to the complex domain, i.e.,
\begin{equation}
\begin{aligned}
&s(z)\!\definedas\!\lambda_0\!\left(z^2\!-\!\sigma^2\right)\!-\!\sum\limits_{l\in\mathcal{L}}\! \lambda_l\left(P_l(z)\!-\!P_{l,\req}\right)\!+\!C\\
&+\!\frac{1}{2}\!\log_2(2\pi\e\sigma_n^2)\!+\!\!\frac{1}{\sqrt{2\pi\sigma_n^2}}\!\!\int\!\!\!\e^{-\frac{(y-zh_{\rm I})^2}{2\sigma_n^2}}\!\log_2(p(y;F_0))\dd y,
\end{aligned}
\label{eq:sz}
\end{equation}
where $z\in\mathbb{C}$. The extension to the complex domain is necessary to use the identity theorem for analytic functions in complex analysis. In particular, we first establish the analyticity of function $s(z)$. In (\ref{eq:sz}), the quadratic function, the exponential function, and their compositions are analytic functions in the whole  complex domain $z\!\in\!\mathbb{C}$\cite{SMITH19712}. In {\color{black} $
P_l(z)=\min\left(\left[\frac{1}{a}W_0\left(a\e^aI_0\left(\sqrt{2}B h_{{\rm E}_l} z \right)\right)-1\right]^2 i_{\rm s}^2R_{\rm L},\frac{B_{\rm v}^2}{4R_{\rm L}}\right)$,} defined in (\ref{eq:P_x_closedform}), {\color{black}the $\min(\cdot)$ function}, the quadratic function $(\cdot)^2$, the modified Bessel function $I_0(\cdot)$, {\color{black}and their compositions} are analytic on the whole complex domain $z\in\mathbb{C}$. The principal branch of the LambertW function $W_0(\cdot)$ is analytic everywhere in the complex domain with the exception of the branch cut along the negative real axis, i.e., on $(-\infty,-1/\e)$ \cite{Lagrange_Inversion_LambertW,Corless1996}. Hence, the function $s(z)$ is analytic in the domain $D$ defined by $D\definedas \{z\in\mathbb{C}: a\e^a I_0(\sqrt{2}B h_{{\rm E}_l} z)\in\mathbb{C}\setminus (-\infty,-1/\e),\,\,\forall l\in\mathcal{L}\}.$
Thus, we have an analytic function $s(z)$ in a domain $D$ that is zero over an infinite set of points $E_0$ having an accumulation point in $D$. By the identity theorem  \cite{QAGC_Shamai_Bar_David_1995, capacity_Rayleigh_fading_Abou_Faycal2001},  function $s(z)$ must be zero in the whole domain $D$, i.e., $s(z)=0$, $\forall z\in D$. Next, we show that this condition is invalid over a subset of  $D$, which violates the original assumption on  $E_0$ being continuous or discrete with infinite number of mass points. 

First, we restrict our attention to $z\in\mathbb{R}$ which is a subset\footnote{$\mathbb{R}\!\subset\! D$ since for $z\!\in\!\mathbb{R}$, the argument of $W_0(\cdot)$ is $a\e^a I_0(\sqrt{2}B h_{{\rm E}_l} z)\in(a\e^a,\infty)\!\subset\!\mathbb{R}^+$, i.e., it is in the analytical domain of $W_0(\cdot)$.} of $D$ and solve for the unknown distribution $p(y;F_0)$ in (\ref{eq:sz}) for which $s(z)\!=\!0$. Similar to \cite{Hermite_bases_Abou_Faycal_2012}, we set $\sigma_n^2\!=\!1$ to simplify the proof without loss of generality (w.l.o.g.) and express the last integral term in (\ref{eq:sz}) in terms of the Hermite polynomials $H_m(y)$ defined in \cite[Appendix F]{Hermite_bases_Abou_Faycal_2012}. In particular, since $\log_2(p(y;F_0))$ is a continuous function of $y$ and is square integrable with respect to  $\e^{-y^2/2}$,  it can be written as $\log_2(p(y;F_0))\!=\!\sum_{m=0}^\infty c_mH_m(y)$, 
%
where $c_m$ are constants. Hence, the last term of $s(z)$ in (\ref{eq:sz}) can be written as $Z\!\definedas\!\frac{1}{\sqrt{2\pi}}\!\!\int\!\!\e^{-\frac{y^2}{2}}\e^{-\frac{(h_{\rm I}z)^2}{2}+h_{\rm I} zy}\log_2(p(y;F_0)) \dd y =\frac{1}{\sqrt{2\pi}}\!\!\int\!\!\e^{-\frac{y^2}{2}}\sum_{n=0}^\infty \!\frac{(h_{\rm I} z)^n}{n!}H_n(y)\!\sum_{m=0}^\infty\!c_mH_m(y) \dd y$, where we used the Hermite polynomial expansion $\e^{-(h_{\rm I} z)^2/2+h_{\rm I}z y}\!=\!\sum_{n=0}^\infty \frac{(h_{\rm I} z)^n}{n!}H_n(y)$  \cite{Hermite_bases_Abou_Faycal_2012}. 

Using the orthogonality of the Hermite polynomials with $\e^{-y^2/2}$ given by $\int_{-\infty}^\infty H_n(y)H_m(y) \e^{-y^2/2}\dd y\!=\!m!\sqrt{2\pi}$ if $m\!=\!n$ and zero otherwise \cite[Appendix F]{Hermite_bases_Abou_Faycal_2012}, then $Z$ reduces to $Z\!=\!\sum_{m=0}^\infty c_m(h_{\rm I} z)^m$. Next, we replace $P_l(z)$ by its Taylor series to write $s(z)$ in (\ref{eq:sz}) in a polynomial form. In particular, the LambertW function admits a convergent Taylor series around an arbitrary point  $x_0\in\mathbb{R}$ given by \cite[Eq. (8), (10)]{Lagrange_Inversion_LambertW}
\begin{equation*}
\begin{aligned}
W_0(x)\!=\!\!&\sum\limits_{n=1}^{\infty}\!\frac{\e^{-n W_0(x_0)}}{(1\!+\!W_0(x_0))^{2n\!-\!1}}\frac{p_n(W_0(x_0))}{n!}(x\!-\!x_0)^n\\
& \!+\!W_0(x_0)\definedas \sum_{n=0}^\infty d_n(x_0) x^n,\, |x-x_0|<x_{{\rm ROC}},
\end{aligned}
\label{eq:LambertW_series}
\end{equation*}
where $p_n(\cdot)$ is a polynomial with coefficients given in \cite[Table I.]{Lagrange_Inversion_LambertW}. This series holds for some radius of convergence (ROC)  $|x-x_0|<x_{{\rm ROC}}$ and can be expanded to a polynomial in $x$ defined as $W_0(x)=\sum_{n=0}^\infty d_n(x_0) x^n$. Moreover, using the Taylor series expansion of the modified Bessel function given by $I_0(z)=\sum_{m=0}^\infty \frac{(z/2)^{2m}}{(m!)^2}$, the Bessel function in $P_l(z)$ can be written as $I_0\left(\sqrt{2}Bh_{{\rm E}_l} z\right)\!=\!\sum_{m=0}^\infty \alpha_{m,l} z^{2m}$, where $\alpha_{m,l}\!=\!\frac{(Bh_{{\rm E}_l}/\sqrt{2})^{2m}}{(m!)^2}$. Hence, the harvested power function $P_l(z)$ in (\ref{eq:P_x_closedform}) is a quadratic function of a polynomial of another polynomial function with even powers of $z$. Thus, $P_l(z)$ can be written as
{\color{black}
\begin{equation} 
\begin{aligned}
P_l(z)\!=\!\min\!\Bigg(&\!\!\left[\!\frac{1}{a}\!\sum\limits_{n=0}^\infty \!d_n\!(x_0)\!\!\left(\!a\e^a\!\sum_{m=0}^\infty \!\alpha_{m,l} z^{2m} \!\right)^n\!\!\!\!-\!1\right]^2\!\!i_{\rm s}^2R_{\rm L}\\
&,\frac{B_{\rm v}^2}{4R_{\rm L}}\Bigg)
 \definedas\min\left(\sum\limits_{m=0}^{\infty} q_{m,l}z^{2m},\frac{B_{\rm v}^2}{4R_{\rm L}}\right),
\end{aligned}
\label{eq:Harvested_power_series}
\end{equation} }
where $q_{m,l}\in\mathbb{R}$ and (\ref{eq:Harvested_power_series}) holds for  $|a\e^aI_0(\sqrt{2}B h_{{\rm E}_l} z)\!-\!x_0|\!<\!x_{{\rm ROC}}$. From (\ref{eq:sz}), $s(z)\!=\!0$ reduces to 
{\color{black}
\begin{gather}
\begin{aligned}
\sum\limits_{m=0}^\infty\! c_m h_{\rm I}^m z^m
\!=\!&\sum\limits_{l\in\mathcal{L}}\!\lambda_l\!\left(\min\!\left(\!\sum\limits_{m=0}^{\infty}\! q_{m,l}z^{2m},\frac{B_{\rm v}^2}{4R_{\rm L}}\right)\!\!-\!\!P_{l,\req}\!\right)\\
&-\lambda_0(z^2-\sigma^2)-C-\frac{1}{2}\log_2(2\pi\e).
\end{aligned}
\raisetag{4pt}
\label{eq:sz_0}
\end{gather}
}
Equating the coefficients of $z^m$, we get $c_{m}=0$ for odd $m$. To obtain $c_{m}$ for even $m$,  {\color{black}let $A\!\geq\!A_{{\rm T,sat}_l}$, $\forall$ $l\!\in\!\mathcal{L}_{\rm sat}$, i.e., with $P_l(z)=\frac{B_{\rm v}^2}{4R_{\rm L}}$, and $A\!< \!A_{{\rm T,sat}_l}$, $\forall$ $l\!\in\!{\cal{L}_{\rm non-sat}}$, then

\begin{gather}
\begin{aligned}
&c_0\!=\left(\!\sum\limits_{l\in\mathcal{L}_{\rm non-sat}} \! \lambda_l \,q_{0,l}\!+\!\sum\limits_{l\in\mathcal{L}_{\rm sat}} \! \lambda_l\frac{B_{\rm v}^2}{4R_{\rm L}} \!-\sum\limits_{l\in\mathcal{L}}\!\lambda_l P_{l,\req}\!\right)\\
&\quad\quad\quad+\!\lambda_0\sigma^2\!-\!C\!-\!0.5\log_2(2\pi\e), \\
&c_2\!=\!\left(\sum\limits_{l\in\mathcal{L}_{\rm non-sat}}\!\!\!\!\!\!\!\lambda_l q_{1,l}-\lambda_0\right)/h_{\rm I}^2,\\
&c_m\!=\!\sum\limits_{l\in\mathcal{L}_{\rm non-sat}} \frac{\lambda_l q_{\frac{m}{2},l}}{h_{\rm I}^m},\,\, \forall\,\, m\geq 4.
\end{aligned}
\raisetag{17pt}
\label{eq:cm}
\end{gather}}

Using (\ref{eq:cm}), the output pdf in $\log_2(p(y;F_0))\!=\!\sum_{m=0}^\infty c_mH_m(y)$ reduces to $p(y;F_0)=\e^{\ln(2)\sum_{n=0}^\infty c_{2n}H_{2n}(y)}
$. Next, we consider two cases based on whether or not the EH constraints are  active. We will show that in both cases, the optimal input distribution is  discrete with finite number of mass points. \\
\emph{Case 1 ($\lambda_l\!=\!0$, $\forall\,\, l\in\mathcal{L}$)}: If all EH constraints are inactive, i.e., they are satisfied with strict inequality, then  $\lambda_l\!=\!0$, $\forall\,\, l\in\mathcal{L}$, from the complementary slackness, cf. Theorem \ref{theo:unique_distribution}. In this case, the coefficients in (\ref{eq:cm})   reduce to $c_0\!=\!\lambda_0\sigma^2-C\!-\!0.5\log_2(2\pi\e)$, $c_2\!=\!-\lambda_0/h_{\rm I}^2$, and $c_m\!=\!0,\,\forall\,m\neq\{0,2\}$. Using the Hermite polynomials $H_0(y)\!=\!1$, $H_2(y)\!=\!y^2\!-\!1$  \cite[Appendix F]{Hermite_bases_Abou_Faycal_2012}, the output pdf  reduces to  $p(y;F_0)\!=\!\e^{\ln(2)(c_{0}-c_2)}\e^{\ln(2)c_2 y^2}$.
Since the support of $p(y;F_0)$ is the whole real line $\mathbb{R}$ and $c_2\!<\!0$, this output distribution is the Gaussian distribution with zero mean. Now, for $\V{y}$ to be Gaussian distributed for the AWGN channel model $\V{y}\!=\!\V{x}h_{\rm I}\!+\!\V{n}$, $\V{x}$ must also be Gaussian distributed. However, with the PP constraint $|\V{x}|\!\leq\! A$, $\V{x}$ cannot be Gaussian distributed on a bounded interval. Thus, the obtained output distribution is invalid. Hence, the condition $s(z)\!=\!0$ is invalid on a subset of the domain $D$ defined by the region of convergence of (\ref{eq:Harvested_power_series}). This contradicts the original assumption that $E_0$ is continuous or discrete with infinite number of mass points. Therefore, $E_0$ must be discrete with finite number of mass points. \\
\emph{Case 2 $(\lambda_l\!>\!0,$ $l\!\in\!\mathcal{L}_{\rm A}\!\subset\!\mathcal{L})$}: In this case, some of the EH constraints are active, i.e., $\forall\,\, l \in\!\mathcal{L}_{\rm A}$, ${\rm C}_l$ in (\ref{eq:capacity_problem}) is satisfied with equality and the coefficients $c_m$ are given by (\ref{eq:cm}). From \cite{Hermite_bases_Abou_Faycal_2012}, the Hermite polynomials of even orders are function of even powers of $y$. Thus, the output distribution reduces to $p(y;F_0)\!=\!\e^{\ln(2)\sum_{n=0}^\infty t_{n}y^{2n}}=\prod_{n=0}^\infty\e^{\ln(2)t_ny^{2n}}$, where $t_n$ are non-zero constants. It can be verified that for some $n\!\to\!\infty$, $\exists\, t_n\!>\!0$, in which case $p(y;F_0)$ cannot be a valid distribution since it is unbounded. Hence, we conclude that the set $E_0$  must be discrete and finite. This completes the proof. 

%% file: proof_Only_one_EH_Rx_effective.tex
\section{Proof of Theorem \ref{theo:Only_one_EH_Rx_effective}}
\label{App:proof_Only_one_EH_Rx_effective}
Consider problem (\ref{eq:capacity_problem}) with $A< A_{\rm sat}$.  Assume that only one EH Rx, having index $1$, requires a certain average harvested power $P_{1,\req}$ and the remaining EH Rxs passively harvest energy from their received signals. Assume further that EH constraint  ${\rm C}_{1}$ is feasible and active and the optimal distribution is $F_1$. The average harvested powers at the EH Rxs are $\E_{F_1}[P_{l}(\V{x})]$, $\forall \,l\in\mathcal{L}$.  Next, we show that if another feasible EH constraint ${\rm C}_{2}$ is added to problem (\ref{eq:capacity_problem}), then at most one of the two constraints, denoted by ${\rm C}_a$,  is active.  In particular, we add constraint ${\rm C}_{2}$ for EH Rx $2$ with required average harvested power $P_{2,\req}$. In this case, if $P_{2,\req}\!<\! \E_{F_1}[P_{2}(\V{x})]$, then $F_1$ remains the optimal distribution and the EH constraint ${\rm C}_{2}$ is inactive, i.e., ${\rm C}_a\!=\!{\rm C}_{1}$. Otherwise, if $P_{2,\req}\!>\!\E_{F_1}[P_{2}(\V{x})]$, then distribution $F_1$ fails to satisfy the EH requirement of EH Rx $2$ and a more energy-biased distribution is needed. In this case, if  problem (\ref{eq:capacity_problem}) with only the EH constraint for EH Rx  $2$ is feasible with optimal distribution $F_2$, then $F_2$ is  more energy-biased but also less information-biased compared to $F_1$. In particular, according to Lemma 
\ref{lem:monotonicity_Power_function_before_sat}, $P_{2,\req}\!=\!\E_{F_2}[P_{2}(\V{x})]\!>\!\E_{F_1}[P_{2}(\V{x})]$ also implies $\E_{F_2}[P_{1}(\V{x})]\!>\!\E_{F_1}[P_{1}(\V{x})]\!=\!P_{1,\req}$, i.e., using $F_2$, the EH constraint of EH Rx $1$  is satisfied with strict inequality. Hence, ${\rm C}_{1}$ is inactive and only ${\rm C}_{2}$ is active, i.e., ${\rm C}_a\!=\!{\rm C}_{2}$. The same discussion holds for constraint ${\rm C}_a$ and any other added constraint ${\rm C}_{3}$. Hence, by induction, at most one of the EH constraints in problem (\ref{eq:capacity_problem}) is active.  Moreover, due to the R-E tradeoff, a more energy-biased input distribution implies a lower information rate at the ID Rx, e.g., if $\E_{F_2}[P_{l}(\V{x})]\!>\!\E_{F_1}[P_{l}(\V{x})]$, then $I(F_2)\!<\!I(F_1)$ holds. Hence, the active EH constraint is the one, which when all other EH constraints are removed, leads to the lowest achievable rate at the ID Rx. This completes the proof.

%% file: proof_max_WPT.tex
\section{Proof of Theorem \ref{theo:EH_only}}
\label{App:proof_max_WPT}
Consider first the case when $A\!< \!A_{{\rm T,sat}_l}$. Suppose $x_0$ is a point of increase of distribution $F_0$ having probability $p_0$, where $0\!<\!x_0\!<\!A$. Thereby, we introduce a new distribution $F_A$ which is constructed from $F_0$ by removing the mass point at $x_0$ and increasing the probabilities of mass points  $0$ and $A$ by $p_0-p_0x_0^2/A^2$ and $p_0x_0^2/A^2$, respectively. This transformation maintains the unity of the sum of probabilities of the mass points and ensures that the AP and PP constraints hold. Next, we show that if  $P_l(0)=0$ and $P_l(x)/x^2$ is a monotonically increasing function in {\color{black}$0\!<\!x\!<\!A$}, then the contribution of the mass points at $x=0$ and $x=A$ to the average harvested power is higher than the contribution of any other point $x_0\in(0,A)$. In particular,
 \begin{equation}
 \begin{aligned}
&p_0P_l(x_0) < 
p_0\frac{x_0^2}{A^2}P_l(A)
+p_0\left(1-\frac{x_0^2}{A^2}\right)P_l(0)\\
&\Rightarrow \quad \frac{P_l(x_0)}{x_0^2}<\frac{P_l(A)}{A^2},
\end{aligned}
\label{eq:compare_dist_max_EH}
\end{equation}
Condition (\ref{eq:compare_dist_max_EH}) holds  for the harvested power function in (\ref{eq:P_x_closedform}) since (a) $P_l(0)\!=\!0$, as $I_0(0)\!=\!1$ and $W_0(a\e^a)\!=\!a$ and (b) $P_l(x)/x^2$ increases monotonically in {\color{black}$0\!<\!x\!<\!A$}. We derive this monotonicity by proving that $xP_l'(x)\!-\!2P_l(x)>0$ holds for {\color{black}$0\!<\!x\!<\!A$}, where $P_l'(x)$ is the first-order derivative of $P_l(x)$. In particular, this condition can be expressed as $D(x)\!\definedas\!\left[\frac{u_l x I_1(u_lx)}{I_0(u_lx)[1+W_0\left(a\e^aI_0(u_lx)\right)]}\!-\!1\right]\frac{W_0\left(a\e^aI_0(u_lx)\right)}{a}\!>\!-1,\,\forall\,x,$ {\color{black}$0\!<\!x\!<\!A$,}
where $u_l\!\definedas\! \sqrt{2} B h_{{\rm E}_l}$. It can be shown that $D(x)$ equals $-1$ for  $x\!=\!0$ and is larger than $-1$ for $x\!>\!0$.
Moreover, since the harvested power is an even function of $x$, i.e., $P_l(x)=P_l(-x)$, the weight of the mass point at $A$ can be arbitrarily distributed between $A$ and $-A$. To satisfy the AP constraint in (\ref{eq:energy_problem}) with equality, the total weights on the peak amplitudes $A$ and $-A$ should satisfy $p+q=\min\{\sigma^2/A^2,1\}$.
{\color{black}Consider next the case when $A\!\geq\!A_{{\rm T,sat}_l}$. For $x_0\in(0,A_{{\rm T,sat}_l})$, since $P_l(x)/x^2$ increases monotonically in $0\!\leq\!x\!<\!A_{{\rm T,sat}_l}$, similar to (\ref{eq:compare_dist_max_EH}), the contribution of the mass points at $x\!=\!0$ and $x\!=\!A_{{\rm T,sat}_l}$ to the average harvested power is higher than the contribution of any other point $x_0\in(0,A_{{\rm T,sat}_l})$, since $\frac{P_l(x_0)}{x_0^2}<\frac{P_l(A_{{\rm T,sat}_l})}{A_{{\rm T,sat}_l}^2}$. On the other hand, for $x_0\in(A_{{\rm T,sat}_l},A)$, $P_l(x_0)$ is constant, i.e.,  $P_l(x_0)=P_l(A_{{\rm T,sat}_l})=B_v^2/(4R_{\rm L})$. In this case, $\frac{P_l(x_0)}{x_0^2}<\frac{P_l(A_{{\rm T,sat}_l})}{A_{{\rm T,sat}_l}^2}$ since $\frac{1}{x_0^2}<\frac{1}{A_{{\rm T,sat}_l}^2}$, i.e., $P_l(x)/x^2 \propto 1/x^2$ is monotonically decreasing in $x$, for $x\geq A_{{\rm T,sat}_l}$. As a result the contribution of the mass points at $x=0$ and $x=A_{{\rm T,sat}_l}$ to the average harvested power is higher than the contribution of any other point $x_0\in(A_{{\rm T,sat}_l},A)$.} This leads to the EH maximizing distribution  given in (\ref{eq:PMF_max_EH}). This completes the proof.

%% file: proof_indep_opt_amplitude_phase_UD.tex
\section{Proof of Lemma \ref{lemma:indep_opt_amplitude_phase_UD}}
\label{App:proof_indep_opt_amplitude_phase_UD}
The proof of the optimality of input signals with independent amplitude and phase distributions parallels that in \cite[Section II.B]{QAGC_Shamai_Bar_David_1995}. We start by expressing the mutual information as \cite[eq. (6)]{QAGC_Shamai_Bar_David_1995}
\begin{gather}
\begin{aligned}
&I(\V{y};\V{x})=H(\V{y})-H(\V{n})\\
&=H(\V{R},\V{\psi})+\int\limits_{R=0}^{\infty} f_{\V{R}}(R) \log_2(R) \dd R -\log_2(2\pi\e\sigma_n^2).
\end{aligned}
\raisetag{7pt}
\label{eq:I_y_x_diff_entropies}
\end{gather}
We note that the joint entropy $H(\V{R},\V{\psi})$ is maximized for independent $\V{R}$ and $\V{\psi}$ with a maximum of $\max\left(H(\V{R},\V{\psi})\right)= H(\V{R})+H(\V{\psi})$, and the entropy  $H(\V{\psi})$ is maximized for a uniformly distributed phase $\V{\psi}$ with a maximum of $\max\left(H(\V{\psi})\right)=\log_2(2\pi)$. Hence, we have
\begin{gather}
\begin{aligned}
&\sup \,  H(\V{R},\V{\psi})=\sup\, H(\V{R})+\log_2(2\pi)\\
&=\sup \left\{ -\int\limits_{R=0}^{\infty} f_{\V{R}}(R) \log_2(f_{\V{R}}(R)) \dd R\right\}+\log_2(2\pi),
\end{aligned}
\raisetag{7pt}
\label{eq:max_joint_entropy}
\end{gather}
which when combined with (\ref{eq:I_y_x_diff_entropies}) results in
\begin{gather}
\begin{aligned}
&\sup I(\V{y};\V{x})=\sup I(F_{\V{r}})\\
&=\sup\!\left\{\!-\!\!\!\!\int\limits_{R=0}^{\infty} \!\!\!\!f_{\V{R}}(R;\!F_{\V{r}}) \log_2\!\!\left(\!\frac{f_{\V{R}}(\!R;\!F_{\V{r}}\!)}{R}\!\right) \!\dd R\!\right\}\!-\!\log_2(\e\sigma_n^2),
\end{aligned}
\raisetag{7pt}
\label{eq:sup_I_F}
\end{gather}
where $F_{\V{r}}$ in $f_{\V{R}}(R;F_{\V{r}})$ is used to emphasize that $f_{\V{R}}(R)$ depends on $F_{\V{r}}$.
Similar to \cite[eq. (11)]{QAGC_Shamai_Bar_David_1995}, 
\begin{gather}
f_{\!\V{R}\!,\!\V{\psi}|\V{r}\!,\V{\theta}}(R,\psi|r,\theta)\!=\!\frac{R}{2\pi\sigma_n^2}\e^{\!-\!\frac{(R^2+r^2|h_{\rm I}|^2-2Rr|h_{\rm I}|\cos(\psi-\theta-\phi))}{2\sigma_n^2}}.
\raisetag{7pt}
\label{eq:cartesian_to_polar_conditional_pdf2}
\end{gather}
Furthermore, similar to \cite[eq. (10)]{QAGC_Shamai_Bar_David_1995}, it can be shown that 
\begin{gather}
\begin{aligned}
f_{\V{R}}(R;F_{\V{r}}) &=\int\limits_{r=0}^{A}\frac{R}{\sigma_n^2}\e^{-\frac{R^2+r^2|h_{\rm I}|^2}{2\sigma_n^2}}I_0\left(\frac{Rr|h_{\rm I}|}{\sigma_n^2}\right)\dd F_{\V{r}}(r)\\
& \definedas \int\limits_{r=0}^{A}K(r,R)\dd F_{\V{r}}(r).
\end{aligned}
\raisetag{22pt}
\label{eq:f_R_final}
\end{gather}
Hence, it is concluded that $f_{\V{R}}(R;F_{\V{r}})$ is independent of $F_{\V{\theta}}(\theta)$. Next, we prove that selecting independent $\V{r}$ and $\V{\theta}$, with uniformly distributed $\V{\theta}$, i.e.,  $\dd^2F_{\V{r},\V{\theta}}(r,\theta)=\frac{1}{2\pi}\dd \theta \dd F_{\V{r}}(r)$, results in independent $\V{R}$ and $\V{\psi}$, with uniformly distributed $\V{\psi}$, i.e., $f_{\V{R},\V{\psi}}(R,\psi)=\frac{1}{2\pi}f_{\V{R}}(R;F_{\V{r}}) $. In particular, using (\ref{eq:cartesian_to_polar_conditional_pdf2}) and (\ref{eq:f_R_final}), we get 
\begin{gather}
\begin{aligned}
&f_{\!\V{R},\V{\psi}}(\!R,\psi\!)\!=\!\!\!\int\limits_{r=0}^{A}\!\int\limits_{-\!\pi}^{\pi}\!\! f_{\!\V{R},\V{\psi}|\V{r},\V{\theta}}(\!R,\psi|r,\theta\!)\dd^2\!F_{\!\V{r},\V{\theta}}(\!r,\theta\!)\\
&=\!\!\!\int\limits_{r=0}^{A}\!\! \frac{R}{2\pi\sigma_n^2}\e^{-\frac{R^2\!+\!r^2|h_{\rm I}|^2}{2\sigma_n^2}}I_0\!\!\left(\!\!\frac{Rr|h_{\rm I}|}{\sigma_n^2}\!\!\right)\! \dd F_{\V{r}}(\!r\!)\!=\!\frac{1}{2\pi}f_{\V{R}}(\!R;F_{\V{r}}\!) ,
\end{aligned}
\raisetag{7pt}
\label{eq:f_R_psi_indep_UD}
\end{gather}
which, from (\ref{eq:max_joint_entropy}) and (\ref{eq:sup_I_F}), maximizes the joint entropy $ H(\V{R},\V{\psi})$ and the mutual information $I(F_{\V{r}})$. Hence, independent $\V{r}$ and $\V{\theta}$ with uniformly distributed $\V{\theta}$ are optimal. This completes the proof. 

%% file: proof_C_necessary_sufficient_condition2_complex.tex
\section{Proof of Theorem  \ref{theo:C_necessary_sufficient_condition2_complex}}
\label{App:proof_C_necessary_sufficient_condition2_complex}
The proof of the uniqueness of the solution of problem (\ref{eq:capacity_problem_Complex}) is similar to that in Appendix \ref{App:proof_unique_optimal_distribution}. In particular, the constraints in (\ref{eq:capacity_problem_Complex}) are convex and compact in $F_{\V{r}}$. From (\ref{eq:sup_I_F}), $I(F_{\V{r}})$ can be written as
\begin{gather}
\begin{aligned}
I(F_{\V{r}})&=\!-\!\int_{\nu=0}^{\infty}\! f_{\V{\nu}}(\nu;F_{\V{r}}) \log_2\left(f_{\V{\nu}}(\nu;F_{\V{r}})\right) \dd \nu\!-\!\log_2(\e\sigma_n^2)\\
&\definedas h(\V{\nu};F_{\V{r}})-\log_2(\e\sigma_n^2),
\end{aligned}
\raisetag{10pt}
\label{eq:sup_I_x_y_entropy_R_psi_partb}
\end{gather}
where we used the change of variables $\nu\!=\!R^2/2$, hence $f_{\V{R}}(R;F_{\V{r}})/R\!=\!f_{\V{\nu}}(\nu;F_{\V{r}})$ and $\dd \nu\!=\!R\dd R$. Hence, the mutual information depends on the entropy associated with random variable $\V{\nu}$ which is strictly concave in $F_{\V{r}}$. Therefore, the solution to problem (\ref{eq:capacity_problem_Complex}) is unique.
From (\ref{eq:sup_I_F}) and (\ref{eq:f_R_final}), $I(F_{\V{r}})$ can be written as
\small
\begin{gather}
\begin{aligned}
&I(F_{\V{r}})\!=\!\!-\!\!\!\!\int\limits_{R=0}^{\infty} \int\limits_{r=0}^{A}\!\!\!K(r\!,\!R) \log_2\!\left(\!\!\frac{f_{\!\V{R}}(R;\!F_{\V{r}}\!)}{R}\!\!\right)\! \dd R \dd F_{\V{r}}(r)\!-\!\log_2(\e\sigma_n^2) \\
&\!= \!\!\!\int\limits_{r=0}^{A}\!\!\!\!\left( \!-\!\!\!\int\limits_{R=0}^{\infty}\!\!K(r,R) \log_2\left(\!\frac{f_{\V{R}}(R;F_{\V{r}})}{R}\!\right) \dd R \!-\!\log_2(\e\sigma_n^2)\!\right)\! \dd F_{\V{r}}(r)\\
&\definedas  \int\limits_{r=0}^{A} i(r;F_{\V{r}}) \dd F_{\V{r}}(r),
\end{aligned}
\raisetag{17pt}
\label{eq:mutual_info_marginal_complex}
\end{gather}
\normalsize
 where we used $ \int_{r=0}^{A} \dd F_{\V{r}}(r)=1$.
 Next, we obtain the necessary and sufficient conditions for the input distribution $F_{\V{r}_0}(r)$ to be optimal. These conditions correspond to those in Theorem \ref{theo:C_necessary_sufficient_condition1} and Corollary \ref{corol:C_necessary_sufficient_condition2} but for the input amplitude $r$. In particular, with the definition of $I(F_{\V{r}})$ in (\ref{eq:mutual_info_marginal_complex}), the complex signaling problem in (\ref{eq:capacity_problem_Complex}) is symbolically equivalent to that with real signaling in (\ref{eq:capacity_problem}) after replacing random variable $\V{x}\in[-A,A]$ by $\V{r}\in[0,A]$ and using the marginal mutual information $i(r;F_{\V{r}})$  in (\ref{eq:mutual_info_marginal_complex}).  Following 
Appendices \ref{App:proof_C_necessary_sufficient_condition1} and \ref{App:proof_C_necessary_sufficient_condition2}, the conditions in (\ref{eq:necessary_condition2_general1}) and (\ref{eq:necessary_condition2_general2}) generalize to $i(r;F_{\V{r}_0})\!\leq\! C\!+\!\sum_{l\in\{0\}\cup\mathcal{L}}\lambda_l\left(A_l(r)\!-\!a_l\right), \, \forall r\in[0,A]$, 
where equality holds if $r\!\in\! E_0$. Substituting with  $A_0(r)\!=\!r^2$, $a_0\!=\!\sigma^2$, $A_l(r)\!=\!-P_l(r)$, 
$a_l\!=\!-P_{l,\req}$, $l\in\mathcal{L}$, 
$i(r;F_{\V{r}_0})$ from (\ref{eq:mutual_info_marginal_complex}), and $K(r,R)$ from (\ref{eq:f_R_final}), we obtain  (\ref{eq:C_necessary_sufficient_condition2_complex}).

Next, we prove that the optimal input distribution must be discrete with finite number of mass points. Similar to Appendix \ref{App:proof_discrete_optimal_distribution}, we prove that the complex extension of $s(r)$ in (\ref{eq:C_necessary_sufficient_condition2_complex}) cannot be zero over an infinite set of points having an accumulation point and hence $E_0$ must be discrete and finite. For simplicity and w.l.o.g., we set $\sigma_n^2\!=\!1$. Extending $s(r)$ in (\ref{eq:C_necessary_sufficient_condition2_complex}) to the complex $z$ domain, we can write $s(z)\!=\!0$ as
\begin{gather}
\begin{aligned}
&\int\limits_{0}^\infty Q(\nu,z|h_{\rm I}|)\log_2\left(f_{\V{\nu}}(\nu;F_{\V{r_0}})\right) \dd \nu\\
&\! =\!
-\!\lambda_0\!\left(z^2\!-\!\sigma^2\right)\!+\!\sum\limits_{l\in\mathcal{L}} \lambda_l\!\left(\!P_l(z)\!\!-\!P_{l,\req}\!\right)\!-\!C\!-\!\log_2(\e)  ,
\end{aligned}
\raisetag{8pt}
\label{eq:s_z_0_complex}
\end{gather}
where we used  $\nu=R^2/2$, hence $f_{\V{R}}(R;F_{\V{r}})/R=f_{\V{\nu}}(\nu;F_{\V{r}})$ and $\dd \nu=R\dd R$ and we define the kernel $Q(\nu,z|h_{\rm I}|)$ as $Q(\nu,z|h_{\rm I}|)\definedas \e^{-\nu-\frac{z^2|h_{\rm I}|^2}{2}}\!I_0\left(\sqrt{2\nu}z|h_{\rm I}|\right)$. The analyticity of the RHS of (\ref{eq:s_z_0_complex}) is proved in Appendix \ref{App:proof_discrete_optimal_distribution}. The LHS of (\ref{eq:s_z_0_complex}) is analytic in $z\in\mathbb{C}$, which follows by the differentiation lemma and the Schwarz property of the kernel $Q(\nu,z|h_{\rm I}|)$ \cite[Appendix I]{QAGC_Shamai_Bar_David_1995}. Hence, $s(z)$ is analytic over domain $D$ defined in Appendix \ref{App:proof_discrete_optimal_distribution}. Thus, from the identity theorem \cite{QAGC_Shamai_Bar_David_1995, capacity_Rayleigh_fading_Abou_Faycal2001}, if $s(z)=0$ $\forall \,z\in E_0$ and $E_0$ is an infinite set of points with an accumulation point, then $s(z)=0,$ $\forall\,\, z\in D$. Next, we restrict our attention to $z\in\mathbb{R}\subset D$ and check if there exists a valid output pdf $f_{\V{\nu}}(\nu;F_{\V{r_0}})$ that satisfies (\ref{eq:s_z_0_complex}). 
Eq. (\ref{eq:s_z_0_complex}) is an integral transform which was proved in \cite[Appendix II]{QAGC_Shamai_Bar_David_1995} to be invertible, i.e., there exists a unique solution for the unknown function $f_{\V{\nu}}(\nu;F_{\V{r_0}})$\footnote{Unlike with real signaling, where $f_{\V{\nu}}(\nu;F_{\V{r_0}})$ was obtained in terms of the Hermite polynomials, with complex signaling, the kernel $Q(\nu,z|h_{\rm I}|)$ is not orthogonal to the Hermite polynomials. Hence, we use a different approach to obtain  $f_{\V{\nu}}(\nu;F_{\V{r_0}})$.}. 
Applying the integral transform to the power function $\nu^n$, we get $\int_{0}^\infty Q(\nu,z|h_{\rm I}|) \nu^n \dd \nu=n!L_n(-z^2|h_{\rm I}|^2/2)$, where $L_n(\cdot)$ is the Laguerre Polynomial defined as $L_n(x)=\sum_{m=0}^n\binom{n}{m}\frac{(-1)^m}{m!}x^m$. Hence, the integral transform of the polynomial $\sum_{n=0}^\infty c_n\nu^n$ is
\begin{gather}
\int\limits_{0}^\infty Q(\nu,z|h_{\rm I}|) \sum\limits_{n=0}^\infty c_n\nu^n \dd \nu = \sum\limits_{n=0}^\infty c_n n!\,L_n\left(-z^2|h_{\rm I}|^2/2\right).
\raisetag{7pt}
\label{eq:integral_transform_sum_Laguerre}
\end{gather}
If there exist coefficients $c_n$ such that the RHS of (\ref{eq:s_z_0_complex}) equals the RHS of (\ref{eq:integral_transform_sum_Laguerre}), then the unique solution of $f_{\V{\nu}}(\nu;F_{\V{r_0}})$ must satisfy 
$\log_2\left(f_{\V{\nu}}(\nu;F_{\V{r_0}})\right)= \sum_{n=0}^\infty c_n\nu^n $. Next, we write the RHS of (\ref{eq:s_z_0_complex}) in a polynomial form, and check whether this polynomial can be written as a weighted sum of Laguerre polynomials as  in the RHS of (\ref{eq:integral_transform_sum_Laguerre}). Using {\color{black}$P_l(z)\!=\!\min\left(\sum_{m=0}^{\infty} q_{m,l}z^{2m},B_v^2/(4R_{\rm L})\right)$} from (\ref{eq:Harvested_power_series}), we write the RHS of (\ref{eq:s_z_0_complex}) as $\sum_{m=0}^\infty \alpha_{m} z^{2m}$, where {\color{black}$\alpha_{0}=\lambda_0\sigma^2\!+\!\Big(\sum\limits_{l\in\mathcal{L}_{\rm non-sat}} \lambda_lq_{0,l}\!+\!\!\sum\limits_{l\in\mathcal{L}_{\rm sat}}\lambda_l B_v^2/(4R_{\rm L}) -\sum\limits_{l\in\mathcal{L}}\!\lambda_l P_{l,\req}\Big)\!-\!C\!-\!\log_2(\e)$, $\alpha_1\!=\!-\lambda_0\!+\!\sum\limits_{l\in\mathcal{L}_{\rm non-sat}} \lambda_lq_{1,l}$, and $\alpha_m\!=\!\sum\limits_{l\in\mathcal{L}_{\rm non-sat}} \lambda_lq_{m,l}$ $\forall\, m \geq 2$.} Hence, (\ref{eq:s_z_0_complex}) can be written as 
\begin{equation}
\int\limits_{0}^\infty Q(\nu,z|h_{\rm I}|)\log_2\left(f_{\V{\nu}}(\nu;F_{\V{r_0}})\right) \dd \nu\! =\! \sum\limits_{m=0}^\infty \alpha_m z^{2m}.
\label{eq:s_z_0_complex2}
\end{equation}
Hence, the problem reduces to finding the coefficients $c_n$ that make the RHSs of 
(\ref{eq:integral_transform_sum_Laguerre}) and (\ref{eq:s_z_0_complex2}) equal.  Using the definition of the Laguerre polynomial, we require $\sum_{m=0}^\infty \alpha_m z^{2m}=\sum_{n=0}^\infty c_n n! \sum_{m=0}^n \binom{n}{m}\frac{|h_{\rm I}|^{2m}}{m!2^m}z^{2m}.$
Thus, $c_n$ satisfies the linear system of equations $\alpha_m\!=\!\sum_{n=m}^{\infty}c_n n! \binom{n}{m}\frac{|h_{\rm I}|^{2m}}{m!2^m}$, $m\!=\!0,\ldots,\infty$. Truncating the summation order to $2S$ for some large $S$, the linear system of equations can be written as $\V{\alpha}\!=\!M\V{c}$, where $\V{\alpha}\!=\![\alpha_0,\ldots,\alpha_S]^T$, $\V{c}\!=\![c_0,\ldots,c_S]^T$, and $M$ is an upper triangular matrix whose non-zero entries in the $m^{\text{th}}$ row and the $n^{\text{th}}$ column are $n! \binom{n}{m}\frac{|h_{\rm I}|^{2m}}{m!2^m}$, for $n\!\geq\! m$ and zero, otherwise. Since any upper triangular matrix with non-zero main diagonal entries is invertible, the coefficients $c_n$ can be obtained uniquely as $\V{c}\!=\!M^{-1}\V{\alpha}$. Using these coefficients, the RHSs of (\ref{eq:integral_transform_sum_Laguerre}) and (\ref{eq:s_z_0_complex2}) are equal, and therefore their LHSs are also equal, i.e., $f_{\V{\nu}}(\nu;F_{\V{r_0}})=\e^{\ln(2)\sum_{n=0}^\infty c_n\nu^n}$, which does not correspond to a legitimate pdf. Hence, (\ref{eq:s_z_0_complex}) cannot hold over an infinite set of points with an accumulation point, i.e., $E_0$ must be discrete and finite. This completes the proof. 

%% file: Biographies.tex
\begin{IEEEbiography}[{\includegraphics[width=1.6in,height=1.4in,trim=0in 0.2in 0in 0in,clip,keepaspectratio]{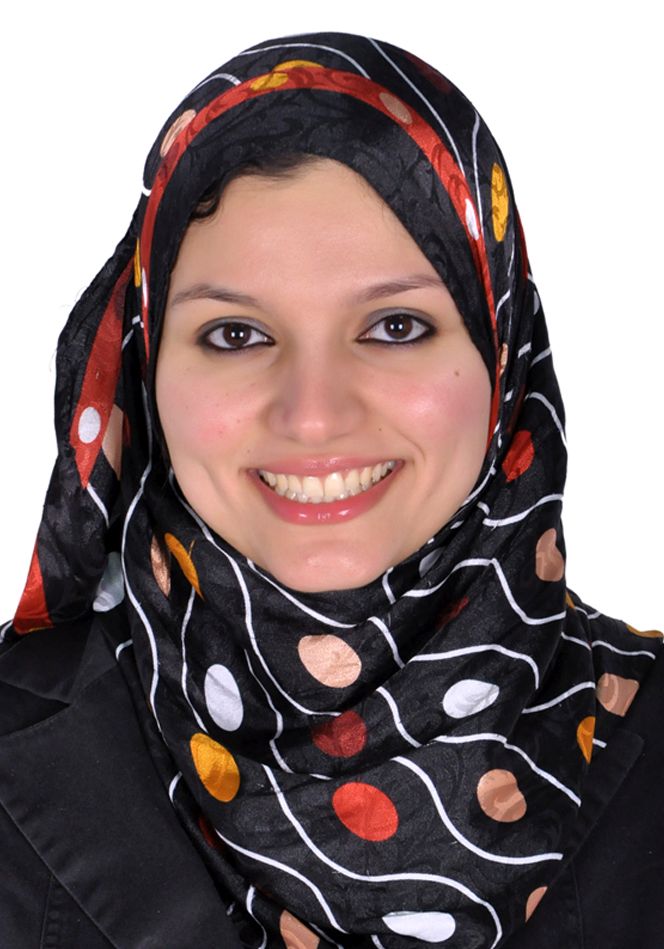}}]{Rania Morsi}
Rania Morsi (S'07) received the B.Sc. degree with highest honors from the faculty of Information Engineering and Technology (IET) at the German University in Cairo (GUC) in 2009. In 2010, she received the M.Sc. degree from the international Masters program ``Communications Technology" at Ulm University, Germany. From February 2011 to July 2012, she was a teaching and research assistant at the GUC. Since August 2012, she is working toward the Ph.D. degree at the Institute of Digital Communications at the Friedrich-Alexander-University (FAU), Erlangen, Germany. Her research interests fall into the broad area of wireless communications including the field of energy harvesting and wireless power transfer. 

Rania received a scholarship for her bachelor studies at the GUC. Her Master studies were sponsored by the German Academic Exchange Service (DAAD). In 2009, she ranked first in the B.Sc. degree over the faculty of IET in the GUC. In 2010, she received the second best student award at the Master program in Ulm University. Rania won the third best tutorial evaluation twice in FAU in 2012 and 2013 for teaching the \quotes{Digital Communications} master course.
\end{IEEEbiography}

\begin{IEEEbiography}[{\includegraphics[width=1in,height=1.5in,clip,keepaspectratio]{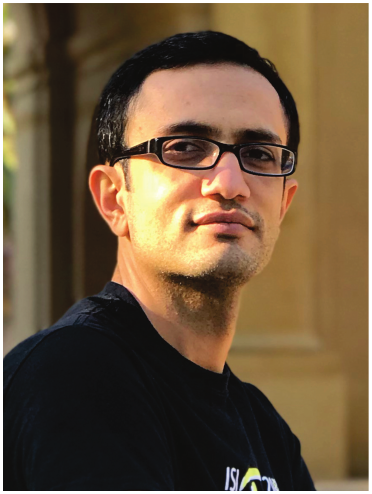}}]{Vahid Jamali} (S'12) received the B.S. and M.S. degrees (honors) in electrical engineering from the K. N. Toosi University of Technology, Tehran, Iran, in 2010 and 2012, respectively, and the Ph.D. degree (with distinctions) from the Friedrich-Alexander-University (FAU) of Erlangen-Nurnberg, Erlangen, Germany, in 2019. In 2017, he was a Visiting Research Scholar with Stanford University, CA, USA. He is currently a Postdoctoral Fellow with the Institute for Digital Communication, FAU. His research interests include wireless and molecular communications, Bayesian inference and learning, and multiuser information theory.

Dr. Jamali received several awards, including the Exemplary Reviewer Certificates \textsc{IEEE Communications Letters} in 2014 and the \textsc{IEEE Transactions on Communications} in 2017 and 2018, the Best Paper Award from the IEEE International Conference on Communications in 2016, the Doctoral Scholarship from the German Academic Exchange Service (DAAD) in 2017, and the Goldener Igel Publication Award from the Telecommunications Laboratory (LNT), FAU, in 2018. He has served as a member of the Technical Program Committee for several IEEE conferences and is currently an Associate Editor of \textsc{IEEE Communications Letters} and \textsc{IEEE Open Journal of Communications Society}.
\end{IEEEbiography}

\begin{IEEEbiography}[{\includegraphics[width=1.1in,height=1.4in,keepaspectratio]{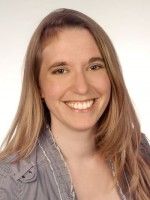}}]{Amelie Hagelauer} Amelie Hagelauer (S'08-M'10-SM'18) received the Dipl.-Ing. degree in mechatronics and the Dr.-Ing. degree in electrical engineering from the Friedrich-Alexander-University Erlangen-Nuremberg, Germany in 2007 and 2013, respectively. She has been a professor at Universitat Bayreuth since Aug 2019.  Prior to that, she joined the Institute for Electronics Engineering in November 2007, where she was working on thin ﬁlm BAW ﬁlters towards her PhD. Since 2013 she is focusing on SAW/BAW and RF MEMS components, as well as integrated circuits for frontends up to 180 GHz. Dr. Hagelauer has been the Chair of MTT-2 Microwave Acoustics from 2015 - 2017. She is continuously contributing to the development of RF Acoustics community by organizing workshops and student design competitions. She has been acting as Associate Editor of the IEEE MTT Transactions, as Guest Editor for a special issue of the IEEE MTT Transactions on the topic \quotes{RF Frontends for Mobile Radio} as well as for a special issue in the MDPI Journal Sensors on the topic \quotes{Surface Acoustic Wave and Bulk Acoustic Wave Sensors}.
\end{IEEEbiography}
\vfill
\begin{IEEEbiography}[{\includegraphics[width=1.1in,height=1.4in,keepaspectratio]{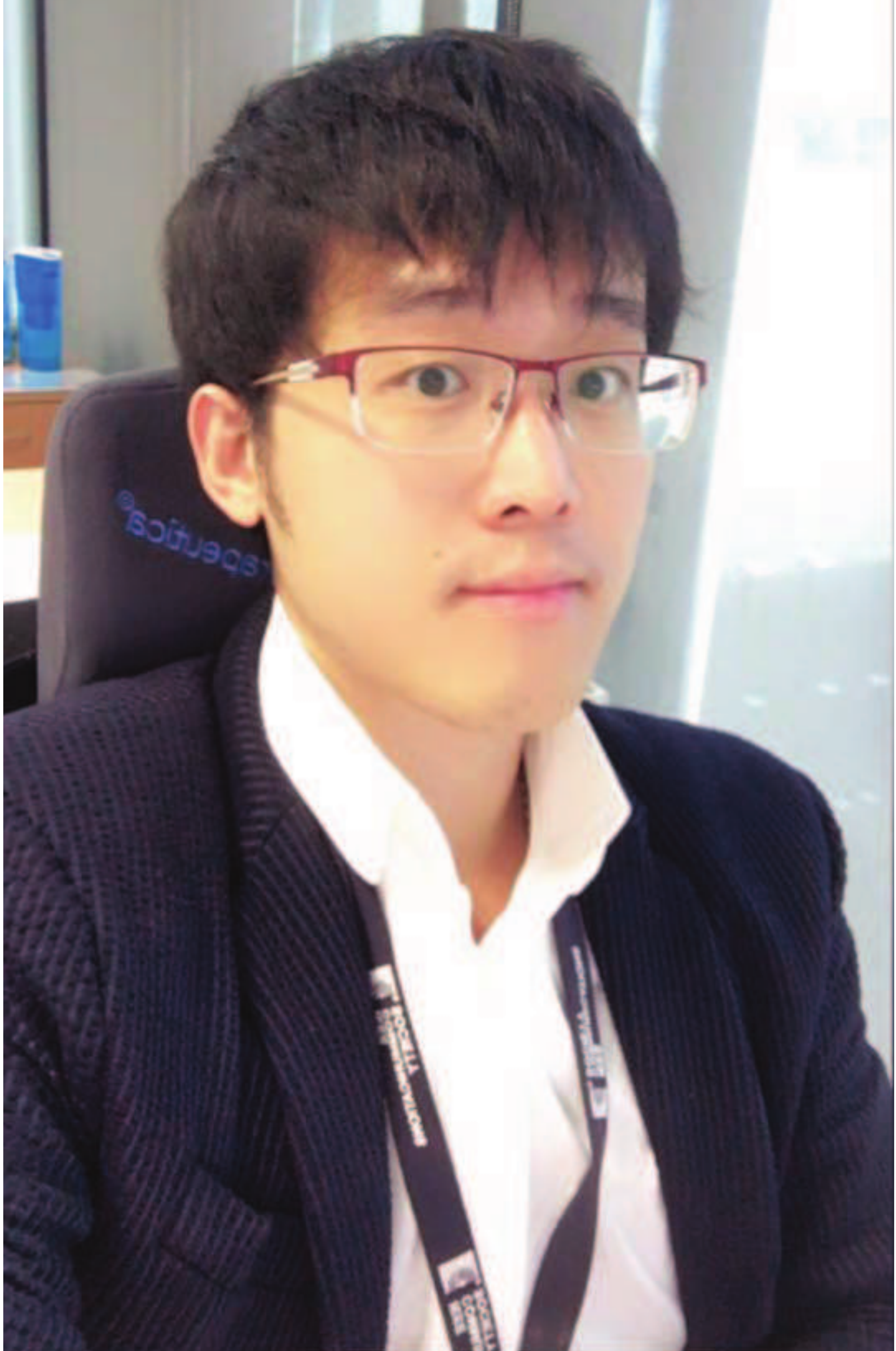}}]{Derrick
Wing Kwan Ng} (S'06-M'12-SM'17) received the bachelor degree with first-class honors and the Master of Philosophy (M.Phil.) degree in electronic engineering from the Hong Kong University of Science and Technology (HKUST) in 2006 and 2008, respectively. He received his Ph.D. degree from the University of British Columbia (UBC) in 2012. He was a senior postdoctoral fellow at the Institute for Digital Communications, Friedrich-Alexander-University Erlangen-N\"urnberg (FAU), Germany. He is now working as a Senior Lecturer and an ARC DECRA Research Fellow at the University of New South Wales, Sydney, Australia.  His research interests include convex and non-convex optimization, physical layer security, wireless information and power transfer, and green (energy-efficient) wireless communications.

Dr. Ng received the Best Paper Awards at the IEEE TCGCC Best Journal Paper Award 2018, INISCOM 2018, IEEE International Conference on Communications (ICC) 2018,  IEEE International Conference on Computing, Networking and Communications (ICNC) 2016,  IEEE Wireless Communications and Networking Conference (WCNC) 2012, the IEEE Global Telecommunication Conference (Globecom) 2011, and the IEEE Third International Conference on Communications and Networking in China 2008.  He is now serving as an area editor for IEEE Open Journal of the Communications Society, an editor for IEEE Transactions on Wireless Communications, IEEE TGCN, IEEE Open Journal of Vehicular Technology, and a guest editor of IEEE JSAC Multiple Antenna Technologies for Beyond 5G, IEEE JSAC Massive Access for 5G and beyond. Besides, he has been serving as an editorial assistant to the Editor-in-Chief of the IEEE Transactions on Communications since Jan. 2012. In addition, he is listed as a Highly Cited Researcher by Clarivate Analytics in 2018 and 2019.
\end{IEEEbiography}

\begin{IEEEbiography}[{\includegraphics[width=1.4in,height=1.3in,trim=0in 1.5in 0in 0.5in,clip,keepaspectratio]{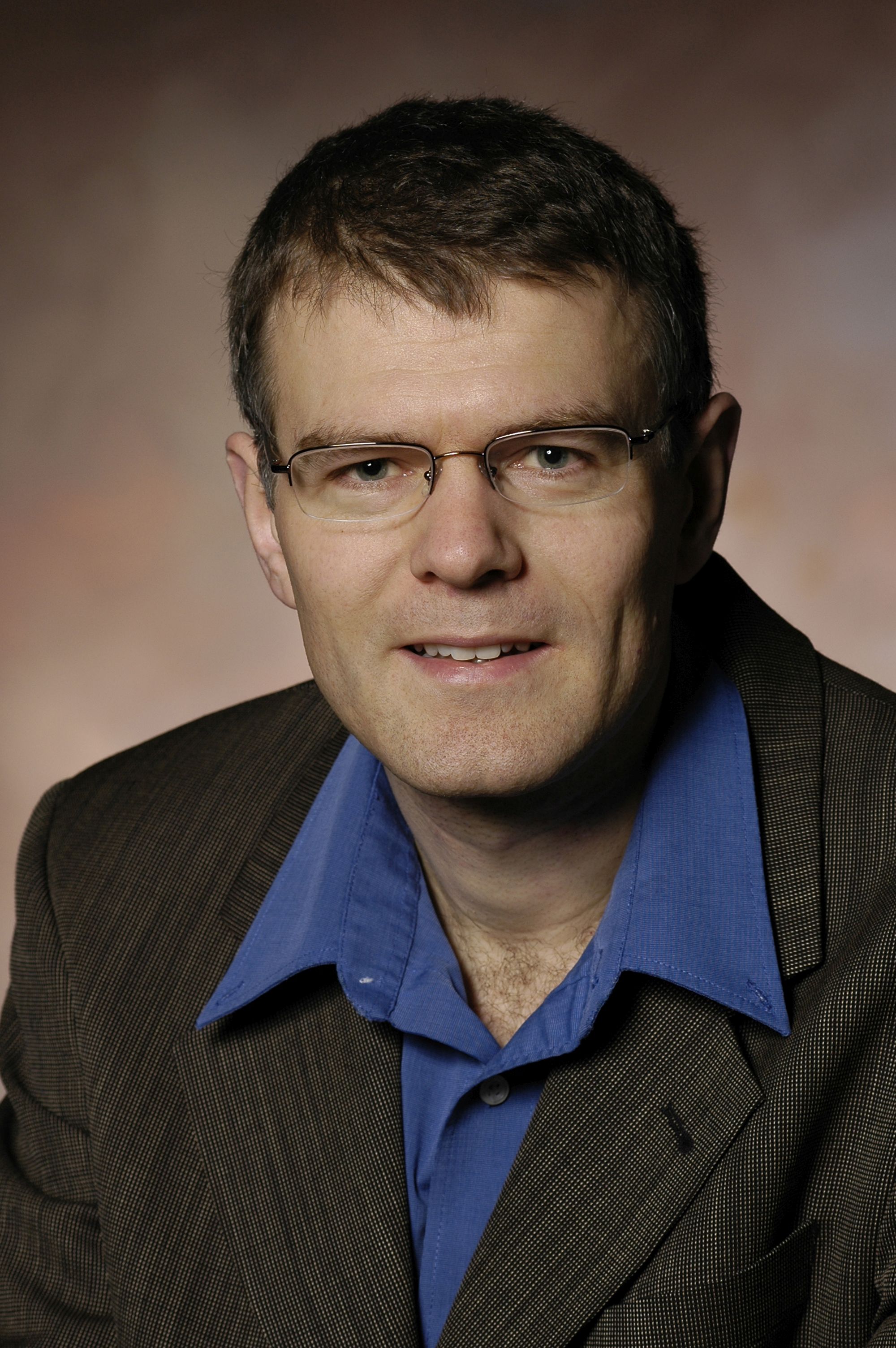}}]{Robert Schober}
Robert Schober (S'98, M'01, SM'08, F'10) was born in Neuendettelsau, Germany, in 1971. He received the Diplom (Univ.) and the Ph.D. degrees in electrical engineering from Friedrich Alexander University (FAU), Erlangen, Germany in 1997 and 2000, respectively. From May 2001 to April 2002 he was a Postdoctoral Fellow at the University of Toronto, Canada, sponsored by the German Academic Exchange Service (DAAD). From May 2002 to December 2011 he was a Professor and Canada Research Chair at the University of British Columbia (UBC), Vancouver, Canada. Since January 2012 he is an Alexander von Humboldt Professor and the Chair for Digital Communication at FAU. His research interests fall into the broad areas of Communication Theory, Wireless Communications, and Statistical Signal Processing.

Dr. Schober received several awards for his work including the 2002 Heinz Maier–Leibnitz Award of the German Science Foundation (DFG), the 2004 Innovations Award of the Vodafone Foundation for Research in Mobile Communications, the 2006 UBC Killam Research Prize, the 2007 Wilhelm Friedrich Bessel Research Award of the Alexander von Humboldt Foundation, the 2008 Charles McDowell Award for Excellence in Research from UBC, a 2011 Alexander von Humboldt Professorship, and a 2012 NSERC E.W.R. Steacie Fellowship. In addition, he received several best paper awards for his research. Dr. Schober is a Fellow of the Canadian Academy of Engineering and a Fellow of the Engineering Institute of Canada. From 2012 to 2015 he served as Editor-in-Chief of the IEEE Transactions on Communications. Currently, he serves as the Chair of the Steering Committee of the IEEE Transactions on Molecular, Biological and Multiscale Communication, on the Editorial Board of the Proceedings of the IEEE, and on the Board of Governors of the IEEE Communication Society (ComSoc). He is also an IEEE ComSoc  Distinguished Lecturer.
\end{IEEEbiography}
\vfill

%% file: paper.bbl
\begin{thebibliography}{10}
\providecommand{\url}[1]{#1}
\csname url@samestyle\endcsname
\providecommand{\newblock}{\relax}
\providecommand{\bibinfo}[2]{#2}
\providecommand{\BIBentrySTDinterwordspacing}{\spaceskip=0pt\relax}
\providecommand{\BIBentryALTinterwordstretchfactor}{4}
\providecommand{\BIBentryALTinterwordspacing}{\spaceskip=\fontdimen2\font plus
\BIBentryALTinterwordstretchfactor\fontdimen3\font minus
  \fontdimen4\font\relax}
\providecommand{\BIBforeignlanguage}[2]{{%
\expandafter\ifx\csname l@#1\endcsname\relax
\typeout{** WARNING: IEEEtran.bst: No hyphenation pattern has been}%
\typeout{** loaded for the language `#1'. Using the pattern for}%
\typeout{** the default language instead.}%
\else
\language=\csname l@#1\endcsname
\fi
#2}}
\providecommand{\BIBdecl}{\relax}
\BIBdecl

\bibitem{Morsi_ICC2018}
R.~Morsi, V.~Jamali, D.~W.~K. Ng, and R.~Schober, ``{On the Capacity of SWIPT
  Systems with a Nonlinear Energy Harvesting Circuit},'' in \emph{Proc. IEEE
  Intern. Commun. Conf.}, May 2018, pp. 1--7.

\bibitem{Varshney2008}
L.~R. Varshney, ``{Transporting Information and Energy Simultaneously},'' in
  \emph{Proc. IEEE Intern. Sympos. on Inf. Theory}, Jul. 2008, pp. 1612--1616.

\bibitem{Shannon_meets_tesla_Grover2010}
P.~Grover and A.~Sahai, ``{Shannon Meets Tesla: Wireless Information and Power
  Transfer},'' in \emph{Proc. IEEE Intern. Sympos. on Inf. Theory}, Jun. 2010,
  pp. 2363--2367.

\bibitem{Rate_energy_MIMO_RuiZhang2013}
R.~Zhang and C.~K. Ho, ``{MIMO Broadcasting for Simultaneous Wireless
  Information and Power Transfer},'' \emph{IEEE Trans. Wireless Commun.},
  vol.~12, no.~5, pp. 1989--2001, May 2013.

\bibitem{Schober_SWPT_review_2015}
Z.~Ding, C.~Zhong, D.~W.~K. Ng, M.~Peng, H.~A. Suraweera, R.~Schober, and H.~V.
  Poor, ``{Application of Smart Antenna Technologies in Simultaneous Wireless
  Information and Power Transfer},'' \emph{IEEE Commun. Mag.}, vol.~53, no.~4,
  pp. 86--93, Apr. 2015.

\bibitem{Kwan_Robert_Book}
D.~Ng, T.~Duong, C.~Zhong, and R.~Schober, \emph{{Wireless Information and
  Power Transfer: Theory and Practice}}, ser. Wiley - IEEE.\hskip 1em plus
  0.5em minus 0.4em\relax Wiley, 2019.

\bibitem{SWIPT_relaying_survey_2018}
Z.~{Wei}, X.~{Zhu}, S.~{Sun}, Y.~{Jiang}, A.~{Al-Tahmeesschi}, and M.~{Yue},
  ``{Research Issues, Challenges, and Opportunities of Wireless Power
  Transfer-Aided Full-Duplex Relay Systems},'' \emph{IEEE Access}, vol.~6, pp.
  8870--8881, 2018.

\bibitem{RE_tradeoff_nonlinear_EH_Kim2018}
J.~{Kang}, I.~{Kim}, and D.~I. {Kim}, ``{Wireless Information and Power
  Transfer: Rate-Energy Tradeoff for Nonlinear Energy Harvesting},'' \emph{IEEE
  Trans. Wireless Commun.}, vol.~17, no.~3, pp. 1966--1981, Mar. 2018.

\bibitem{SWIPT_Clerckx_OFDM_Multisine2016}
B.~Clerckx, ``{Wireless Information and Power Transfer: Nonlinearity, Waveform
  Design, and Rate-Energy Tradeoff},'' \emph{IEEE Trans. Signal Process.},
  vol.~66, no.~4, pp. 847--862, Feb. 2018.

\bibitem{SWIPT_Clerckx_Single_Carrier2017}
M.~Varasteh, B.~Rassouli, and B.~Clerckx, ``{Wireless Information and Power
  Transfer over an AWGN Channel: Nonlinearity and Asymmetric Gaussian
  Signaling},'' in \emph{Proc. IEEE Inf. Theory Workshop}, Nov. 2017, pp.
  181--185.

\bibitem{JSAC_Schober_Clreckx_Poor2019}
B.~Clerckx, R.~Zhang, R.~Schober, D.~W.~K. Ng, D.~I. Kim, and H.~V. Poor,
  ``{Fundamentals of Wireless Information and Power Transfer: From RF Energy
  Harvester Models to Signal and System Designs},'' \emph{IEEE J. Select. Areas
  Commun.}, vol.~37, no.~1, pp. 4--33, Jan. 2019.

\bibitem{RF_EH_networks_survey_2015}
X.~Lu, P.~Wang, D.~Niyato, D.~I. Kim, and Z.~Han, ``{Wireless Networks with RF
  Energy Harvesting: A Contemporary Survey},'' \emph{IEEE Commun. Surveys
  Tuts.}, vol.~17, no.~2, pp. 757--789, Second quarter 2015.

\bibitem{Letter_non_linear}
E.~Boshkovska, D.~Ng, N.~Zlatanov, and R.~Schober, ``{Practical Non-linear
  Energy Harvesting Model and Resource Allocation for SWIPT Systems},''
  \emph{IEEE Commun. Lett.}, vol.~19, no.~12, pp. 2082--2085, Dec. 2015.

\bibitem{Theoretical_Analysis_Rectifiers_2014}
J.~Guo, H.~Zhang, and X.~Zhu, ``{Theoretical Analysis of RF-DC Conversion
  Efficiency for Class-F Rectifiers},'' \emph{IEEE Trans. on Microwave Theory
  and Techniques}, vol.~62, no.~4, pp. 977--985, Apr. 2014.

\bibitem{Georgiadis_WPT_book_2016}
S.~Nikoletseas, Y.~Yang, and A.~Georgiadis, Eds., \emph{Wireless Power Transfer
  Algorithms, Technologies and Applications in Ad Hoc Communication
  Networks}.\hskip 1em plus 0.5em minus 0.4em\relax Springer International
  Publishing, 2016.

\bibitem{Optimum_behaviour_Georgiadis2013}
A.~Boaventura, A.~Collado, N.~B. Carvalho, and A.~Georgiadis, ``{Optimum
  Behavior: Wireless Power Transmission System Design Through Behavioral Models
  and Efficient Synthesis Techniques},'' \emph{IEEE Microwave Mag.}, vol.~14,
  no.~2, pp. 26--35, Mar. 2013.

\bibitem{Waveform_optimization_SPAWC_Rui_Zhang_2017}
M.~R.~V. Moghadam, Y.~Zeng, and R.~Zhang, ``{Waveform Optimization for
  Radio-Frequency Wireless Power Transfer},'' in \emph{Proc. IEEE Intern.
  Workshop on Signal Process. Adv. in Wireless Commun. (SPAWC)}, Jul. 2017, pp.
  1--6.

\bibitem{Waveform_design_WPT_Clerckx_2016}
B.~Clerckx and E.~Bayguzina, ``{Waveform Design for Wireless Power Transfer},''
  \emph{IEEE Trans. Signal Process.}, vol.~64, no.~23, pp. 6313--6328, Dec.
  2016.

\bibitem{Polozec1994}
R.~G. Harrison and X.~L. Polozec, ``{Nonsquarelaw Behavior of Diode Detectors
  Analyzed By The Ritz-Galerkin Method},'' \emph{IEEE Trans. on Microwave
  Theory and Techniques}, vol.~42, no.~5, pp. 840--846, May 1994.

\bibitem{Shockley1949}
W.~{Shockley}, ``{The Theory of P-N Junctions in Semiconductors and P-N
  Junction Transistors},'' \emph{Bell System Tech. J.}, vol.~28, no.~3, pp.
  435--489, Jul. 1949.

\bibitem{Shannon}
C.~E. Shannon, ``{A Mathematical Theory of Communication},'' \emph{Bell System
  Tech. J.}, vol.~27, no.~3, pp. 379--423, 1948.

\bibitem{SMITH19712}
J.~G. Smith, ``{The Information Capacity of Amplitude- and Variance-Constrained
  Scalar Gaussian Channels},'' \emph{Inf. and Control}, vol.~18, no.~3, pp.
  203--219, 1971.

\bibitem{QAGC_Shamai_Bar_David_1995}
S.~Shamai and I.~Bar-David, ``{The Capacity of Average and Peak-Power-Limited
  Quadrature Gaussian Channels},'' \emph{IEEE Trans. Inf. Theory}, vol.~41,
  no.~4, pp. 1060--1071, Jul. 1995.

\bibitem{Bruno_Capacity}
M.~{Varasteh}, B.~{Rassouli}, and B.~{Clerckx}, ``{On Capacity-Achieving
  Distributions for Complex AWGN Channels Under Nonlinear Power Constraints and
  their Applications to SWIPT},'' \emph{ArXiv e-prints}, Dec. 2017,
  arXiv:1712.01226.

\bibitem{capacity_Rayleigh_fading_Abou_Faycal2001}
I.~C. Abou-Faycal, M.~D. Trott, and S.~Shamai, ``{The Capacity of Discrete-Time
  Memoryless Rayleigh-Fading Channels},'' \emph{IEEE Trans. Inf. Theory},
  vol.~47, no.~4, pp. 1290--1301, May 2001.

\bibitem{Hermite_bases_Abou_Faycal_2012}
J.~J. Fahs and I.~C. Abou-Faycal, ``{Using Hermite Bases in Studying
  Capacity-Achieving Distributions Over AWGN Channels},'' \emph{IEEE Trans.
  Inf. Theory}, vol.~58, no.~8, pp. 5302--5322, Aug. 2012.

\bibitem{Nikola_FD}
N.~Zlatanov, E.~Sippel, V.~Jamali, and R.~Schober, ``{Capacity of the Gaussian
  Two-Hop Full-Duplex Relay Channel With Residual Self-Interference},''
  \emph{IEEE Trans. Commun.}, vol.~65, no.~3, pp. 1005--1021, Mar. 2017.

\bibitem{Polozec2015_Input_Impedance}
{Le Polozec, Xavier}, ``{Input Impedance of Series Schottky Diode Detector at
  Low and High Power},'' May 2015, {DOI:10.13140/RG.2.1.4530.9600}.

\bibitem{Corless1996}
R.~M. Corless, G.~H. Gonnet, D.~E.~G. Hare, D.~J. Jeffrey, and D.~E. Knuth,
  ``{On the LambertW Function},'' \emph{Advances in Computational Mathematics},
  vol.~5, no.~1, pp. 329--359, Dec. 1996.

\bibitem{WPT_LambertW_2019}
K.~{Kim}, H.~{Lee}, and J.~{Lee}, ``{Waveform Design for Fair Wireless Power
  Transfer With Multiple Energy Harvesting Devices},'' \emph{IEEE J. Select.
  Areas Commun.}, vol.~37, no.~1, pp. 34--47, Jan. 2019.

\bibitem{EH_solid_state_2008}
T.~Le, K.~Mayaram, and T.~Fiez, ``{Efficient Far-Field Radio Frequency Energy
  Harvesting for Passively Powered Sensor Networks},'' \emph{IEEE J. of
  Solid-State Circuits}, vol.~43, no.~5, pp. 1287--1302, May 2008.

\bibitem{Skyworks_SMS7630}
{Skyworks Solutions, Inc.}, \emph{SMS7630-061: Surface Mount, 0201 Zero Bias
  Silicon Schottky Detector Diode}, May 2015.

\bibitem{ADS}
{The Keysight Technologies, Inc.}, ``{Electronic Design Automation (EDA)
  Software, Advanced Design System (ADS), Version 2017}.''

\bibitem{CVX}
M.~Grant and S.~Boyd, ``{{CVX}: Matlab Software for Disciplined Convex
  Programming, Version 2.0 Beta},'' Sep. 2012.

\bibitem{Lagrange_Inversion_LambertW}
D.~J. Jeffrey, G.~A. Kalugin, and N.~Murdoch, ``{Lagrange Inversion and Lambert
  W},'' in \emph{Intern. Sympos. on Symbolic and Numeric Algorithms for
  Scientific Comput. (SYNASC)}, Sep. 2015, pp. 42--46.

\end{thebibliography}
